\documentclass[12pt,english]{article}

\usepackage[T1]{fontenc}
\usepackage[utf8]{inputenc}   
\usepackage{microtype}
\usepackage{comment}
\usepackage{tikz}
\usepackage{float}

\usepackage{graphicx} 
\usepackage{amsmath}
\usepackage{amsthm}
\usepackage{amssymb}
\usepackage{newtxtext}
\usepackage[cmintegrals,varg]{newtxmath}
\theoremstyle{plain} 
\newtheorem{theorem}{Theorem}
\newtheorem{lemma}{Lemma}
\newtheorem{corollary}{Corollary}

\theoremstyle{remark} 
\newtheorem{remark}{Remark}

\usepackage[authoryear]{natbib}
\usepackage{tabularx}
\usepackage{siunitx}
\usepackage{booktabs}
\newcolumntype{Y}{>{\centering\arraybackslash}X}

\usepackage{geometry}
\usepackage{setspace}
\makeatother
\title{The Generalized Fisher Transformation: Finite-Sample Properties and Inference}

\author{Ilya Archakov$^{\ddagger}$ and Peter Reinhard Hansen$^{\mathsection}$
\\ \small $^{\ddagger}$Department of Economics, York University
\\ \small $^{\mathsection}$Department of Economics, University of North Carolina at Chapel Hill}
\date{\today}

\begin{document}

\maketitle

\begin{abstract}
We study the finite-sample behavior of the Generalized Fisher Transformation (GFT), the parametrization of a correlation matrix $C$ by $\gamma(C)=\operatorname{vecl}\log C$. The GFT coordinates extend Fisher's transformation to dimension $n>2$: for elliptical data their finite-sample distributions are close to Gaussian. More strikingly, the coordinates are nearly uncorrelated and their covariance is largely invariant to $C$. This approximate orthogonality and invariance make GFT-based inference far better behaved in finite samples than inference based on sample correlations or element-wise Fisher transformed correlations, yielding estimation errors that are approximately Gaussian, weakly dependent, and nearly pivotal.
\end{abstract}



\section{Introduction}

Correlation matrices are central to multivariate analysis in economics and finance, yet statistical inference for these matrices remains notoriously difficult. The elements of the standard sample correlation matrix, $\hat{C}$, are bounded in $[-1,1]$ and exhibit complex dependencies, even when the underlying variables are independent. \citet{Fisher:1915} resolved these issues for the bivariate case ($n=2$) with his celebrated $z$-transformation, which stabilizes variance and reduces skewness under appropriate conditions. \citet{Hotelling:1953} emphasized that this simultaneous variance stabilization and approximate normalization is a special feature of Fisher's transformation, rather than a generic property of transformations of statistics. However, a multivariate generalization that preserves these desirable properties for correlation matrices of dimension $n>2$ has remained elusive.

In this paper, we analyze the Generalized Fisher Transformation (GFT), $\gamma(C)=\operatorname{vecl}(\log C)$, which was introduced by \citet{ArchakovHansen:Correlation}. The GFT maps the manifold of positive definite correlation matrices to the Euclidean space $\mathbb{R}^{d}$, where $d=n(n-1)/2$. This is achieved by taking the matrix logarithm of $C$ and vectorizing the below-diagonal elements, as illustrated in the following example for $n=3$:
$$C=\left[ \begin{array}{ccc} 1.00 & 0.90 & 0.70 \\ 0.90 & 1.00 & 0.40 \\ 0.70 & 0.40 & 1.00 \end{array} \right], \quad 
\log C = \left[ \begin{array}{rrr} -1.66 & 1.81 & 1.14 \\ 1.81 & -1.03 & -0.15 \\ 1.14 & -0.15 & -0.43 \end{array} \right],
\quad \gamma = \operatorname{vecl}\log C = \left[ \begin{array}{r} 1.81 \\ 1.14 \\ -0.15  
\end{array} \right].$$
This is a generalization of the Fisher transformation,
$\phi=\tfrac{1}{2}\log\tfrac{1+\rho}{1-\rho}$, to correlation
matrices, in the sense that $\gamma(C)$ is identical to the Fisher
transformation for a $2\times2$ correlation matrix, specifically
$$
C=\left[\begin{array}{cc}
1 & \rho\\
\rho & 1
\end{array}\right]\quad\Rightarrow\quad\log C=\left[\begin{array}{cc}
\tfrac{1}{2}\log(1-\rho^{2}) & \tfrac{1}{2}\log\tfrac{1+\rho}{1-\rho}\\
\tfrac{1}{2}\log\tfrac{1+\rho}{1-\rho} & \tfrac{1}{2}\log(1-\rho^{2})
\end{array}\right]\quad\text{for }|\rho|<1.
$$
For $n \geq 3$, the elements of $\gamma$ do not coincide with the element-wise Fisher transformed correlations. Consequently, $\hat\gamma=\gamma(\hat{C})$ has distinct statistical properties.

\begin{figure}[htb]
\begin{centering}
\includegraphics[width=0.8\textwidth]{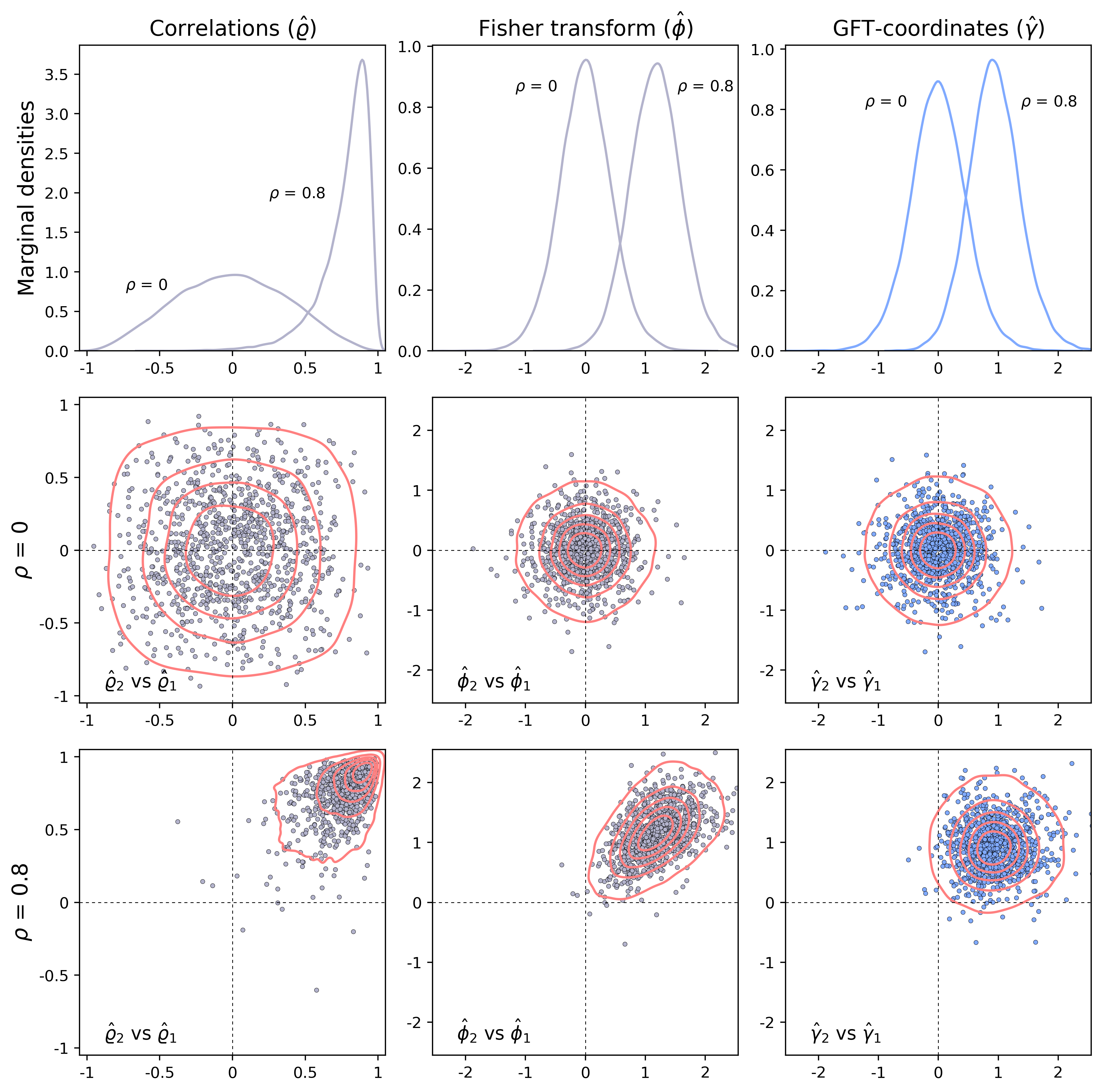}
\par\end{centering}
\caption{\small Bivariate distributions of pairs of correlations, $(\hat{\varrho}_{1},\hat{\varrho}_{2})$,
Fisher transformed correlations, $(\hat{\phi}_{1},\hat{\phi}_{2})$,
and Generalized Fisher transformed correlations, $(\hat{\gamma}_{1},\hat{\gamma}_{2})$,
computed from $X_{t}\sim iid N_{3}(0,\Sigma)$, where $\Sigma$ is an
equicorrelation matrix with off-diagonal elements equal to either
$\rho=0$ or $\rho=0.8$ and $T=8$ observations.\label{fig:FisherDesign}}
\end{figure}

The intuition behind the GFT is illustrated in Figure~\ref{fig:FisherDesign}. To generate this figure, we replicated the design used by \citet{Fisher:1921} to motivate his original transformation. The design uses just $T=8$ observations and an equicorrelation structure with either $\rho=0$ or $\rho=0.8$. The first two panels in the first row reproduce \citet[figures 1 and 2]{Fisher:1921}. 
The left panels display the marginal and joint distributions of two standard sample correlation coefficients, $(\hat{\varrho}_{1},\hat{\varrho}_{2})$, which correspond to $(\hat{\rho}_{12},\hat{\rho}_{13})$ in this three-dimensional design. For $\rho=0.8$, the distribution is heavily skewed and exhibits strong dependence, as seen in the lower-left panel. The middle column displays the corresponding Fisher transformed correlations, $(\hat{\phi}_{1}, \hat{\phi}_{2})$. The marginal distributions are close to normal with unit variance, but the strong dependence between the two elements persists for $\rho=0.8$. The right panels display the corresponding GFT coordinates, $(\hat{\gamma}_{1},\hat{\gamma}_{2})$. Not only are the marginal distributions approximately Gaussian, the contours of the bivariate distribution are remarkably spherical for both $\rho=0$ and $\rho=0.8$. Thus, the GFT maps the complex geometry of the correlation cone to coordinates in which the estimates are approximately uncorrelated.

While \citet{ArchakovHansen:Correlation} established the theoretical bijection and asymptotic results for $\gamma(C)$, the finite-sample statistical properties have largely remained unexplored for general correlation structures and non-Gaussian data. This paper documents three main finite-sample properties of the GFT and shows how they translate into improved inference for correlation matrices.

First, we show that the finite-sample distributions of the elements of $\hat{\gamma}$ are well approximated by their Gaussian asymptotic distributions for elliptically distributed data. This mirrors the behavior of the univariate Fisher transformation and extends an important scalar property to correlation matrices of higher dimension. However, as with the Fisher transformed correlation, the Gaussian approximation can become unreliable in small samples when the data are strongly skewed or otherwise depart from elliptical symmetry.

Second, we document the near orthogonality of the GFT coordinates. The finite-sample covariance matrix $V_{\gamma,T}(C)=\operatorname{var}{\sqrt{T}(\hat{\gamma}-\gamma)}$ is approximately diagonal across a wide range of designs. This contrasts sharply with the strong dependence found in both $\hat\varrho$ and $\hat\phi$, where $\hat\varrho=\varrho(\hat{C})$ and $\hat\phi=\phi(\hat{C})$ denote the vector of sample correlations and the vector of element-wise Fisher transformed sample correlations, respectively.

Third, we document that $V_{\gamma,T}(C)$ is far less sensitive to the true correlation matrix, $C$, than the corresponding covariance matrices $V_{\varrho,T}(C)$ and $V_{\phi,T}(C)$. This near invariance has important implications for inference. Since $V_{\gamma}(C)$ is relatively stable across values of $C$, the plug-in covariance estimator $V_{\gamma}(\hat{C})$ is much less affected by estimation error in $\hat{C}$ than the corresponding plug-in estimators for $\hat\varrho$ and $\hat\phi$. As a result, standardized statistics based on $\hat{\gamma}$ have substantially better finite-sample behavior, and Wald tests based on the GFT converge much faster to their nominal size.

The remainder of the paper is organized as follows. Section~\ref{sec:notation} introduces the notation, simulation designs, and empirical data sets. Section~\ref{sec:MarginalJoint} documents the marginal and joint finite-sample properties of $\hat{\varrho}$, $\hat{\phi}$, and $\hat{\gamma}$ using Gaussian simulations, non-Gaussian simulations, and empirical resampling designs. Section~\ref{sec:covstability} gives theoretical results that explain the covariance stability and weak dependence of the GFT coordinates. Section~\ref{sec:inference} studies the implications for inference and shows that GFT-based standardized statistics have substantially better finite-sample behavior. Section~\ref{sec:conclusion} concludes. The Appendix contains proofs, and the Supplement reports additional simulation and empirical results.

\section{Notation, Simulation Designs, and Empirical Data}\label{sec:notation}

We present results for the sample correlation matrix, $\hat{C}$,
associated with the sample covariance matrix, 
$$\hat{\Sigma}=\tfrac{1}{T}\sum_{t=1}^{T}(X_{t}-\bar{X})(X_{t}-\bar{X})^{\prime},
\qquad \bar{X}=\tfrac{1}{T}\sum_{t=1}^{T}X_{t}.$$ 
Let $\mu=\mathbb{E}[X_{t}]$ and $\Sigma=\operatorname{var}(X_{t})$. In our simulations, we can, without loss of generality, take $\mu=0$ and $\Sigma=C$, because $Y=DX+b$
has the same correlation matrix as $X$ for any non-singular
diagonal matrix, $D$, and any vector $b$, and the empirical correlation
matrix, $\hat{C}$, retains this invariance. We present results for
the case where $X_{t}$ has a multivariate normal distribution and
cases where it follows non-Gaussian distributions, including the uniform,
multivariate $t$, and Inverse Gaussian distributions.
We also simulate by drawing from an empirical distribution of daily
industry returns.

As introduced earlier, we let $\varrho=\operatorname{vecl}(C)$ and $\phi=\phi(C)$
denote the vectors of correlations and the corresponding vector of
Fisher transformed correlations, respectively. Similarly, $\hat{\varrho}=\varrho(\hat{C})$,
$\hat{\phi}=\phi(\hat{C})$, and $\hat{\gamma}=\gamma(\hat{C})$ denote
the empirical vectors of correlation measures, where $\hat{C}$ is
the sample correlation matrix. Their asymptotic distributions are,
under suitable regularity conditions, given by
\begin{eqnarray}
\sqrt{T}(\hat{\varrho}-\varrho) & \overset{d}{\rightarrow} & N(0,V_{\varrho}(C)),\nonumber \\
\sqrt{T}(\hat{\phi}-\phi) & \overset{d}{\rightarrow} & N(0,V_{\phi}(C)),\label{eq:LimitDistributions}\\
\sqrt{T}(\hat{\gamma}-\gamma) & \overset{d}{\rightarrow} & N(0,V_{\gamma}(C)),\nonumber
\end{eqnarray}
respectively. General expressions for the asymptotic covariance matrices,
$V_{\varrho}(C)=\operatorname{avar}(\hat{\varrho})$, $V_{\phi}(C)=\operatorname{avar}(\hat{\phi})$,
and $V_{\gamma}(C)=\operatorname{avar}(\hat{\gamma})$, are given in Appendix
\ref{app:additional}, along with simplified expressions for the special case where
$X_{t}$ is iid and normally distributed. We use the correlation matrix,
$C$, as an argument in the expressions for the asymptotic covariance
matrices to make explicit that they depend on the true correlation
matrix. The corresponding finite-sample variance-covariance matrices
are denoted by
$$
V_{\varrho,T}(C)=\operatorname{var}(\sqrt{T}(\hat{\varrho}-\varrho)),\quad V_{\phi,T}(C)=\operatorname{var}(\sqrt{T}(\hat{\phi}-\phi)),\quad V_{\gamma,T}(C)=\operatorname{var}(\sqrt{T}(\hat{\gamma}-\gamma)),
$$
respectively, where $T$ is the sample size. Similarly, we define
the correlation matrices,
$$
R_{\varrho,T}(C)=\operatorname{corr}(\sqrt{T}(\hat{\varrho}-\varrho)),\quad R_{\phi,T}(C)=\operatorname{corr}(\sqrt{T}(\hat{\phi}-\phi)),\quad R_{\gamma,T}(C)=\operatorname{corr}(\sqrt{T}(\hat{\gamma}-\gamma)),
$$
and their asymptotic counterparts are denoted by $R_{\varrho}(C)$,
$R_{\phi}(C)$, and $R_{\gamma}(C)$, respectively.

\subsection{Some Empirical Correlation Matrices}

Before we present results on the finite-sample properties of $\hat{\varrho}$, $\hat{\phi}$, and $\hat{\gamma}$, it will be useful to illustrate the GFT coordinates with three empirical data sets. These provide realistic correlation structures for the simulation designs in Section~\ref{sec:MarginalJoint}.
\begin{figure}[!htbp]
\begin{centering}
\includegraphics[width=0.99\textwidth]{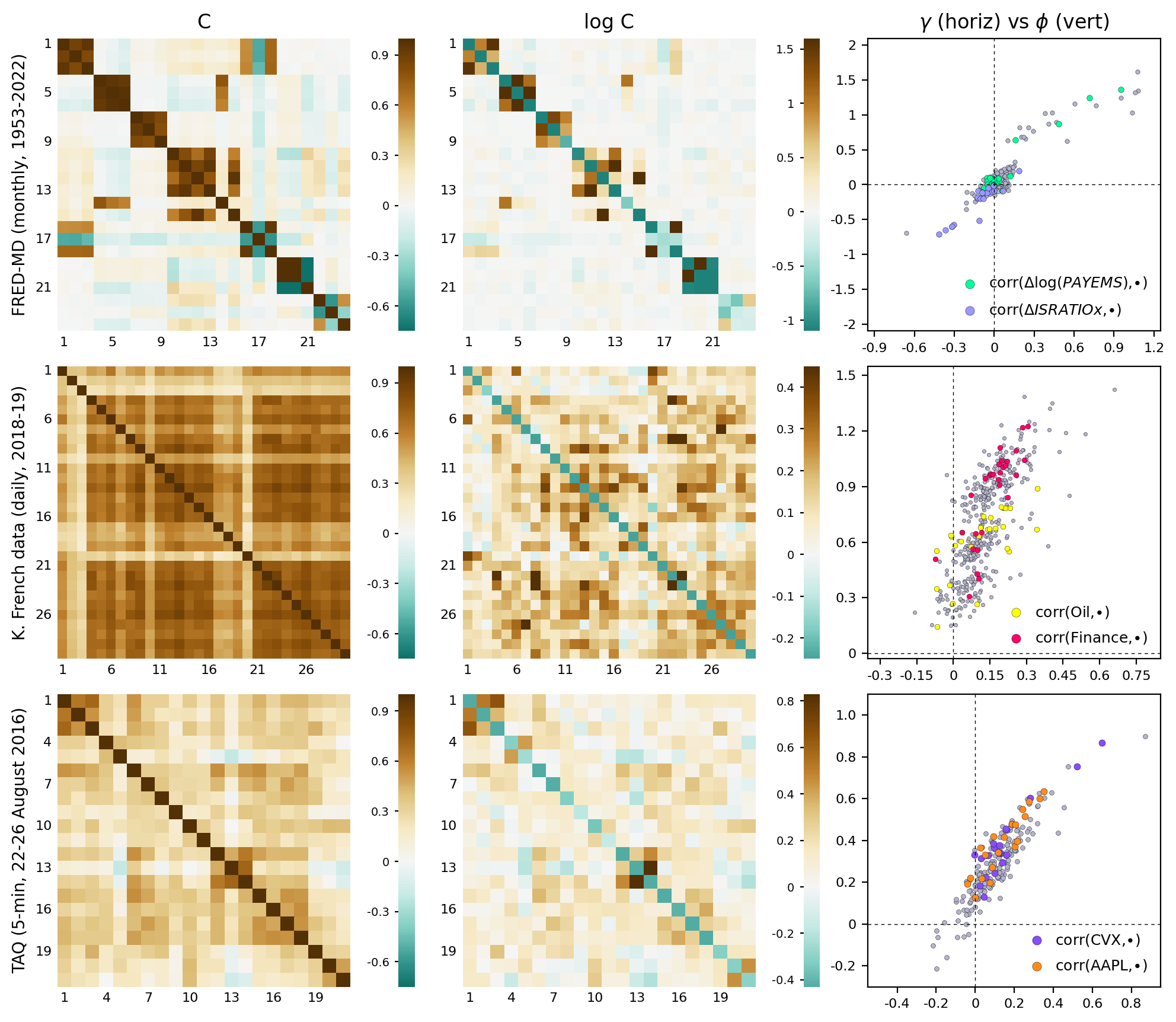}
\par\end{centering}
\caption{{\small Empirical correlation matrices, $\hat{C},$ for three data
sets are shown in the left panels and the corresponding $\log\hat{C}$
shown in the middle panels. The right panels are scatter plots of
$\hat{\phi}_{i}$ against $\hat{\gamma}_{i}$.}\label{Fig:Dataset}}
\end{figure}

The three data sets are: (1) the macroeconomic time-series from \citet{McCrackenNg:2016}
(monthly data from 1953 to 2022); (2) Daily returns for 30 Fama-French
industry portfolios downloaded from Kenneth French's website (January
2, 2018 to December 31, 2019); and (3) Realized correlation matrix
for 21 assets based on five-minute intraday returns from opening hours
9:30 AM to 4:00 PM during the week, August 22-26, 2016.

We present the empirical correlation matrix for these three data sets
in Figure~\ref{Fig:Dataset} (left panels) along with the logarithmically
transformed correlation matrix (middle column). The right panels are
scatter plots of the elements of $\phi(\hat{C})$ plotted against
the corresponding elements of $\gamma(\hat{C})$. These reveal a relationship
between $\hat{\gamma}$ and $\hat{\phi}$, which also implies a relationship
between $\hat{\gamma}$ and $\hat{\varrho}$. Thus, even though the
matrix logarithm is a highly nonlinear mapping, where each element
of $\log\hat{C}$ depends on all sample correlation coefficients in $\hat{C}$, it is evident that $\hat{\gamma}_{i}$ is primarily influenced by the corresponding sample correlation element, $\hat{\varrho}_{i}$. This is not surprising because the Jacobian, $\partial\varrho/\partial\gamma$, is typically dominated by its
diagonal elements.

Detailed information about elements of $\hat{C}$ and $\log\hat{C}$
will be presented in Section~\ref{subsec:EmpiricalBasedResults} for
the industry portfolios along with an extended time series of the
realized correlations. Additional empirical results are presented
in the Supplementary Material.

\section{Marginal and Joint Finite-Sample Distributions}\label{sec:MarginalJoint}

The asymptotic distributional theory for $\hat{\gamma}$, $\hat{\varrho}$, and $\hat{\phi}$ is given in \citet{ArchakovHansen:Correlation}. The goal of this section is to document their finite-sample behavior, with particular emphasis on how the GFT coordinates $\hat{\gamma}$ compare to the conventional coordinates $\hat{\varrho}$ and $\hat{\phi}$.

We first examine variance, skewness, and kurtosis of the marginal distributions, using these moments to assess variance stabilization and the quality of the Gaussian approximation in finite samples. We then study joint behavior through the finite-sample correlation matrices $R_{\varrho,T}(C)$, $R_{\phi,T}(C)$, and $R_{\gamma,T}(C)$ of the corresponding coordinate vectors. The first two are nearly identical, whereas $R_{\gamma,T}(C)$ is typically much closer to the identity matrix.

\subsection{Results for $C$ with Toeplitz Structure}

We first consider finite-sample properties when $\hat{C}$ is computed
from $X_{t} \sim iid N_{n}(0,C)$, $t=1,\ldots,T$, where $C$ has the
following Toeplitz structure,
\begin{equation}
C=\left(\begin{array}{ccccc}
1 & \rho & \rho^{2} & \cdots & \rho^{n-1}\\
\rho & 1 & \rho & \cdots & \rho^{n-2}\\
\rho^{2} & \rho & 1 & \cdots & \rho^{n-3}\\
\vdots & \vdots & \vdots & \ddots & \vdots\\
\rho^{n-1} & \rho^{n-2} & \rho^{n-3} & \cdots & 1
\end{array}\right),\label{eq:CorrelationMatrixToeplitz}
\end{equation}
and we use $\rho=0.9$ and $T=100$ in this section. Thus,
when $n=25$, the correlation coefficients range from about $0.08$ to $0.90$.

\begin{figure}[!htbp]
\begin{centering}
\includegraphics[width=0.9\textwidth]{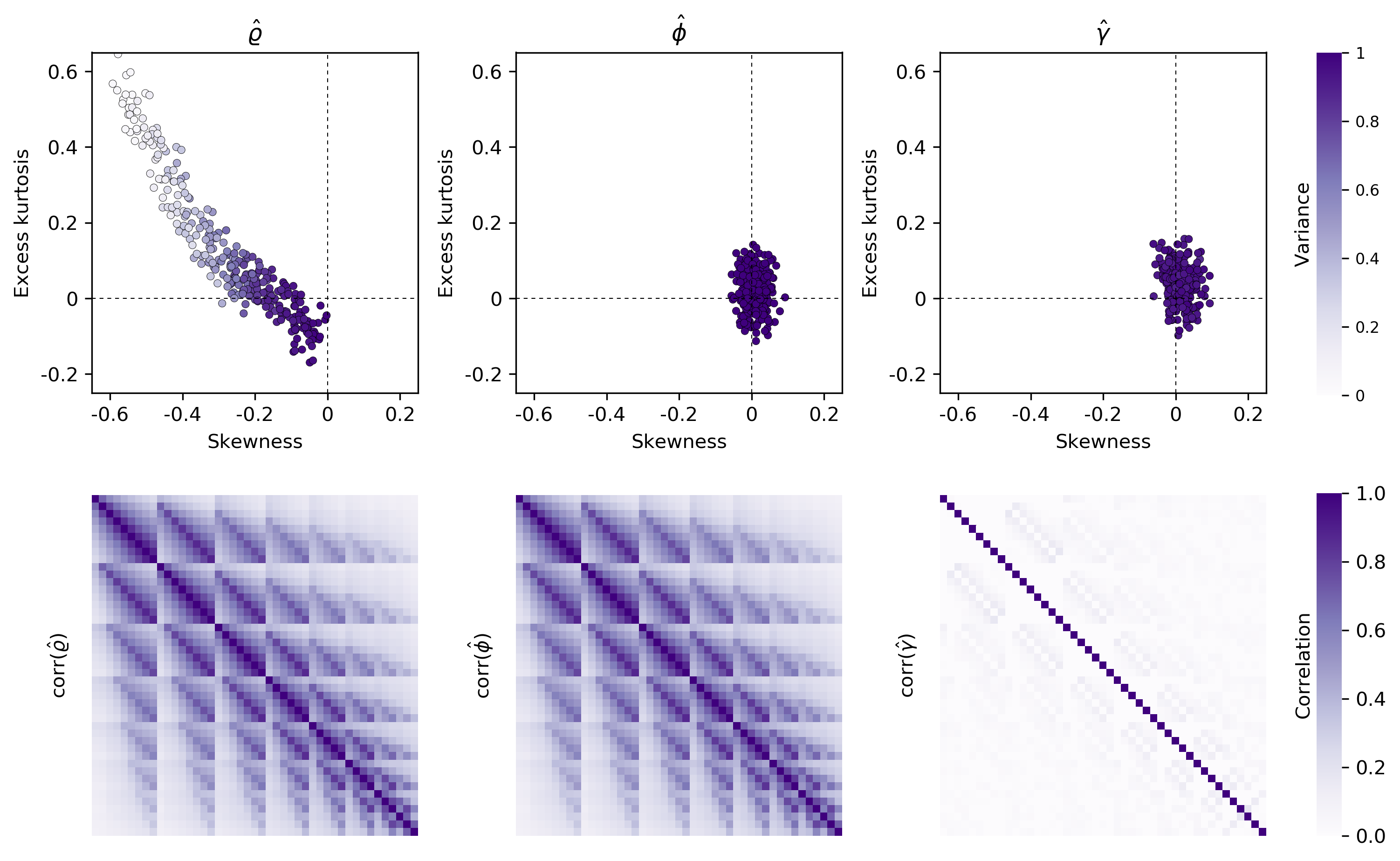}
\par\end{centering}
\caption{{\small Finite-sample skewness and excess kurtosis for elements of $\hat{\varrho}$, $\hat{\phi}$, and $\hat{\gamma}$ based on a Toeplitz structure with $\rho=0.9$ ($n=25$ in upper panels and $n=10$ in lower panels) and sample size $T=100$. }\label{Fig:FiniteDistToeplitz}}
\end{figure}

The upper panels of Figure~\ref{Fig:FiniteDistToeplitz} are based
on dimension $n=25$ and present the skewness and kurtosis for the
marginal distributions of the $d=300$ elements of $\hat{\varrho}$,
$\hat{\phi}$, and $\hat{\gamma}$. We have plotted excess kurtosis
against skewness, where the color intensity of each dot indicates the variance of the corresponding element. Light colors correspond to small, near-zero variances, and dark colors correspond to variances close to one. For reference, the standard normal distribution, $N(0,1)$, has zero skewness, zero excess kurtosis, and unit variance.

The finite-sample distributions of the sample correlations, $\hat{\varrho}$, are very non-Gaussian, as is evident from their skewness and excess kurtosis. Moreover, there is large variation in their variances because these depend on the population correlations. The elements of $\hat{\phi}$ are the Fisher transformed correlations, and each is well approximated by its asymptotic distribution, which is the standard normal distribution, $N(0,1)$. Moreover, the transformation successfully stabilizes their variances, as expected. The corresponding results for $\hat{\gamma}$ are presented in the upper right panel of Figure~\ref{Fig:FiniteDistToeplitz}. Interestingly, the marginal distributions of the elements of $\hat{\gamma}$ are similar to those of $\hat{\phi}$. Skewness and excess kurtosis are both relatively close to zero, and the variances are approximately stabilized. However, the asymptotic marginal distributions for the elements of $\hat{\gamma}$ are not exactly $N(0,1)$, and the diagonal elements of the asymptotic covariance matrix, $V_{\gamma}(C)$, range from 0.865 to 1.005 in this design.\footnote{The majority of the variances are within a narrow band, with 80\% of them falling between 0.911 and 0.966. The asymptotic variances are given by the diagonal elements of $V_{\gamma}(C)$; see \citet[eq.~2]{ArchakovHansen:Correlation}
for the expression.}

A surprising property of $\hat{\gamma}$ emerges from the lower panels of
Figure~\ref{Fig:FiniteDistToeplitz}. These are based on $n=10$ (i.e.,
$d=45$) to ensure legibility of the correlation matrices, and display the
$45\times45$ correlation matrices, $R_{\varrho,T}(C)$, $R_{\phi,T}(C)$, and $R_{\gamma,T}(C)$, using color codes. The correlation matrices for
$\hat{\varrho}$ and $\hat{\phi}$ are virtually identical, which is a consequence of the delta method and the fact that $\hat{\phi}$ is defined by an element-wise transformation of $\hat{\varrho}$. The surprising result in Figure~\ref{Fig:FiniteDistToeplitz} is the weak dependence between the elements of $\hat{\gamma}$. Not only does the mapping $\gamma(C)$ produce coordinates whose finite-sample distributions are well approximated by normal distributions when the underlying data are normally distributed, the elements are also nearly uncorrelated. The latter is an unexpected property of the GFT coordinates. The weak dependence appears to be a robust property of $\hat{\gamma}$, whereas the near-Gaussian finite-sample distributions of the elements of $\hat{\phi}$ and $\hat{\gamma}$ are more fragile, in the sense that they rely on the underlying data being normally distributed. We next examine whether these findings persist for correlation matrices with more general structures and for correlation matrices estimated from non-Gaussian samples.

\subsection{Correlation Matrices with Arbitrary Structures\label{subsec:RandomCorr}}

The Toeplitz structure in (\ref{eq:CorrelationMatrixToeplitz}) defines
a special class of correlation matrices. We now study correlation
matrices with more general, arbitrary structures. We generate random
correlation matrices using the method developed in
\citet{ArchakovHansenLuo-RandomCorr:2024}. This is straightforward,
as it only requires drawing a random vector $\gamma$, followed by evaluating
$C(\gamma)$. Specifically, we first draw $\omega$ uniformly on the interval
$[0,\tfrac{5}{n}\log n]$, and then draw
$\gamma|\omega\sim N_d(\xi\iota,\omega^2I_d)$, where
$\xi=\tfrac{\log(1+n)}{n}$.\footnote{The dependence on $n$ is motivated by the fact that an equicorrelation
matrix with common correlation coefficient, $\rho$, will translate
to a vector, $\gamma$, with identical elements equal to $\tfrac{1}{n}\log(1+n\tfrac{\rho}{1-\rho})$.}
This particular design is chosen because it generates a sufficiently
interesting range of correlation matrices, including near-singular
matrices. Moreover, the design is such that the probability density
is positive for any valid non-singular correlation matrix.\footnote{Any distribution for $\gamma$ on $\mathbb{R}^{d}$, with $d=n(n-1)/2$,
will translate to a distribution on $\mathcal{C}_{n}$, and any distribution
with full support on $\mathbb{R}^{d}$ will induce a distribution
with full support on $\mathcal{C}_{n}$.} Three examples of random correlation matrices generated with this
design are:

\setlength{\arraycolsep}{1.5pt}
\begin{scriptsize}
$$
C_1 =
\left[ \begin{array}{rrrrr} 1.00 & 0.56 & \phantom{{}.{}}0.33 & \phantom{{}.{}}0.52 & 0.10 \\ 0.56 & 1.00 & 0.54 & 0.81 & -0.03 \\ 0.33 & 0.54 & 1.00 & 0.68 & 0.47 \\ 0.52 & 0.81 & 0.68 & 1.00 & 0.39 \\ 0.10 & -0.03 & 0.47 & 0.39 & 1.00 \\ \end{array} \right],
\quad
C_2 =
\left[ \begin{array}{rrrrr} 1.00 & 0.42 & -0.61 & 0.42 & 0.62 \\ 0.42 & 1.00 & -0.23 & -0.53 & 0.12 \\ -0.61 & -0.23 & 1.00 & -0.61 & -0.96 \\ 0.42 & -0.53 & -0.61 & 1.00 & 0.73 \\ 0.62 & 0.12 & -0.96 & 0.73 & 1.00 \\ \end{array} \right],
\quad
C_3 =
\left[ \begin{array}{rrrrr} 1.00 & 0.24 & -0.03 & 0.82 & 0.22 \\ 0.24 & 1.00 & -0.52 & -0.32 & 0.47 \\ -0.03 & -0.52 & 1.00 & 0.16 & -0.96 \\ 0.82 & -0.32 & 0.16 & 1.00 & 0.05 \\ 0.22 & 0.47 & -0.96 & 0.05 & 1.00 \\ \end{array} \right].
$$
\end{scriptsize}

We generate 1,000 random correlation matrices, and for each, we generate
$N=10{,}000$ estimates of $C$. These are sample correlation matrices
based on $X_{t}{\sim}iid N_n(0,C)$ for $t=1,\ldots,T$, with $T=40$,
$T=100$, and $T=250$. We present results based on random correlation
matrices with dimension $n=25$, which translates to $d=300$ distinct
correlation coefficients in $C$, such that the correlation matrix
for $\hat{\varrho}$, $R_{\varrho,T}(C)$, has 44,850 unique correlations.
The same dimensions apply to the other parametrizations, $\hat{\phi}$
and $\hat{\gamma}$.

While the correlation matrices can take any value in $\mathcal{C}_{n}$,
the variation in the estimators' properties across correlation matrices
is largely determined by the smallest eigenvalue $\lambda_{\min}(C)=\min\{\lambda\in\sigma(C)\}$.
For this reason, we present results that are stratified by $\lambda_{\min}(C)$.

\begin{figure}[!htbp]
\centering{}\includegraphics[width=0.9\textwidth]{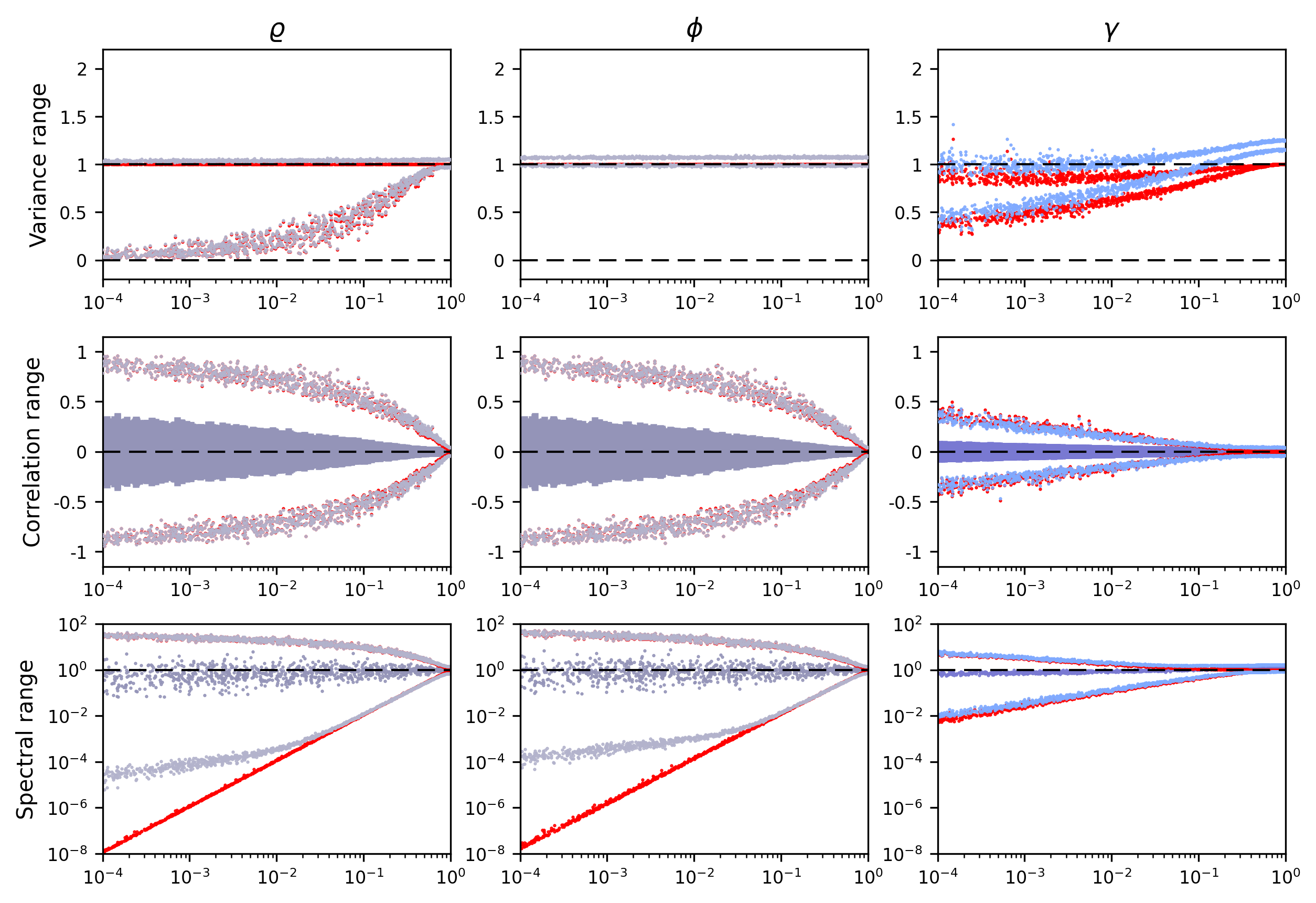}
\caption{\small Properties of $\sqrt{T}(\hat{\varrho}-\varrho)$, $\sqrt{T}(\hat{\phi}-\phi)$, and $\sqrt{T}(\hat{\gamma}-\gamma)$ for $n=25$ ($d=300$), plotted against $\lambda_{\min}(C)$ for 1,000 random correlation matrices. Columns correspond to $\varrho$, $\phi$, and $\gamma$. Rows show marginal variances, pairwise correlations, and covariance-matrix spectra. Blue dots are finite-sample quantities, red dots are asymptotic values, and shaded regions contain 80\% of the pairwise correlations. In the bottom row, eigenvalues are reported in $\log_{10}$ units, and light blue dots show the variances of standardized average elements.\label{fig:VarCorrSpectrum}}
\end{figure}

We focus on the GFT coordinates' ability to stabilize variance
compared to that of $\hat{\phi}$, and the dependence between elements
of the three parametrizations. In the Supplementary Material, we present
Q-Q plots showing that the marginal distributions of the elements
of $\hat{\phi}$ and $\hat{\gamma}$ are well approximated by their
asymptotic normal distributions. We also present results for additional
sample sizes and for dimensions $n=5$ and $n=10$.

\subsubsection{Variances of Correlation Parameters}

Next, we consider the finite-sample variance of the elements of $\sqrt{T}(\hat{\varrho}-\varrho)$,
$\sqrt{T}(\hat{\phi}-\phi)$, and $\sqrt{T}(\hat{\gamma}-\gamma)$,
which are given by the diagonal elements of $V_{\varrho,T}(C)$, $V_{\phi,T}(C)$,
and $V_{\gamma,T}(C)$, respectively. There are $d=300$ variances
for each of the 1,000 random correlation matrices. We present the
smallest and largest finite-sample variances (blue dots) for each
correlation matrix, plotted in the top row of Figure~\ref{fig:VarCorrSpectrum}
against the smallest eigenvalue of $C$, $\lambda_{\min}(C)$, for
the sample size $T=100$. Results for other sample sizes are presented
in the Supplementary Material, see Figure~\ref{fig:RandomDesignVariances}.
For comparison, we also include their asymptotic variances (red dots),
obtained from $V_{\varrho}(C)$, $V_{\phi}(C)$, and $V_{\gamma}(C)$.
The finite-sample variances are estimated by simulating 10,000 independent
estimates for each of the 1,000 correlation matrices.

The results for $\sqrt{T}(\hat{\varrho}-\varrho)$ and $\sqrt{T}(\hat{\phi}-\phi)$
are as expected. The finite-sample variances for $\hat{\varrho}$
depend on the elements of $C$, and asymptotically they approach $(1-C_{ij}^{2})^{2}$.
The Fisher transformation, shown in the middle column of panels, successfully
stabilizes the variances to unity, as it is designed to do for normally
distributed data. The GFT coordinates do not stabilize the marginal
variances to the same degree. There is some dispersion in the variances,
especially for near-singular correlation matrices, as seen for small
values of $\lambda_{\min}(C)$. For $\hat{\gamma}$, we also note that
the finite-sample variances are reasonably close to their asymptotic
values when $T=250$ and $n=25$, and for smaller $n$ the finite-sample
variances are closer to their asymptotic values. The agreement is also
good for smaller sample sizes, including $T=40$ (see Supplementary Material).

Next, we turn to the correlation between the elements within each
of the vectors, $\hat{\varrho}$, $\hat{\phi}$, and $\hat{\gamma}$.
Rather unexpectedly, we observe substantially less dependence between
the elements of $\hat{\gamma}$ than those of the other two vectors.
The middle row of panels in Figure~\ref{fig:VarCorrSpectrum} presents
the largest and smallest finite-sample correlations between elements
of $\hat{\varrho}$, $\hat{\phi}$, and $\hat{\gamma}$ (blue dots)
and their corresponding asymptotic population values (red dots) for
$T=100$. Note that $n=25$ implies $d=300$ elements in each vector,
which translates to $d(d-1)/2=44,850$ distinct correlations. The
shaded area in each panel covers the range for 80\% of these correlations.
The results are plotted against the smallest eigenvalue, $\lambda_{\min}(C)\in[10^{-4},1]$.
Results for other sample sizes are presented in Figure~\ref{fig:R-CorrelationRange}
in the Supplementary Material.

Variance stabilization in the multivariate case is better quantified
by considering linear combinations of the vectors, e.g., $\operatorname{var}(a^\prime\hat{\phi})$
with $a^{\prime}a=1$. The minimum and maximum variance across normalized
linear combinations are given by the smallest and largest eigenvalues
of the covariance matrix, respectively. These are presented in the
last row of panels in Figure~\ref{fig:VarCorrSpectrum} for $T=100$.
Results for other sample sizes are presented in the Supplementary
Material, see Figure~\ref{fig:Veigenvalues}. The GFT coordinates,
$\gamma(C)$, are far better at stabilizing the spectrum of variances,
with the range being orders of magnitude smaller than that of $\hat{\varrho}$
and $\hat{\phi}$ (note that a logarithmic scale is used for the $y$-axis
in lower panels). Thus, while the Fisher transformation is excellent
at standardizing the variance of individual elements of $\hat{\phi}$,
it is unable to moderate covariances between elements.

In the Supplementary Material, Figure~\ref{fig:min_max_var_corr_toeplitz},
we report additional results for the correlations between elements
of $\hat{\gamma}$ for the case where $C$ has a Toeplitz structure,
while maintaining the Gaussian distribution, $X_{t}\sim iid N_n(0,C)$.
Whenever $X_{t}$ is drawn from the Gaussian, we find that the variances
of the elements in $\hat{\gamma}$ have similar values and are approximately
equal to one. This appears to hold for all types of correlation matrices
with eigenvalues bounded away from zero. Moreover, as $n$ increases,
the results for Toeplitz correlation matrices become very similar
to those obtained with random correlation matrices. These findings
suggest that the weak dependence and approximate variance stabilization
of $\hat{\gamma}$ are not artifacts of a particular correlation design,
but rather robust finite-sample features of the GFT coordinates.

\subsection{Results for Non-Gaussian Distributions}

We next examine whether the weak dependence between elements of $\hat{\gamma}$
persists when the data are not Gaussian. The marginal Gaussian approximation
for transformed correlations is known to be more fragile in this case. For
example, skewness and kurtosis of the Fisher transformed correlations can
deviate substantially from the standard normal, as illustrated in Figure
\ref{fig:IGproperties} in the Supplementary Material.

Figure~\ref{fig:R-Correlations-NonGaussian} reports results for samples drawn
from uniform, Student's $t$, and Inverse Gaussian marginal distributions, using
the Toeplitz correlation matrix with $n=25$ and $\rho=0.9$. Each panel shows
the distribution of the 44,850 off-diagonal elements of $R_{\phi,T}(C)$ and
$R_{\gamma,T}(C)$ for $T=40$, $T=100$, and $T=250$. The corresponding results
for $R_{\varrho,T}(C)$ are essentially indistinguishable from those for
$R_{\phi,T}(C)$. The correlations between elements of $\hat{\gamma}$ remain
concentrated near zero in all three non-Gaussian designs, whereas the
correlations between elements of $\hat{\phi}$ are far more dispersed. Thus,
the weak dependence of the GFT coordinates appears to be much more robust than
the marginal Gaussian approximation.

\begin{figure}[!htbp]
\centering{}\includegraphics[width=0.9\textwidth]{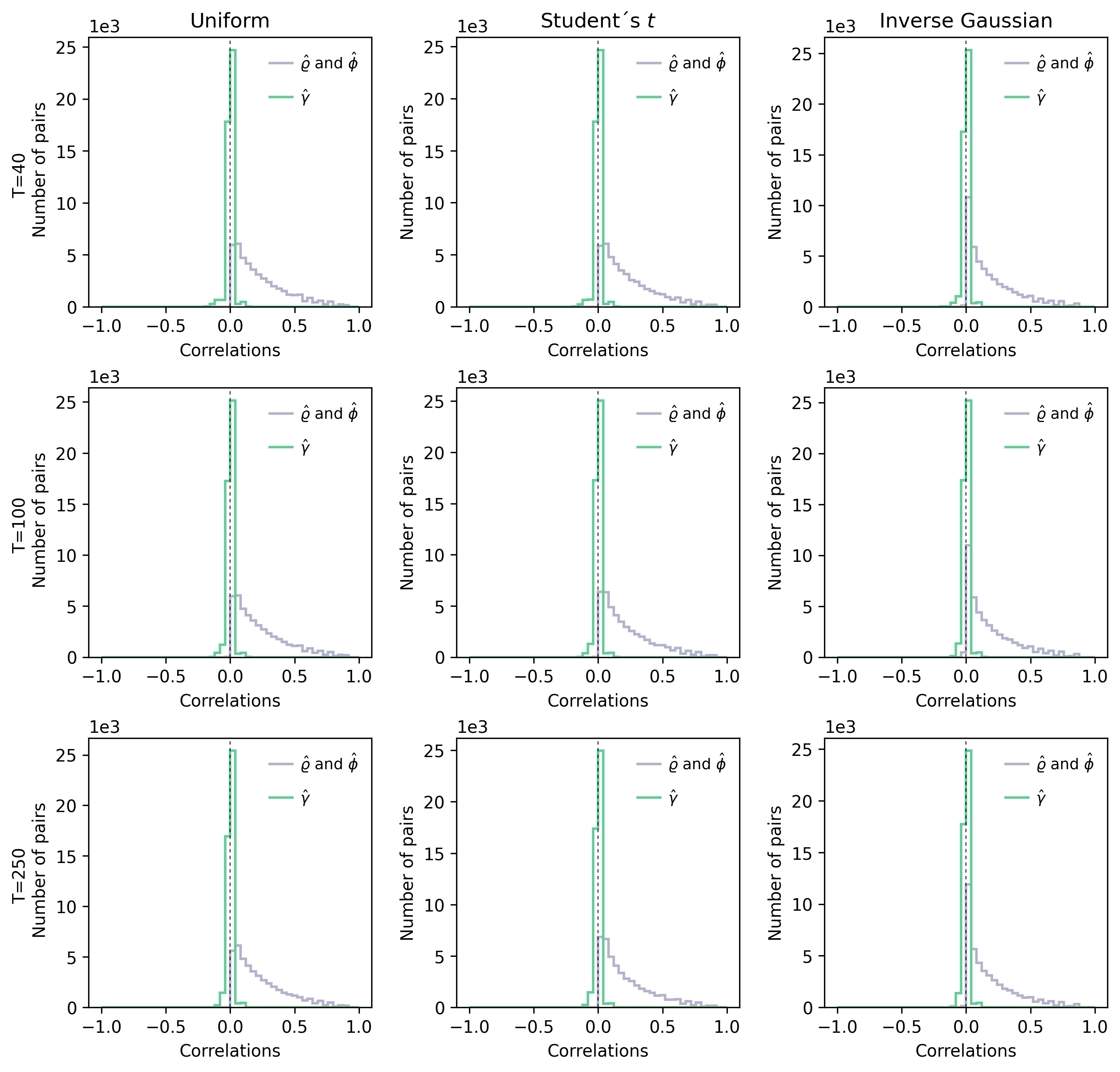}
\caption{Finite-sample correlations between the elements of $\hat{\varrho}$,
$\hat{\phi}$, and $\hat{\gamma}$, when $C$ ($n=25$) has a Toeplitz
structure with $\rho=0.9$ and $X_{it}$ is distributed uniformly
on $[0,1]$, as a Student's $t$ with five degrees of freedom, or
as a standard Inverse Gaussian. The correlations between the elements
of $\hat{\gamma}$ are remarkably concentrated about zero in all cases,
whereas those between elements of $\hat{\varrho}$ and $\hat{\phi}$
are more dispersed. \label{fig:R-Correlations-NonGaussian}}
\end{figure}

\subsection{Sampling from Empirical Data\label{subsec:EmpiricalBasedResults}}

We next study the properties of $\hat{\gamma}$ using empirical structures
from the three data sets used in Figure~\ref{Fig:Dataset}: the FRED-MD
macroeconomic data set from \citet{McCrackenNg:2016}, the 30 Fama-French
industry portfolios, and high-frequency returns for 21 assets, from which
we compute realized correlations using five-minute returns.

We generate artificial samples by resampling, with replacement, from the
empirical distributions and then apply the correlation estimators to these
samples. This preserves empirically relevant correlation structures while
also introducing the non-Gaussian features present in the data. Thus, the
exercise complements the controlled Gaussian and non-Gaussian simulations
above.

To conserve space, we present the results for two of the data sets here, the 30 Fama-French industry portfolios and the high-frequency return data. The
analogous results for a subset of the macro variables are presented in the Supplementary Material.
\begin{table}[htbp]
\begin{centering}
\includegraphics[width=1.0\textwidth]{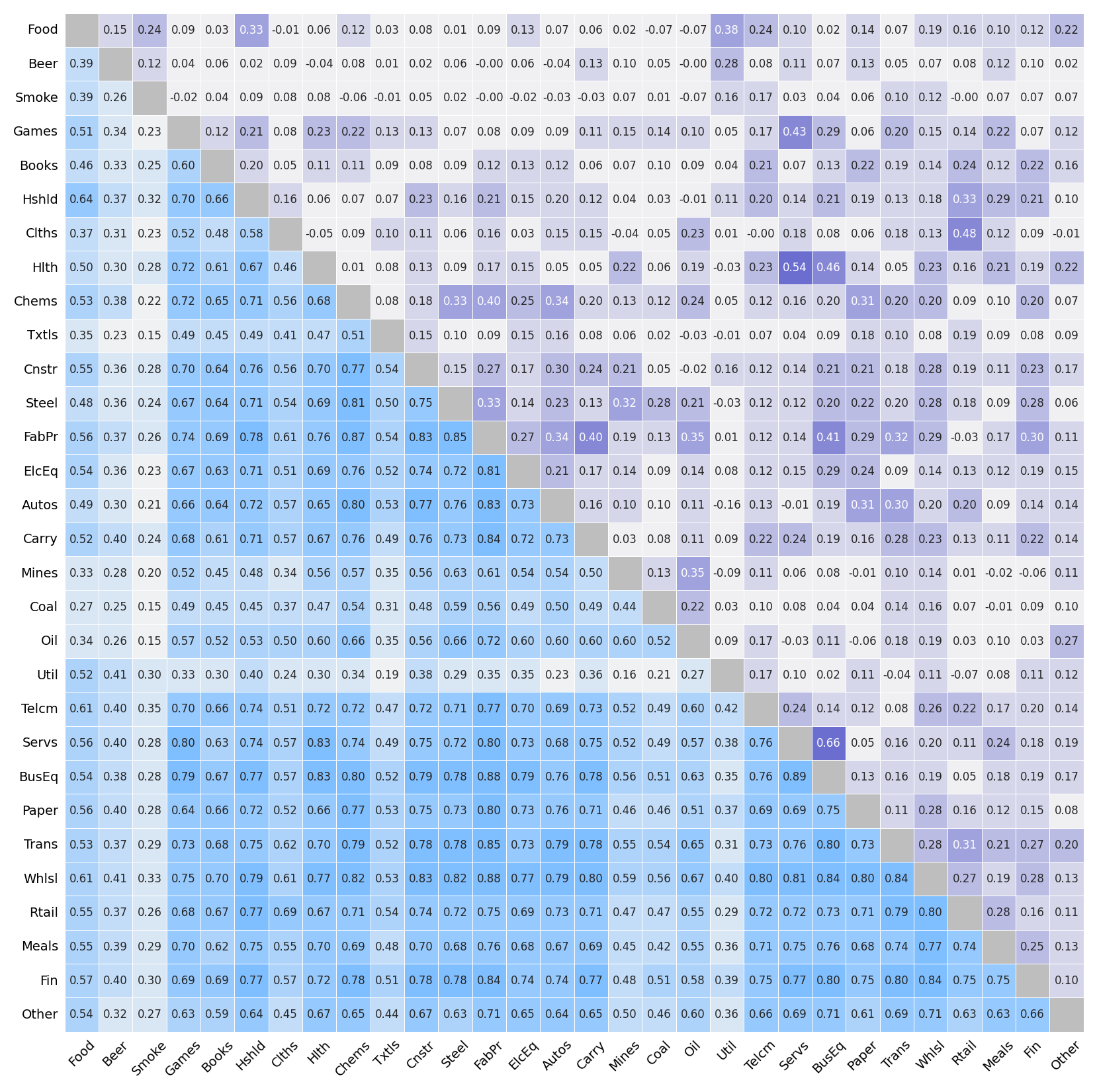}
\par\end{centering}
\caption{{\small Sample correlation matrix for daily returns on 30 Fama-French
industry portfolios. Elements of $\hat{C}$ are below the diagonal
and those of $\log\hat{C}$ are above the diagonal. Sample period
is from January 2, 2018 to December 31, 2019.}\label{tab:ff-Corr}}
\end{table}

\subsubsection{Fama-French 30 Industry Portfolios}\label{sec:FamaFrench}

Our data for the 30 industry portfolios consist of 503 daily returns from January 2, 2018 to December 31, 2019. The sample correlation matrix is shown  below the diagonal in Table~\ref{tab:ff-Corr}, and
the corresponding logarithmic correlation matrix is shown above the diagonal. The smallest eigenvalue is $\lambda_{\min}=0.0704$, so this is a realistic correlation structure that is moderately close to
singularity. 

We construct artificial samples in two ways. First, we simulate Gaussian
return vectors with correlation matrix equal to the empirical correlation
matrix, $X_{t}\sim iid N(0,\hat{C})$. Second, we resample return vectors
from the empirical distribution, with replacement. The first design isolates
the role of the empirical correlation structure, while the second also
incorporates the non-Gaussian features of the observed returns. For each
sample, we estimate its correlation matrix and compute $\hat{\varrho}$,
$\hat{\phi}$, and $\hat{\gamma}$. Figure~\ref{fig:ff-VarSkewKurtT100}
presents results for $T=100$; additional results for $T=40$ and $T=250$
are reported in the Supplementary Material.

The top panels of Figure~\ref{fig:ff-VarSkewKurtT100} show the variance,
skewness, and kurtosis for the 435 elements of
$\sqrt{T}(\hat{\varrho}-\varrho)$, $\sqrt{T}(\hat{\phi}-\phi)$, and
$\sqrt{T}(\hat{\gamma}-\gamma)$ under the Gaussian design. The bottom
panels show the corresponding results based on empirical resampling.
The contrast is clear: under Gaussian sampling, the finite-sample
distributions of $\hat{\phi}$ and $\hat{\gamma}$ are close to their
asymptotic normal distributions, whereas empirical resampling produces
substantial deviations from normality. The GFT coordinates nevertheless
retain better finite-sample behavior than the sample correlations.

\begin{figure}[!htbp]
\begin{centering}
\includegraphics[width=0.9\textwidth]{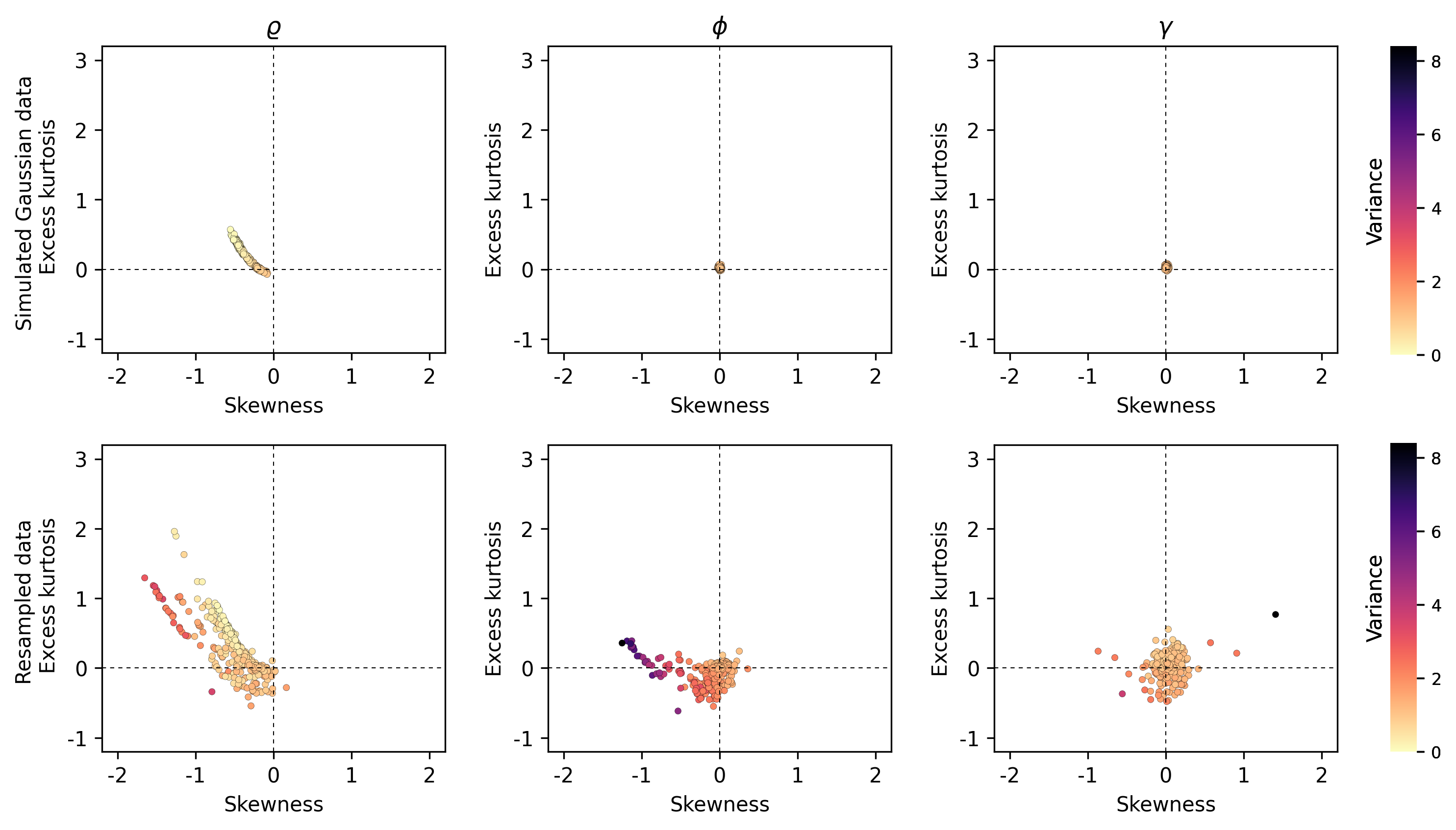}
\par\end{centering}
\centering{}\caption{{\small Finite-sample properties of $\sqrt{T}(\hat{\varrho}-\varrho)$,
$\sqrt{T}(\hat{\phi}-\phi)$, and $\sqrt{T}(\hat{\gamma}-\gamma)$
based on daily returns of 30 Fama-French industry portfolios and $T=100$.
We present the variance, skewness, and kurtosis of the 435 elements
of each vector using two designs. The top panels are for Gaussian
distributed samples using a correlation matrix identical to that of
the data, and the bottom panels are resampled data from the empirical
(non-Gaussian) distribution of return vectors. The finite-sample properties
for $\hat{\phi}$ and $\hat{\gamma}$ are excellent when data are normally distributed, whereas the convergence to the asymptotic normal distribution is clearly far slower when the data are drawn from the
empirical distribution. All quantities are computed from 100,000 artificial samples.}\label{fig:ff-VarSkewKurtT100}}
\end{figure}

Figure~\ref{fig:ff-Contour} provides a bivariate view of the same
phenomenon. It shows the joint finite-sample distributions of selected
elements involving Food Products, Apparel, and Automobiles and Trucks,
computed from empirical resamples with $T=40$, $100$, and $250$. The
left panels are for $\hat{\varrho}$, while the middle and right panels
show the corresponding elements of $\hat{\phi}$ and $\hat{\gamma}$.
The distributions for $\hat{\varrho}$ and $\hat{\phi}$ remain visibly
non-elliptical, and even bimodal, whereas the corresponding distribution
for $\hat{\gamma}$ is much closer to elliptical, even for $T=40$.

\begin{figure}[!htbp]
\centering{}\includegraphics[width=0.9\textwidth]{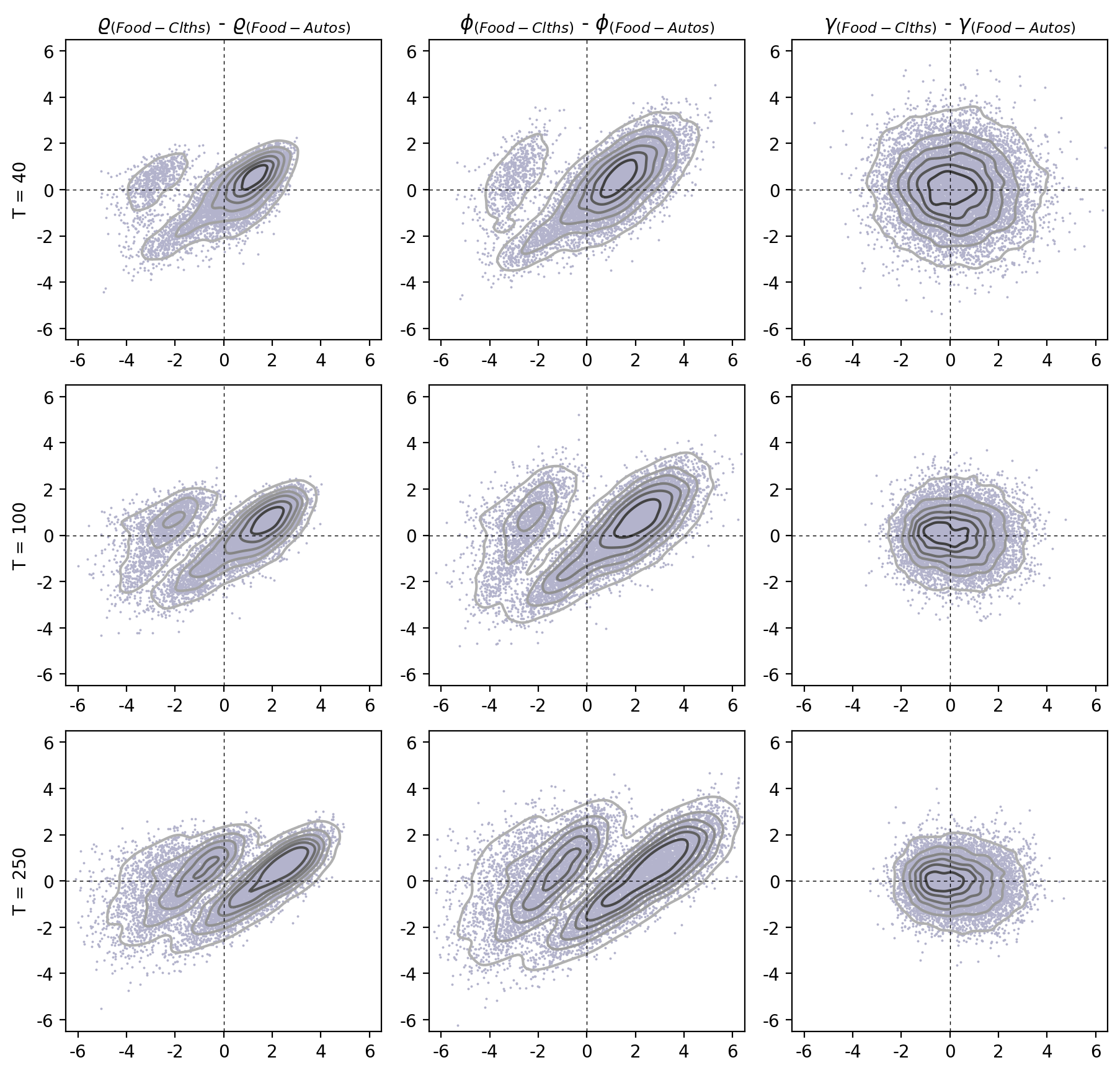}
\caption{{\small Contour plots of bivariate finite-sample distribution of $(\sqrt{T}(\hat{\varrho}_{[\mathrm{Food,Clths}]}-\varrho_{[\mathrm{Food,Clths}]}),\sqrt{T}(\hat{\varrho}_{[\mathrm{Food,Auto}]}-\varrho_{[\mathrm{Food,Auto}]}))$
(left panels), and the corresponding elements from $\hat{\phi}$ and $\hat{\gamma}$
in middle and right panels. Here $\hat{\varrho}_{[\mathrm{Food,Clths}]}$
and $\hat{\varrho}_{[\mathrm{Food,Auto}]}$ are the sample correlations
between the industry portfolio for Food Products and the industry
portfolios for Apparel and Auto, respectively. The contour plots are
based on 10,000 sample correlation matrices computed from resamples
from the empirical distribution of 503 daily return vectors (with
replacement) from the period: January 2, 2018 to December 31, 2019.}\label{fig:ff-Contour}}
\end{figure}

\subsection{Realized Correlation Matrices for 21 Assets}\label{sec:RealizedCorr}

This subsection reports an empirical analysis of realized correlation matrices constructed directly from high-frequency TAQ data. We use the same $n=21$ assets as in Figure~\ref{Fig:Dataset}, but extend the sample from the illustrative week of August 22--26, 2016 to the full period January 1, 2005 to December 31, 2020.

For each ofthe $K=4,027$ trading days, we compute a realized correlation matrix,
$\hat{C}_{t}$, from 78 five-minute intraday returns, yielding daily
time series of $\hat{\varrho}_{t}$, $\hat{\phi}_{t}$, and $\hat{\gamma}_{t}$.
Intraday returns depart from the iid benchmark underlying the closed-form asymptotic variances, since serial correlation and microstructure noise alter the sampling distribution of $\hat C_t$; we therefore read the high-frequency evidence as documenting the robustness of the GFT coordinates under weak dependence rather than as a test of the iid asymptotic variances, and the qualitative findings are unchanged.

We sort days according to the average realized variance across the 21 assets. Figure~\ref{fig:RealCorrHist} presents histograms of realized correlations for the 400 lowest-volatility days and the 400 highest-volatility days. Each histogram is based on $210\times400=84,000$ realized correlations, since $d=210$ in this application. Correlations are higher on high-volatility days, as expected, but the shape of the distribution also changes. In particular, the high-volatility distribution is clearly left-skewed, showing that the distribution of realized correlations depends strongly on the volatility regime.

\begin{figure}[!htbp]
\centering{}\includegraphics[width=0.9\textwidth]{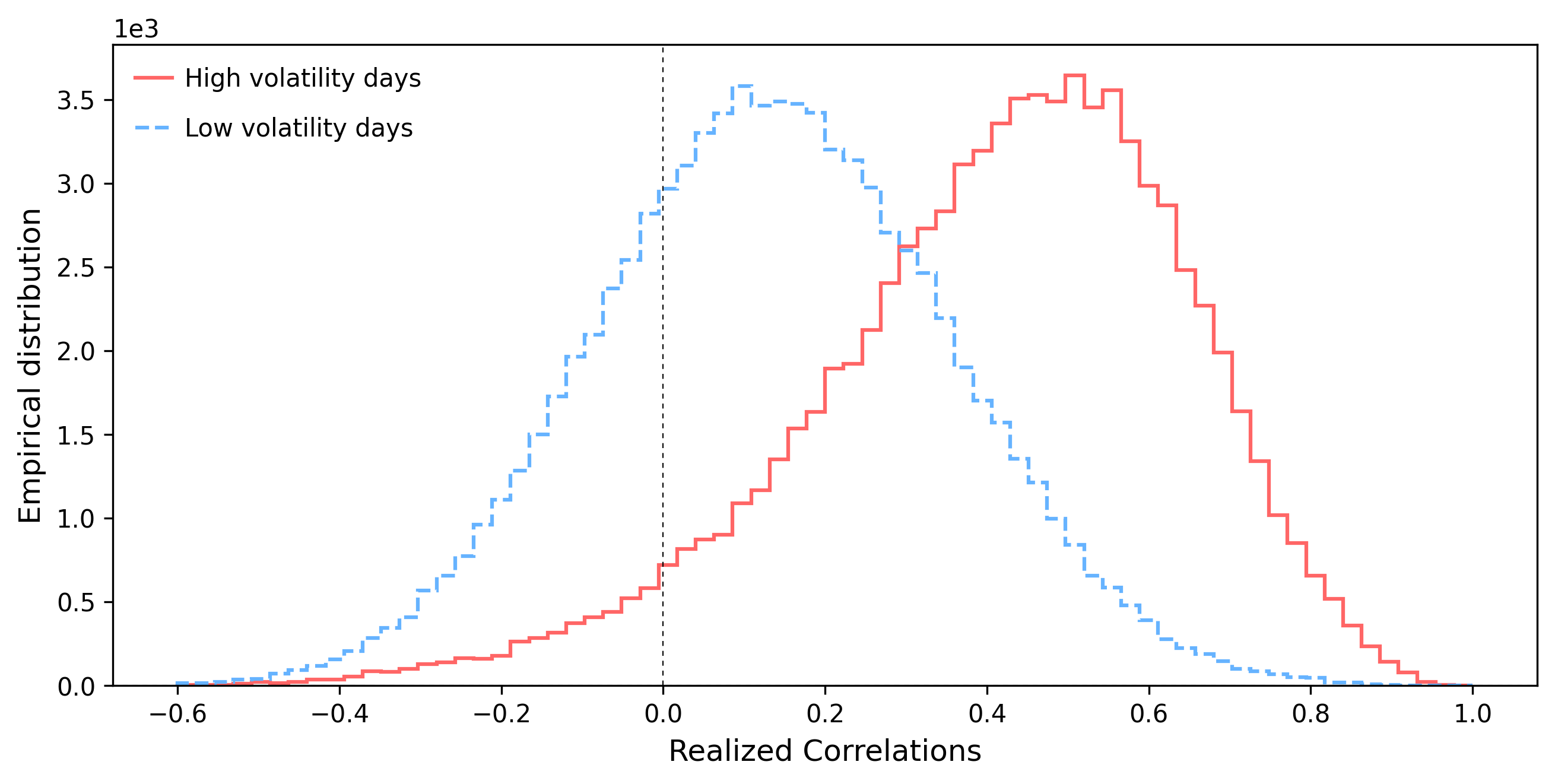}
\caption{{\small Empirical distributions for realized correlations for subsamples
with low (blue) and high (red) market volatility.}\label{fig:RealCorrHist}}
\end{figure}

The regime dependence becomes even clearer in the bivariate case. Figure~\ref{fig:RealCorrPairs} compares joint distributions for two representative coordinate pairs on the low- and high-volatility days. The low- and high-volatility regimes occupy distinct regions of the joint distribution, especially for the Fisher-transformed coordinates. Thus, the apparent bimodality in the pooled distribution of $\hat\phi$ is largely a mixture-of-regimes effect rather than evidence of bimodality within each volatility regime. The corresponding GFT coordinates display less regime separation and weaker dependence.

\begin{figure}[!htbp]
\centering
\includegraphics[width=0.9\textwidth]{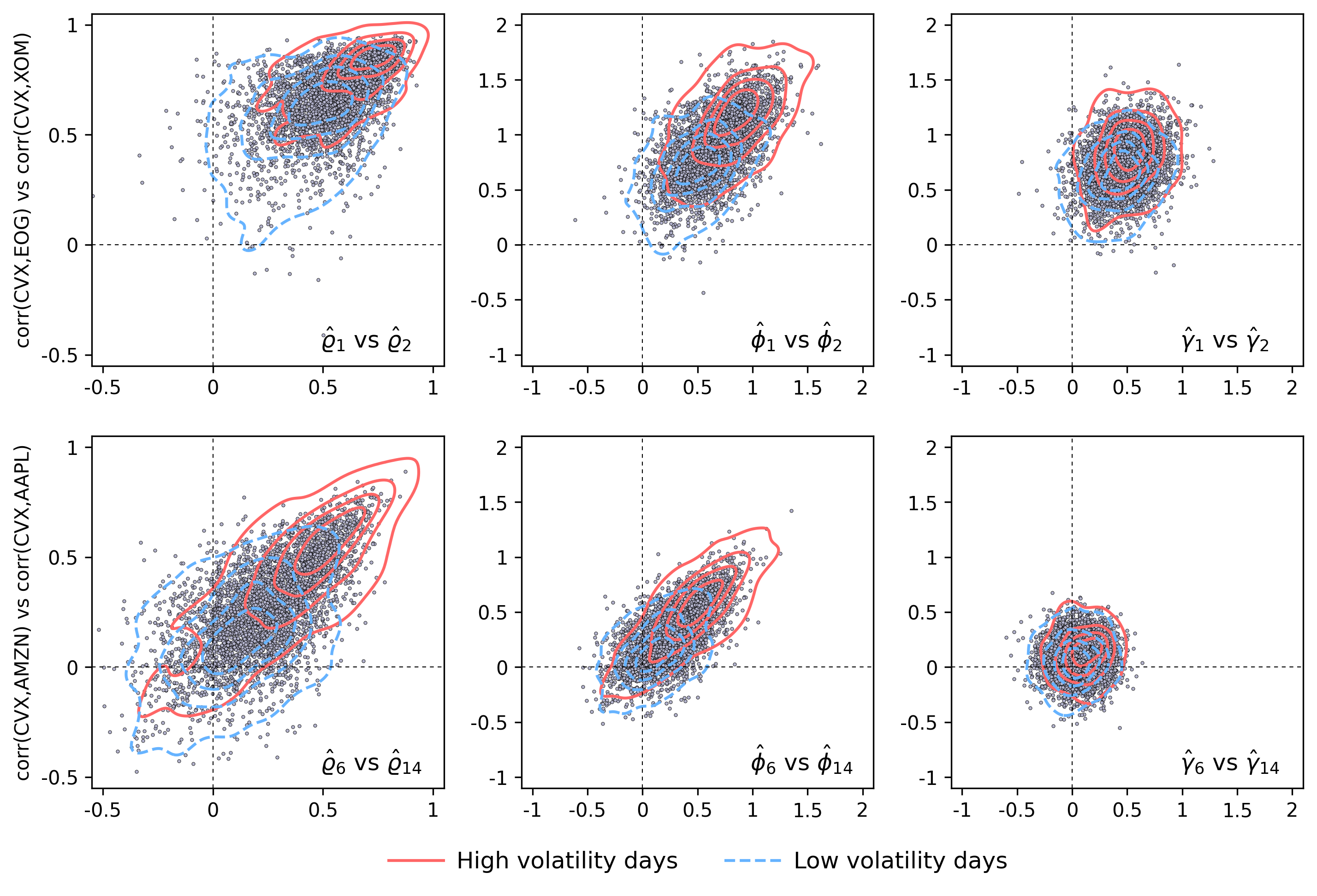}
\caption{\small Joint distributions of two representative coordinate pairs on low- and high-volatility days. The left, middle, and right columns show the corresponding pairs for $\hat\varrho$, $\hat\phi$, and $\hat\gamma$, respectively. Solid red contours correspond to the 400 highest-volatility days, and dashed blue contours correspond to the 400 lowest-volatility days. The separation between volatility regimes is most pronounced for the Fisher-transformed coordinates, which helps explain the apparent bimodality in the pooled marginal distribution of $\hat\phi$. The GFT coordinates show less regime separation and weaker dependence.}
\label{fig:RealCorrPairs}
\end{figure}

We revisit \citet{ABDE:2001}, who found that the unconditional empirical
distribution of realized correlations, $\hat{\varrho}_{i,t}$, for stock
returns is approximately normally distributed.\footnote{This empirical result
was also demonstrated for exchange rate data; see \citet{ABDL:2001}.}
Our focus is on the distributional properties of the measurement errors in
the realized measures. Since the latent paths $\varrho_t$, $\phi_t$, and
$\gamma_t$ are unobserved, we approximate them using a standard linear
Gaussian filter, obtaining smoothed paths $\varrho_t^f$, $\phi_t^f$, and
$\gamma_t^f$. We then define standardized residuals by
$\hat{\varepsilon}(\hat{\varrho})=\sqrt{T}(\hat{\varrho}_{t}-\varrho_{t}^{f})$,
$\hat{\varepsilon}(\hat{\phi})=\sqrt{T}(\hat{\phi}_{t}-\phi_{t}^{f})$,
and $\hat{\varepsilon}(\hat{\gamma})=\sqrt{T}(\hat{\gamma}_{t}-\gamma_{t}^{f})$,
and use these as proxies for the scaled measurement errors.
\begin{figure}[!htbp]
\centering{}\includegraphics[width=0.9\textwidth]{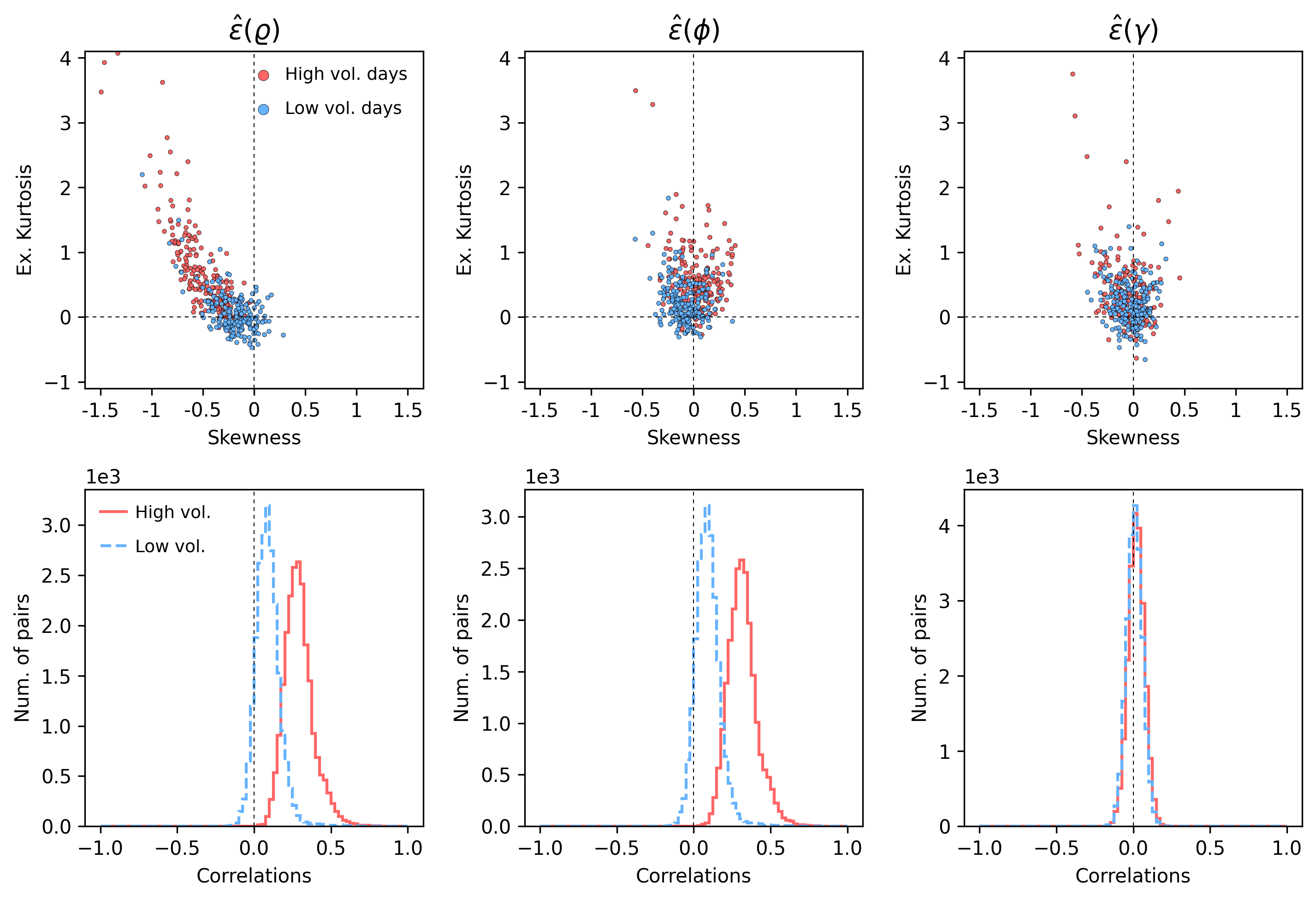}
\caption{{\small On the upper panels, skewness and excess kurtosis of the distributions
of the model residuals, $\hat{\varepsilon}_{i}(\hat{\varrho}_{i})$,
$\hat{\varepsilon}_{i}(\hat{\phi}_{i})$, and $\hat{\varepsilon}_{i}(\hat{\gamma}_{i})$,
for $i=1,\ldots,210$ correlation elements. On the lower panels, the distributions
of all 21,945 pairwise correlations between elements in the vectors of residuals,
$\hat{\varepsilon}(\hat{\varrho})$, $\hat{\varepsilon}(\hat{\phi})$,
and $\hat{\varepsilon}(\hat{\gamma})$. The distributions are constructed from
the residuals obtained for the subsamples of 400 trading days with the lowest
(blue) and highest (red) market volatility.}\label{fig:RealCorrResiduals}}
\end{figure}

The upper panels of Figure~\ref{fig:RealCorrResiduals} show the skewness
and excess kurtosis of these residuals, computed separately for the low-
and high-volatility subsamples. Residuals based on realized correlations
have substantial negative skewness and positive excess kurtosis, especially
on high-volatility days. The residuals based on $\hat{\phi}$ and $\hat{\gamma}$
are much closer to Gaussian and display fewer systematic differences across
volatility regimes.

The lower panels of Figure~\ref{fig:RealCorrResiduals} show the distributions
of correlations between residual elements. Dependence among the residuals is
substantial for $\hat{\varepsilon}(\hat{\varrho})$ and
$\hat{\varepsilon}(\hat{\phi})$, especially in the high-volatility subsample.
In contrast, the residuals based on $\hat{\gamma}$ are only weakly dependent:
their average correlation is close to zero, and the distribution is similar
across the low- and high-volatility subsamples.

Overall, the empirical results agree with the simulation evidence. Elements
of the log-transformed realized correlation matrix are approximately
uncorrelated, and their measurement errors are closer to Gaussian than those
of realized correlations. Thus, the finite-sample properties of $\hat{\gamma}$
appear to persist in a realistic data environment with non-Gaussian intraday
returns and latent correlations that are likely to vary within the trading day
\citep[see][]{HansenLuo-RobustCorr:2023}.

\section{GFT Coordinates: Covariance Stability and Weak Dependence}\label{sec:covstability}

Section~\ref{sec:MarginalJoint} documents that the GFT coordinates have much weaker finite-sample dependence than the raw correlations and the element-wise Fisher transformed correlations. We now give two theoretical results that explain this pattern. The first shows that the asymptotic covariance matrix $V_\gamma(C)$ is stable because the matrix logarithm maps perturbations through a well-conditioned linear operator except near the boundary of the correlation cone. The second shows that, locally around $C=I_n$, the GFT cancels the first-order dependence among overlapping sample correlations. Together, these results explain why $V_\gamma(C)$ is both comparatively stable in $C$ and nearly diagonal in empirically relevant designs.

\subsection{A spectral bound}

The stability of $V_\gamma(\cdot)$ admits a clean explanation under elliptical sampling. Normalizing the sample covariance to a correlation matrix removes the radial component of the data, so the kurtosis enters the GFT covariance only through a scalar; the remaining sensitivity is geometric and is governed by the conditioning of $A_C$, which deteriorates only as the eigenvalue spread of $C$ grows.

For an elliptical distribution we write $X_t\sim E(0,\Sigma,\kappa)$, where $\kappa$ is the scalar kurtosis parameter appearing in the asymptotic covariance of $\hat\Sigma$, as made explicit in the proof. The Gaussian case corresponds to $\kappa=0$.

\begin{theorem}[Spectral bound for the GFT covariance]\label{thm:spectral}
Let $X_1,\ldots,X_T$ be iid draws from $E(0,\Sigma,\kappa)$, where $\Sigma$ is positive definite and fourth moments are finite. Let $C$ be the implied correlation matrix and $\hat C$ the sample correlation matrix. Then, as $T\to\infty$,
$$
\sqrt{T}\operatorname{vecl}(\log\hat C-\log C)\xrightarrow{d}N(0,V_\gamma(C)),
$$
and
$\lambda_{\max}(V_\gamma(C))\leq(1+\kappa)\|\Pi_C\|_2^2$ with 
$\Pi_C=A_C^{-1}P_CA_C$, 
where $\|\cdot\|_2$ is the spectral norm, $A_C=\partial\operatorname{vec}(C)/\partial\operatorname{vec}(\log C)$, and
$$
P_C=I_{n^2}-\frac{1}{2}(I_{n^2}+K_n)(I_n\otimes C)M_d
$$
is the Jacobian of the covariance-to-correlation map at $\Sigma=C$, 
where $K_n$ is the commutation matrix and $M_d$ is the diagonal-selection matrix satisfying
$M_d\operatorname{vec}(X)=\operatorname{vec}(\operatorname{diag}(X))$ for every $n\times n$ matrix $X$, with $\operatorname{diag}(X)$ denoting the diagonal matrix with the same diagonal entries as $X$.
The bound is sharp: equality is approached as $C\to I_n$.
\end{theorem}

\begin{remark}The bound in Theorem~\ref{thm:spectral} reveals two critical features of the GFT covariance. First, the distributional shape of the data is summarized entirely by the scalar $1+\kappa$. This is what makes the relative behavior of the GFT covariance across elliptical distributions so stable: changing the tails of the data rescales $V_\gamma(C)$ but does not reshape it. Second, the geometry enters through $\|\Pi_C\|_2$, where $\Pi_C$ differs from the correlation Jacobian $P_C$ only by the logarithmic-mean multiplier $\Xi$ of (\ref{eq:Aeig}). When the eigenvalues of $C$ are comparable, $\Xi$ is nearly constant and $\Pi_C\approx P_C$, so $V_\gamma(C)$ is nearly the same matrix for every such $C$, and the bound tightens toward the common value $1+\kappa$.
\end{remark}
The conditioning of $A_C$ inflates $\|\Pi_C\|_2$ only as $C$ approaches singularity, which is exactly the regime in which $\lambda_{\min}(C)$ was found to drive finite-sample behavior in Section~\ref{sec:MarginalJoint}. The bound thus provides an analytic account of both the invariance of $V_\gamma(\cdot)$ and its single point of fragility.

\subsection{Second-order local orthogonality}\label{subsec:secondorderorthogonality}
The spectral bound explains why the GFT covariance remains stable when $C$ is away from the boundary of the correlation cone. We next isolate a more local mechanism. Around $C=I_n$, the raw correlations and the element-wise Fisher transformed correlations acquire first-order dependence through overlapping index pairs, whereas the GFT cancels these first-order terms.

For a symmetric zero-diagonal matrix $\Delta$, let $B_\Delta$ denote the $d\times d$ matrix with zero diagonal whose first-order off-diagonal entries are determined by triangles in the lower-triangular indexing convention:
$$
[B_\Delta]_{(ij),(ik)}=\Delta_{jk},\qquad
[B_\Delta]_{(ij),(jk)}=\Delta_{ik},\qquad
[B_\Delta]_{(ik),(jk)}=\Delta_{ij},
$$
for distinct indices $i,j,k$, with entries involving disjoint pairs set equal to zero.

\begin{corollary}[Second-order local orthogonality of the GFT]\label{cor:secondorderorthogonality}
Let $C=I_n+\Delta$, where $\Delta$ is symmetric with zero diagonal, and let $\|\Delta\|_2\to0$. Under the conditions of Theorem~\ref{thm:spectral},
$$
\underset{=V_\varrho(C)}{\underbrace{(1+\kappa)(I_d+B_\Delta)+O(\|\Delta\|_2^2)}},
\quad
\underset{=V_\phi(C)}{\underbrace{(1+\kappa)(I_d+B_\Delta)+O(\|\Delta\|_2^2)}},
\quad
\underset{=V_\gamma(C)}{\underbrace{(1+\kappa)I_d+O(\|\Delta\|_2^2)}}.
$$
The matrix order terms are in spectral norm. Thus the raw-correlation and Fisher coordinates inherit first-order dependence from triangles in the correlation graph, whereas the GFT covariance is diagonal up to second order around $I_n$.
\end{corollary}

The local diagonalness of $V_\gamma(C)$ is therefore stronger than the corresponding local diagonalness of $V_\varrho(C)$ and $V_\phi(C)$. All three covariance matrices are diagonal at $C=I_n$, but the raw and Fisher coordinates acquire first-order off-diagonal terms as soon as the correlation matrix is perturbed away from the identity. These terms arise from triangles: for example, the asymptotic covariance between $\hat\varrho_{ij}$ and $\hat\varrho_{ik}$ is proportional to $C_{jk}$ to first order. The Fisher transform does not remove this effect, because its derivative is equal to one at the origin and its first nonlinear correction is cubic.

The GFT behaves differently. The matrix logarithm has the expansion $\log(I_n+\Delta)=\Delta-\frac{1}{2}\Delta^2+O(\|\Delta\|_2^3)$, so each GFT coordinate subtracts, to second order, the indirect two-step correlation paths running through the other variables. This is the local cancellation mechanism behind the weak dependence seen in the simulations: the raw and Fisher coordinates inherit first-order triangle dependence, whereas the GFT cancels it to first order.

\section{Inference}\label{sec:inference}

The finite-sample regularities documented in Section~\ref{sec:MarginalJoint} and the covariance-stability results in Section~\ref{sec:covstability} have a single inferential payoff: in the GFT coordinates the estimation error is approximately Gaussian, weakly dependent, and governed by a covariance matrix that varies little with the unknown $C$. Standardizing with a plug-in estimate of this covariance is therefore reliable in samples where the analogous operation is much less stable for $\hat\varrho$ and $\hat\phi$. This section examines the consequences for standardized statistics, plug-in covariance estimation, and Wald tests.

\subsection{Standardized statistics and the plug-in problem}

Inference about $C$ can be based on the limit distributions in (\ref{eq:LimitDistributions}), with the asymptotic covariances estimated by the plug-in estimators $V_\varrho(\hat C)$, $V_\phi(\hat C)$, and $V_\gamma(\hat C)$. Define the standardized statistics
\begin{equation}\label{eq:Zstats}
\begin{aligned}
Z_{\varrho,T}&=\sqrt T [V_\varrho(\hat C)]^{-1/2}(\hat\varrho-\varrho),\\
Z_{\phi,T}&=\sqrt T [V_\phi(\hat C)]^{-1/2}(\hat\phi-\phi),\\
Z_{\gamma,T}&=\sqrt T [V_\gamma(\hat C)]^{-1/2}(\hat\gamma-\gamma),
\end{aligned}
\end{equation}
where $A^{1/2}=Q\Lambda^{1/2}Q^\prime $ for the eigendecomposition $A=Q\Lambda Q^\prime $. If $\hat C\xrightarrow{p}C$ and the relevant covariance is continuous at $C$, then each statistic converges to $N_d(0,I_d)$. In finite samples, however, the quality of the $N_d(0,I_d)$ approximation is controlled not by the limit but by two distinct sources of error: the non-normality of the underlying coordinate vector, and the error in the plug-in covariance, $V(\hat C)-V(C)$. The GFT improves both, and for a common reason.

The decisive quantity is the sensitivity of $V(\cdot)$ to its argument. Whitening in (\ref{eq:Zstats}) requires $V(\hat C)$ to be a good estimate of $V(C)$, yet $V(\hat C)$ inherits the full estimation error in $\hat C$. When $V(\cdot)$ varies steeply with $C$, as $V_\varrho(\cdot)$ and $V_\phi(\cdot)$ do, a moderate error in $\hat C$ produces a large error in the whitening matrix, and the standardized statistic is poorly behaved even when its components are individually close to normal. The weak dependence documented in Section~\ref{sec:MarginalJoint} and the stability results in Section~\ref{sec:covstability} imply that the GFT whitening matrix is much less sensitive to the plug-in step: $V_\gamma(\hat C)\approx V_\gamma(C)$ over a wide range of designs, so the standardization behaves much closer to the infeasible standardization based on the true covariance matrix.

\subsection{Finite-sample normality of the standardized statistics}

We first assess the standardized statistics directly. We draw $10{,}000$ random $5\times5$ correlation matrices using the design of Section~\ref{subsec:RandomCorr}; for each we simulate a sample, compute $\hat C$, and form $Z_{\varrho,T}$, $Z_{\phi,T}$, and $Z_{\gamma,T}$ (each of dimension $d=10$). Figure~\ref{fig:qqRhoPhiGamma} reports the pooled Q--Q plots against the standard normal for $T=40,100,250$.

\begin{figure}[!htbp]
\centering{}\includegraphics[width=0.9\textwidth]{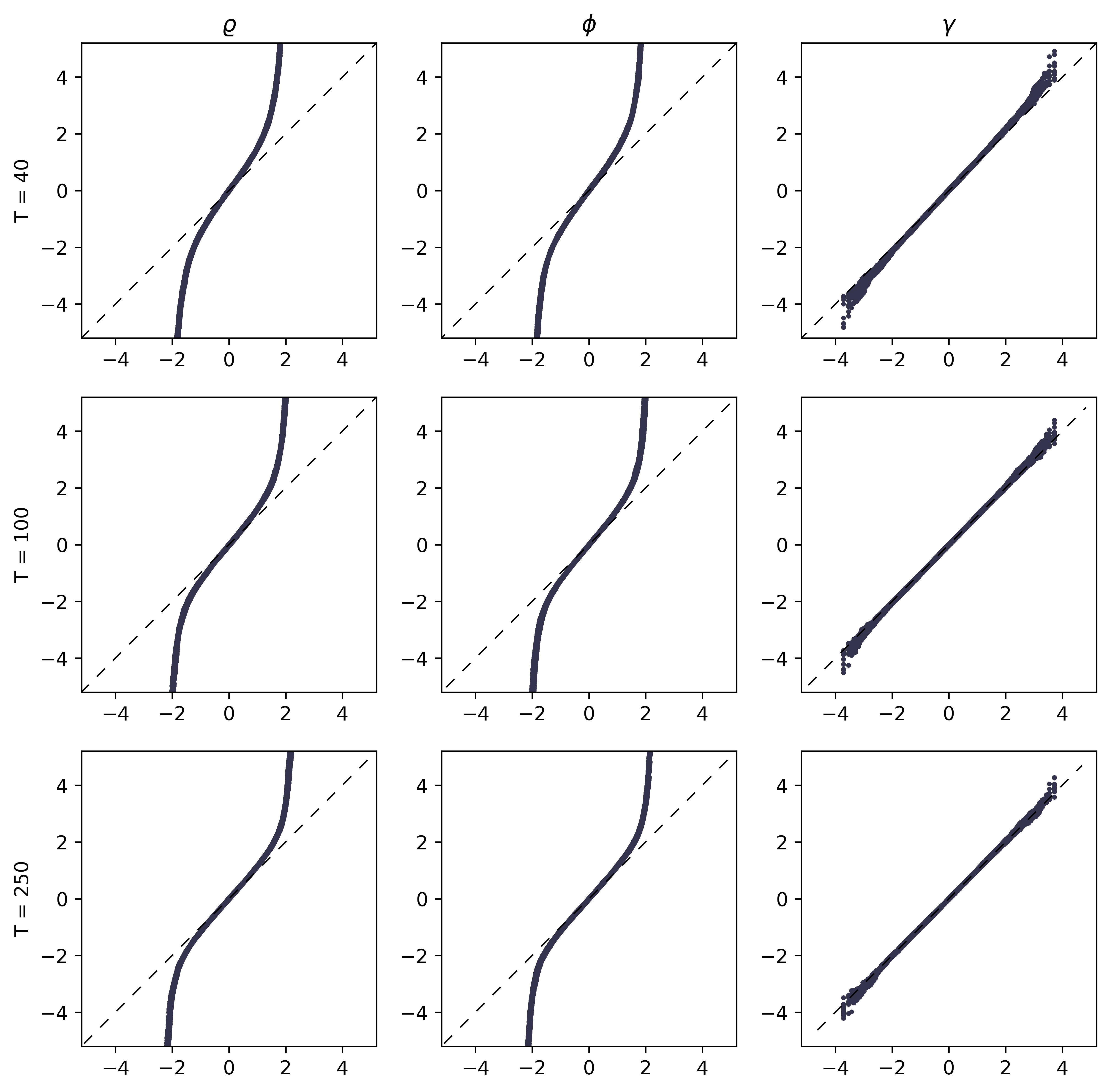}
\caption{{\small Q--Q plots of the standardized estimation errors $Z_{\varrho,T}$ (left), $Z_{\phi,T}$ (middle), and $Z_{\gamma,T}$ (right), for $T=40,100,250$. The solid line is the standard normal reference.}\label{fig:qqRhoPhiGamma}}
\end{figure}

The contrast is stark, and it separates the two error sources. The poor behavior of $Z_{\varrho,T}$ is expected from the non-normality of $\hat\varrho$ documented earlier. The behavior of $Z_{\phi,T}$ is more instructive: even though the \emph{marginal} elements of $\hat\phi$ are close to normal, the dependence among them is strong and highly sensitive to $C$, so the plug-in matrix $V_\phi(\hat C)$ is a poor estimate of $V_\phi(C)$ and the whitening in (\ref{eq:Zstats}) degrades the joint distribution. $Z_{\gamma,T}$ is close to standard normal already at $T=40$ because the GFT coordinates are both nearly Gaussian and nearly orthogonal, and because $V_\gamma(\hat C)$ is much closer to $V_\gamma(C)$ than the corresponding plug-in covariance matrices for $\hat\varrho$ and $\hat\phi$, as documented in Section~\ref{sec:MarginalJoint} and explained by Section~\ref{sec:covstability}.

\subsection{Plug-in covariance estimation}

We quantify the plug-in error directly with the Stein loss
\begin{equation}\label{eq:stein}
L(\hat V,V)=\operatorname{tr}\{\hat VV^{-1}\}-\log\det(\hat VV^{-1})-d ,
\end{equation}
computed for the plug-in estimators of $V_\varrho(C)$, $V_\phi(C)$, and $V_\gamma(C)$ across $1{,}000$ random correlation matrices ($10{,}000$ samples each) and plotted against $\lambda_{\min}(C)$ in Figure~\ref{fig:SteinLoss}.

\begin{figure}[!htbp]
\centering{}\includegraphics[width=0.9\textwidth]{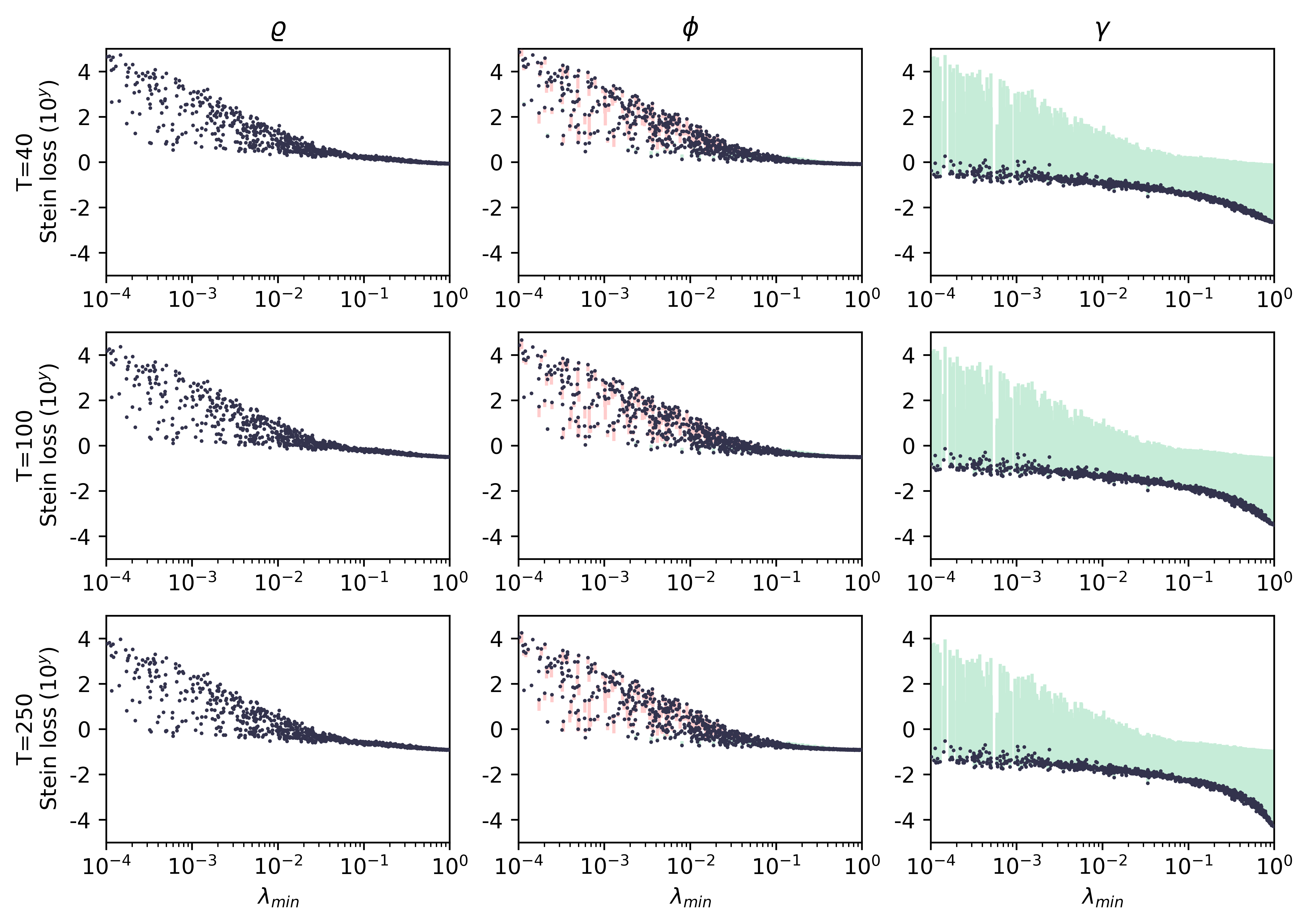}
\caption{{\small Stein loss (\ref{eq:stein}) for the plug-in estimators of the asymptotic covariance matrices of $\hat\varrho$, $\hat\phi$, and $\hat\gamma$. Red segments in the middle column mark the loss difference between $V_\varrho(\hat C)$ and $V_\phi(\hat C)$; green segments in the right column mark the difference between $V_\phi(\hat C)$ and $V_\gamma(\hat C)$. Units are $\log_{10}$, so one vertical unit is a factor of ten.}\label{fig:SteinLoss}}
\end{figure}

The reduction is large and systematic. On the $\log_{10}$ scale, $L\big(V_\gamma(\hat C),V_\gamma(C)\big)$ is typically $2$ to $5$ units below the corresponding losses for $\hat\varrho$ and $\hat\phi$, a reduction in plug-in covariance loss by factors of roughly $10^2$ to $10^5$. The gap widens as $\lambda_{\min}(C)\to0$, consistent with Theorem~\ref{thm:spectral}: it is precisely near singularity that $\|\Pi_C\|_2$, and hence the curvature of $V(\cdot)$ that the plug-in must contend with, grows.

\subsection{Wald tests}

Finally, Figure~\ref{fig:WaldTest} reports the size of Wald tests of $H_0\!: C=C_0$ at the nominal $5\%$ level, using $W_{\varrho,T}=Z_{\varrho,T}^\prime Z_{\varrho,T}$, $W_{\phi,T}=Z_{\phi,T}^\prime Z_{\phi,T}$, and $W_{\gamma,T}=Z_{\gamma,T}^\prime Z_{\gamma,T}$ against the $\chi^2_d$ critical value. The upper panels use an equicorrelation matrix with all correlations $0.8$; the lower panels use the Toeplitz structure (\ref{eq:CorrelationMatrixToeplitz}) with $\rho=0.8$.

\begin{figure}[!htbp]
\centering{}\includegraphics[width=0.9\textwidth]{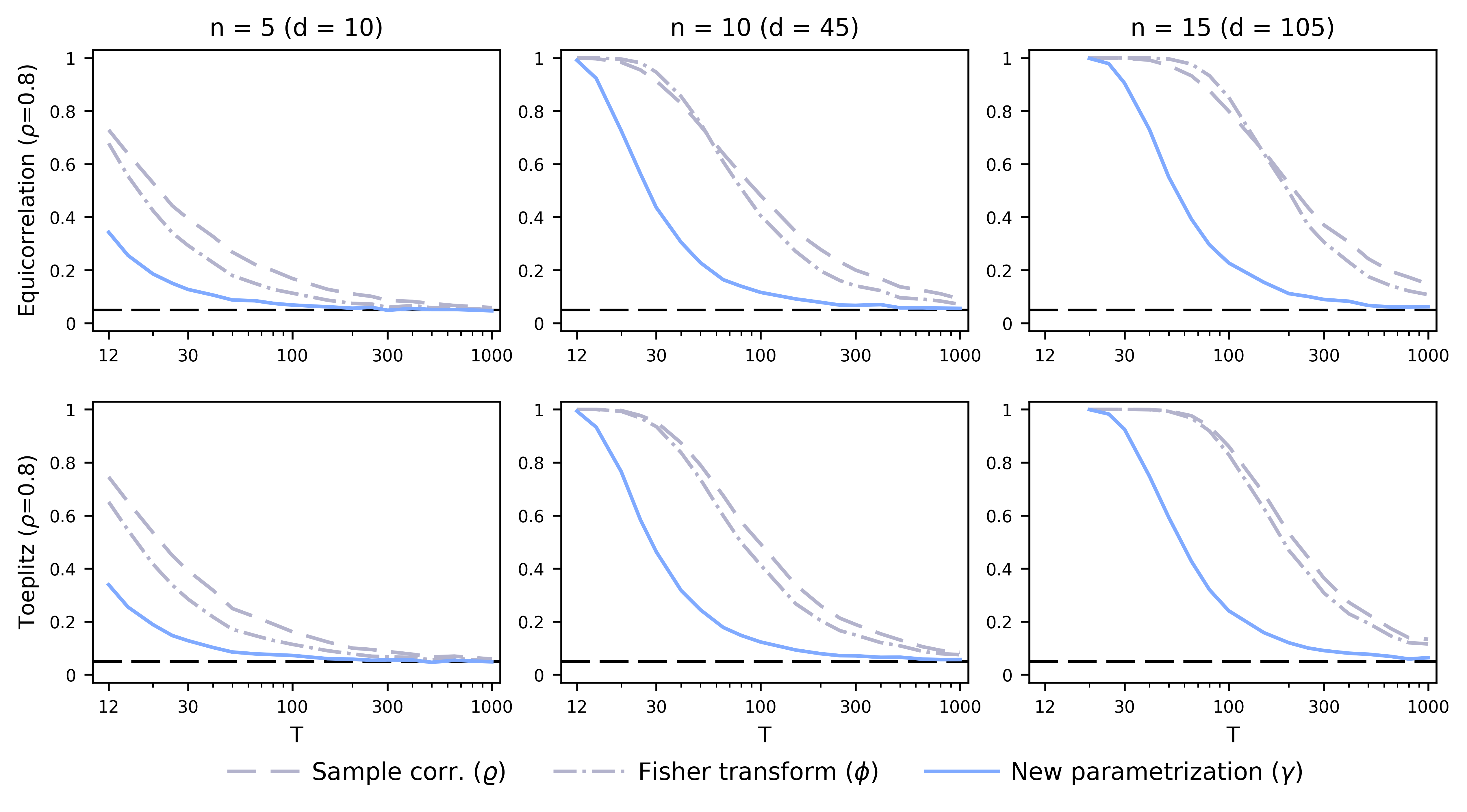}
\caption{{\small Empirical size of nominal $5\%$ Wald tests of $H_0\!:\,C=C_0$ based on $\hat\varrho$, $\hat\phi$, and $\hat\gamma$, as a function of the sample size.}\label{fig:WaldTest}}
\end{figure}

All three tests are oversized when $d/T$ is moderate or large, and all approach the nominal level as $T\to\infty$, since the standardized statistics share the same limit. Convergence is far faster for $W_{\gamma,T}$: in these designs $W_{\varrho,T}$ and $W_{\phi,T}$ require roughly five times as many observations to reach a comparable proximity to nominal size. The improvement traces back to the same source as in Figures~\ref{fig:qqRhoPhiGamma} and \ref{fig:SteinLoss}: a whitening matrix that is reliable because $V_\gamma(\cdot)$ is nearly invariant to $C$. The practical implication is that GFT-based Wald inference is usable at sample sizes where the conventional and element-wise Fisher tests are not.

\section{Conclusion}\label{sec:conclusion}

A non-singular correlation matrix, $C$, can be parametrized by
$\gamma(C)=\operatorname{vecl}\log C$. We have investigated
$\hat{\gamma}=\gamma(\hat{C})$, where $\hat{C}$ is the sample correlation matrix, with emphasis on its finite-sample distribution and dependence
structure. When $\hat{C}$ is computed from independent Gaussian data, the marginal distributions of the elements of $\hat{\gamma}$ are well approximated by their asymptotic normal distributions. Since $\gamma(C)$ coincides with Fisher's transformation when $C$ is a $2\times2$ correlation matrix, this extends a familiar property of the scalar Fisher transformation to higher-dimensional correlation matrices. As with the scalar Fisher transformation, however, the marginal Gaussian approximation can deteriorate
when the data are strongly non-Gaussian.

The more striking finding is the joint behavior of $\hat{\gamma}$. The elements of $\hat{\gamma}$ are nearly uncorrelated in finite samples, and
the covariance matrix $V_{\gamma,T}(C)$ is far more stable across values of
$C$ than the corresponding covariance matrices for $\hat{\varrho}$ and
$\hat{\phi}$. This weak dependence appears to be a robust feature of the
GFT coordinates. It is present for Gaussian samples, persists across a range
of non-Gaussian designs, and is also evident in simulations based on empirical
distributions of macroeconomic variables, daily industry returns, and
high-frequency financial data.

These properties have direct implications for inference. Because
$V_{\gamma}(C)$ is relatively insensitive to the true correlation matrix,
the plug-in estimator $V_{\gamma}(\hat{C})$ is much less affected by
estimation error in $\hat{C}$ than the analogous plug-in covariance matrices
for $\hat{\varrho}$ and $\hat{\phi}$. As a result, standardized statistics
based on
$V_{\gamma}^{-1/2}(\hat{C})\sqrt{T}(\hat{\gamma}-\gamma)$ have substantially
better finite-sample behavior than those based on the sample correlations or
the element-wise Fisher transformed correlations. This is reflected in both
the Q-Q plots and the Wald-test simulations, where the GFT-based statistic
converges much more rapidly to its asymptotic reference distribution.

Our analysis assumes iid observations, which underlies both the closed-form asymptotic covariances and the bootstrap resampling in the empirical designs. Many target applications, such as realized correlations and multivariate volatility, involve serial dependence and, at high frequency, microstructure noise. Under weak dependence the asymptotic covariance of $\hat{C}$ acquires the usual long-run (HAC) form, and feasible inference would replace the iid plug-in with a corresponding long-run covariance estimator. The properties that motivate the GFT, however, namely the near-orthogonality and near-invariance of the coordinate system and the scalar role of higher-order moments in the spectral bound, are properties of the transformation rather than of the sampling scheme, and the high-frequency evidence in Section~\ref{sec:RealizedCorr} suggests they persist under realistic dependence. A formal treatment of dependent data is left for future work.

The approximate orthogonality and covariance stability of $\hat{\gamma}$
also suggest that the GFT coordinates may be useful for regularization of
large correlation and covariance matrices. A systematic treatment of this
possibility, including near-singular and high-dimensional cases, is left for
future research.

\section*{Acknowledgments}
We thank participants at the 2022 Vienna--Copenhagen Conference on Financial Econometrics, the inaugural Virtual Time Series Seminar (VTSS), the 2023 Aarhus Workshop in Econometrics II, the 2024 SoFiE Conference, and seminars at University of Freiburg, the University of Chicago Booth School of Business, and Toronto Metropolitan University for helpful comments and discussions. The second author thanks York University for its hospitality during a research visit.

\section*{Funding and Disclosure Statement}
The authors received no financial support for this research. The authors declare no conflicts of interest.

\section*{Declaration of Generative AI and AI-Assisted Technologies}

In preparing this manuscript, the authors used ChatGPT (GPT-5.5 Thinking, OpenAI) and Claude (Opus 4.8, Anthropic) for language editing, notational consistency checks, and proofreading. The authors reviewed and edited all output and take full responsibility for the content of the manuscript, including the accuracy of the results, references, and conclusions.

\section*{Data Availability Statement}

The FRED-MD macroeconomic data are publicly available from the Federal Reserve Bank of St. Louis. The Fama-French 30 industry portfolio returns are publicly available from Kenneth French's Data Library. The high-frequency intraday return data used in Section~\ref{sec:RealizedCorr} are based on NYSE TAQ data, which are proprietary and cannot be redistributed by the authors. Researchers with TAQ access may reconstruct the data from the securities and date range described in the paper. Simulation code and replication scripts are available from the authors upon request.

\appendix
\setcounter{lemma}{0}
\renewcommand{\thelemma}{A.\arabic{lemma}}

\section{Proofs}\label{app:proofs}
 
Throughout, $N_n=\tfrac12(I_{n^2}+K_n)$ denotes the symmetrizer, so that $N_n\operatorname{vec}(M)=\operatorname{vec}\!\big(\tfrac12(M+M^\prime )\big)$, $N_n^2=N_n$, $N_n=N_n^\prime $, and $N_n(B\otimes B)=(B\otimes B)N_n$ for any symmetric $B$. 
 
\begin{lemma}[De-duplication identity]\label{lem:dedup}
Let $E_l$ be the $d\times n^2$ elimination matrix with $E_l\operatorname{vec}(M)=\operatorname{vecl}(M)$, so that $E_lE_l^\prime =I_d$. Then $E_l N_n E_l^\prime =\tfrac12 I_d$, or equivalently $J\equiv \sqrt2 N_nE_l^\prime $ is an isometry, $J^\prime J=I_d$ and $\|J\|_2=1$.
\end{lemma}
\begin{proof}
$E_lN_nE_l^{\prime}=\tfrac12\big(E_lE_l^{\prime}+E_lK_nE_l^{\prime}\big)=\tfrac12 I_d$, because $E_lE_l^{\prime}=I_d$ and $E_lK_nE_l^{\prime}=0$: the commutation matrix $K_n$ sends each strictly-below-diagonal $\operatorname{vec}$-position to its strictly-above-diagonal transpose, and $E_l$ selects only the former, so no position survives both. Hence $J^{\prime}J=2E_lN_nN_nE_l^{\prime}=2E_lN_nE_l^{\prime}=I_d$.
\end{proof}

\noindent Lemma~\ref{lem:dedup} is the algebraic reason the factor $(I_{n^2}+K_n)=2N_n$ in the covariance of $\operatorname{vec}(\hat C)$ contributes a factor $1$ rather than $2$ once attention is restricted to the $d=n(n-1)/2$ GFT coordinates: each off-diagonal entry is counted twice in $\operatorname{vec}$ but only once in $\operatorname{vecl}$. It is the half-vectorization (duplication/elimination) identity $E_lN_nE_l^{\prime}=\tfrac12 I_d$ stated for the strictly-lower selection.

\begin{lemma}[Logarithmic-mean contraction]\label{lem:logmean}
Let $C$ be positive definite and let
$A_C=\partial\operatorname{vec}(C)/\partial\operatorname{vec}(\log C)$. Then
$A_C^{-1}(C\otimes C)A_C^{-1}\preceq I_{n^2}$.
\end{lemma}

\begin{proof}
Let $C=Q\Lambda_CQ^\prime$ with
$\Lambda_C=\operatorname{diag}(\lambda_1,\dots,\lambda_n)$. By the standard Fr\'echet derivative formula for the matrix exponential, applied at
$\log C=Q(\log\Lambda_C)Q^\prime$ (see \citet{LintonMcCrorie:1995} and
\citet{ArchakovHansen:Correlation}), the Jacobian
$A_C=\partial\operatorname{vec}(C)/\partial\operatorname{vec}(\log C)$ has the spectral representation
\begin{equation}\label{eq:Aeig}
A_C=(Q\otimes Q)\Xi(Q^\prime\otimes Q^\prime),\qquad
\Xi_{(i,j)}=\begin{cases}
\qquad \lambda_i,&\lambda_i=\lambda_j,\\[2pt]
\dfrac{\lambda_i-\lambda_j}{\log\lambda_i-\log\lambda_j},&\lambda_i\neq\lambda_j,
\end{cases}
\end{equation}
where $\Xi$ is diagonal and $\Xi_{(i,j)}$ is the logarithmic mean of $\lambda_i$ and $\lambda_j$. Hence
$$
A_C^{-1}(C\otimes C)A_C^{-1}
=(Q\otimes Q)D_C(Q^\prime\otimes Q^\prime),
\qquad
[D_C]_{(i,j)}=\frac{\lambda_i\lambda_j}{\Xi_{(i,j)}^2}.
$$
Since the logarithmic mean is bounded below by the geometric mean,
$\Xi_{(i,j)}\geq\sqrt{\lambda_i\lambda_j}$, each diagonal entry of
$D_C$ satisfies
$[D_C]_{(i,j)}\leq 1$. Thus $D_C\preceq I_{n^2}$, and the result follows because $Q\otimes Q$ is orthogonal.
\end{proof}

\noindent \textbf{Proof of Theorem~\ref{thm:spectral}.} Consider $C=\operatorname{diag}(\Sigma)^{-1/2}\Sigma\operatorname{diag}(\Sigma)^{-1/2}$. Differentiating $C_{ij}=\Sigma_{ij}/\sqrt{\Sigma_{ii}\Sigma_{jj}}$ and evaluating at $\Sigma=C$ (so $\Sigma_{ii}=1$), one obtains, in matrix form, $dC=d\Sigma-\tfrac12(\Lambda C+C\Lambda)$ with $\Lambda=\operatorname{diag}((d\Sigma)_{11},\ldots,(d\Sigma)_{nn})$. Vectorizing and using $\operatorname{vec}(\Lambda)=M_d\operatorname{vec}(d\Sigma)$, together with $\operatorname{vec}(C\Lambda)=(I_n\otimes C)\operatorname{vec}(\Lambda)$ and $\operatorname{vec}(\Lambda C)=(C\otimes I_n)\operatorname{vec}(\Lambda)$, gives $\operatorname{vec}(dC)=P_C\operatorname{vec}(d\Sigma)$ with $P_C=I_{n^2}-\tfrac12(C\otimes I_n+I_n\otimes C)M_d$. 
Because $\Lambda$ is diagonal and $C$ is symmetric, $\Lambda C=(C\Lambda)^\prime$, and hence
$(C\otimes I_n)\operatorname{vec}(\Lambda)=K_n(I_n\otimes C)\operatorname{vec}(\Lambda)$ for all such $\Lambda$, i.e. $(C\otimes I_n)M_d=K_n(I_n\otimes C)M_d$. Substituting yields $P_C=I_{n^2}-\tfrac12(I_{n^2}+K_n)(I_n\otimes C)M_d$.
 
Next, for $X_t\sim\mathrm{iid}\,E(0,\Sigma,\kappa)$ the sample covariance satisfies $\sqrt T\operatorname{vec}(\hat\Sigma-\Sigma)\xrightarrow{d}N(0,(1+\kappa)(I_{n^2}+K_n)(\Sigma\otimes\Sigma)+\kappa\operatorname{vec}(\Sigma)\operatorname{vec}(\Sigma)^\prime)$, 
see \citet{BrowneShapiro:1986}. By the delta method with the Jacobian, $P_C$, evaluated at $\Sigma=C$, $V_C=(1+\kappa)P_C(C\otimes C)(I_{n^2}+K_n)P_C^\prime+\kappa P_C\operatorname{vec}(C)\operatorname{vec}(C)^\prime P_C^\prime$. Since $C$ is a correlation matrix, $M_d\operatorname{vec}(C)=\operatorname{vec}(I_n)$, and therefore
$
P_C\operatorname{vec}(C)=\operatorname{vec}(C)-\tfrac12(I_{n^2}+K_n)(I_n\otimes C)\operatorname{vec}(I_n)=\operatorname{vec}(C)-\tfrac12(I_{n^2}+K_n)\operatorname{vec}(C)=\tfrac12(I_{n^2}-K_n)\operatorname{vec}(C)=0$, because $C$ is symmetric. The radial (kurtosis) term vanishes and
\begin{equation}\label{eq:VC}
V_C=(1+\kappa)P_C(C\otimes C)(I_{n^2}+K_n)P_C^\prime .
\end{equation}
Thus, after normalization to a correlation matrix, non-Gaussianity enters \emph{only} through the scalar $1+\kappa$; the entire dependence on the shape of the distribution has collapsed to a single number.
 
Moreover, the representation (\ref{eq:Aeig}) shows that $A_C$ is symmetric positive definite and commutes with $K_n$ and $N_n$.
 
The GFT covariance is $V_\gamma(C)=E_lA_C^{-1}V_CA_C^{-1}E_l^\prime$. Substitute (\ref{eq:VC}) and write $(I_{n^2}+K_n)=2N_n$, so that $V_C=2(1+\kappa)P_C(C\otimes C)N_nP_C^\prime$. Because $(C\otimes C)$ commutes with $N_n$ and $N_n^2=N_n$, we have $(C\otimes C)N_n=N_n(C\otimes C)N_n$; writing $(C\otimes C)^{1/2}=C^{1/2}\otimes C^{1/2}$, which also commutes with $N_n$, this factors as a square, $P_C(C\otimes C)N_nP_C^\prime=H_C H_C^\prime$ with $H_C \equiv P_CN_n(C\otimes C)^{1/2}$. Hence $V_\gamma(C)=2(1+\kappa)(E_lA_C^{-1}H_C)(E_lA_C^{-1}H_C)^\prime$ and
\begin{equation}\label{eq:lamM}
\lambda_{\max}(V_\gamma(C))=2(1+\kappa)\big\|E_lA_C^{-1}H_C\big\|_2^2 .
\end{equation}
We bound $\|E_lA_C^{-1}H_C\|_2$ in two steps. First, the operator $A_C^{-1}P_CN_n$ maps into the symmetric subspace, i.e.
\begin{equation}\label{eq:symrange}
N_nA_C^{-1}P_CN_n=A_C^{-1}P_CN_n .
\end{equation}
To see this, note first that $P_C$ preserves the symmetric subspace. Indeed, if $S$ is symmetric and $\Lambda_S=\operatorname{diag}(S_{11},\ldots,S_{nn})$, then $P_C\operatorname{vec}(S)=\operatorname{vec}\left\{S-\tfrac12(\Lambda_SC+C\Lambda_S)\right\}$, which is symmetric because $S$ is symmetric and $(\Lambda_SC+C\Lambda_S)^\prime=\Lambda_SC+C\Lambda_S$. Hence $N_nP_CN_n=P_CN_n$. Since $A_C^{-1}$ commutes with $N_n$ by (\ref{eq:Aeig}), this gives (\ref{eq:symrange}). Therefore, using $E_lA_C^{-1}H_C=E_l(A_C^{-1}P_CN_n)(C\otimes C)^{1/2}$ and (\ref{eq:symrange}), $E_lA_C^{-1}H_C=E_lN_n(A_C^{-1}P_CN_n)(C\otimes C)^{1/2}=\tfrac{1}{\sqrt2}J^\prime(A_C^{-1}P_CN_n)(C\otimes C)^{1/2}$, where $J^\prime=\sqrt2E_lN_n$ has $\|J^\prime\|_2=1$ by Lemma~\ref{lem:dedup}. Hence $\big\|E_lA_C^{-1}H_C\big\|_2\le\tfrac{1}{\sqrt2}\big\|A_C^{-1}P_CN_n(C\otimes C)^{1/2}\big\|_2$. Second, writing $A_C^{-1}P_C=\Pi_CA_C^{-1}$ (with $\Pi_C=A_C^{-1}P_CA_C$, using $A_C$ symmetric) and $A_C^{-1}N_n=N_nA_C^{-1}$, $A_C^{-1}P_CN_n(C\otimes C)^{1/2}=\Pi_C N_n A_C^{-1}(C\otimes C)^{1/2}$, so, since $\|N_n\|_2=1$ and $\big\|A_C^{-1}(C\otimes C)^{1/2}\big\|_2^2=\lambda_{\max}(A_C^{-1}(C\otimes C)A_C^{-1})\leq 1$ by Lemma~\ref{lem:logmean}, $\big\|A_C^{-1}P_CN_n(C\otimes C)^{1/2}\big\|_2\le\|\Pi_C\|_2$.
 
Together, these two bounds imply
$\big\|E_lA_C^{-1}H_C\big\|_2^2\le\tfrac12\|\Pi_C\|_2^2$. Substituting this into (\ref{eq:lamM}) gives $\lambda_{\max}(V_\gamma(C))\le(1+\kappa)\|\Pi_C\|_2^2$ with $\Pi_C=A_C^{-1}P_CA_C$, which is the claimed bound. This completes the proof.
\hfill{}$\square$\medskip{}
 
The constant in the bound is sharp. As $C\to I_n$, one has
$A_C\to I_{n^2}$, $\Pi_C\to P_C$, $\|\Pi_C\|_2\to1$, and
$\lambda_{\max}(V_\gamma(C))\to1+\kappa$. A direct submultiplicative
bound that ignores the half-vectorization step would give only
$2(1+\kappa)\|\Pi_C\|_2^2$; Lemma~\ref{lem:dedup} is what removes this
extraneous factor of two.

\noindent \textbf{Proof of Corollary~\ref{cor:secondorderorthogonality}.}
Write $C=I_n+\Delta$, where $\Delta$ is symmetric with zero diagonal, and let $H$ denote the Gaussian limit of $\sqrt{T}(\hat C-C)$. 
The covariance formula for the sample correlation matrix in Appendix~\ref{app:additional} gives, uniformly for distinct indices $i,j,k$,
$$
\operatorname{cov}(H_{ij},H_{ik})=(1+\kappa)\Delta_{jk}+O(\|\Delta\|_2^2),\qquad
\operatorname{var}(H_{ij})=(1+\kappa)+O(\|\Delta\|_2^2),
$$
with analogous expressions for the other overlapping pairs, while disjoint pairs have covariance $O(\|\Delta\|_2^2)$. Hence
$V_\varrho(C)=(1+\kappa)(I_d+B_\Delta)+O(\|\Delta\|_2^2)$. 
Since the Fisher transform is element-wise and
$\operatorname{arctanh}(x)=x+\frac{x^3}{3}+O(x^5)$, with derivative $1+O(x^2)$, its Jacobian at $C=I_n+\Delta$ is $I_d+O(\|\Delta\|_2^2)$. Therefore the Fisher transform leaves the first-order covariance terms unchanged, and
$$
V_\phi(C)=(1+\kappa)(I_d+B_\Delta)+O(\|\Delta\|_2^2).
$$
For the GFT, the Fr\'echet derivative of the matrix logarithm at $C=I_n+\Delta$ is
$$
D\log_C[H]=H-\frac{1}{2}(\Delta H+H\Delta)+O(\|\Delta\|_2^2\|H\|_2).
$$
Thus, for an overlapping pair $(ij)$ and $(ik)$,
$$
D\gamma_{ij}[H]=H_{ij}-\frac{1}{2}\Delta_{jk}H_{ik}+O_p(\|\Delta\|_2),\qquad
D\gamma_{ik}[H]=H_{ik}-\frac{1}{2}\Delta_{jk}H_{ij}+O_p(\|\Delta\|_2),
$$
where the displayed terms are the only first-order terms that can contribute to the covariance between these two coordinates. Indeed, the remaining first-order terms involve covariances at $C=I_n$ between distinct off-diagonal elements that are zero unless the two elements coincide, by the Pearson--Filon formula. Using the identities above at $C=I_n$, the first-order covariance is
$$
(1+\kappa)\Delta_{jk}-\frac{1}{2}(1+\kappa)\Delta_{jk}-\frac{1}{2}(1+\kappa)\Delta_{jk}=0.
$$
Hence $[V_\gamma(C)]_{(ij),(ik)}=O(\|\Delta\|_2^2)$. The same argument applies to the other overlapping configurations, $(ij),(jk)$ and $(ik),(jk)$, after relabeling the shared and non-shared indices. The diagonal entries satisfy $[V_\gamma(C)]_{(ij),(ij)}=(1+\kappa)+O(\|\Delta\|_2^2)$, since $\Delta_{ij}$ enters only quadratically in the marginal variance. Combining the diagonal and off-diagonal expansions gives $V_\gamma(C)=(1+\kappa)I_d+O(\|\Delta\|_2^2)$. This completes the proof.
\hfill{}$\square$

\section{Expressions for the Asymptotic Variances}\label{app:additional}

The asymptotic distributions for $\hat{\varrho}$, $\hat{\phi}$,
and $\hat{\gamma}$ are inherited from those of $\hat{C}$. Suppose
that $\sqrt T\operatorname{vec}(\hat C-C)\xrightarrow{d}N(0,V_C)$, 
as $T\rightarrow\infty$, where $V_C$ is an $n^2\times n^2$ matrix of rank
$r\leq n(n-1)/2$. Convenient closed-form expressions for $V_{C}$ are available in special cases, see \citet{Nel:1985}, \citet{BrowneShapiro:1986},
and \citet{NeudeckerWesselman:1990}. For sample correlation matrices
computed from independent vectors, the elements of $V_{C}$ are
given by
\begin{align*}
\mu_{i} &= \mathbb{E}[X_{it}],\\
\sigma_{ij} &= \mathbb{E}\bigl[(X_{it}-\mu_{i})(X_{jt}-\mu_{j})\bigr],\\
\mu_{ijkl} &= \mathbb{E}\bigl[(X_{it}-\mu_{i})(X_{jt}-\mu_{j})(X_{kt}-\mu_{k})(X_{lt}-\mu_{l})\bigr],\\
\kappa_{ijkl} &= \frac{\mu_{ijkl}}{\sqrt{\sigma_{ii}\sigma_{jj}\sigma_{kk}\sigma_{ll}}}.
\end{align*}
The elements of $V_{C}$ are then
$$
[V_{C}]_{ij,kl}=\kappa_{ijkl}
+\tfrac{1}{4}C_{ij}C_{kl}\bigl(\kappa_{iikk}+\kappa_{jjkk}+\kappa_{iill}+\kappa_{jjll}\bigr)
-\tfrac{1}{2}C_{ij}\bigl(\kappa_{iikl}+\kappa_{jjkl}\bigr)
-\tfrac{1}{2}C_{kl}\bigl(\kappa_{ijkk}+\kappa_{ijll}\bigr).
$$
In the special case where $X_{t}$ is normally distributed, the expression for $[V_{C}]_{ij,kl}$ simplifies to that of \citet{PearsonFilon:1898},
$$
\begin{aligned}
[V_{C}]_{ij,kl}={}&C_{ik}C_{jl}+C_{il}C_{jk}
+\tfrac{1}{2}C_{ij}C_{kl}\bigl(C_{ik}^{2}+C_{jk}^{2}+C_{il}^{2}+C_{jl}^{2}\bigr)\\
&-C_{ij}\bigl(C_{ik}C_{il}+C_{jk}C_{jl}\bigr)
-C_{kl}\bigl(C_{ik}C_{jk}+C_{il}C_{jl}\bigr).
\end{aligned}
$$
Since $\hat{\varrho}=\operatorname{vecl}(\hat{C})$ and $\hat{\phi}$ is an element-wise transformation of $\hat{\varrho}$, their asymptotic variances are
$V_{\varrho}(C)=E_{l}V_{C}E_{l}^{\prime}$ and $V_{\phi}(C)=D_{\varrho}E_{l}V_{C}E_{l}^{\prime}D_{\varrho}$, where $E_{l}$ is the elimination matrix characterized by $\operatorname{vecl}[M]=E_{l}\operatorname{vec}[M]$ for any $n\times n$ matrix $M$, and
$
D_{\varrho}=\operatorname{diag}\!\left(\tfrac{1}{1-\varrho_{1}^{2}},\tfrac{1}{1-\varrho_{2}^{2}},\ldots,\tfrac{1}{1-\varrho_{d}^{2}}\right),
$
see \citet{LinPerlman1985}. For the GFT coordinates it follows that
$V_{\gamma}(C)=E_{l}A_C^{-1}V_{C}A_C^{-1}E_{l}^{\prime}$, where $A_C=\partial\operatorname{vec}(C)/\partial\operatorname{vec}(\log C)$ is the symmetric Jacobian given in (\ref{eq:Aeig}).

{\footnotesize\bibliographystyle{chicago}
\setstretch{1.3}
\bibliography{prh}
}{\footnotesize\par}

\clearpage{}

\setcounter{equation}{0}\renewcommand{\theequation}{S.\arabic{equation}}
\setcounter{figure}{0}\renewcommand{\thefigure}{S.\arabic{figure}}
\setcounter{table}{0}\renewcommand{\thetable}{S.\arabic{table}}
\setcounter{section}{0}\renewcommand{\thesection}{S.\arabic{section}}
\setcounter{page}{1}\renewcommand{\thepage}{S.\arabic{page}}

\part*{Supplementary Material}

\section{Additional Non-Gaussian Simulation Results}

We next consider a non-Gaussian design with a known correlation structure. The sample correlation matrices are computed from random samples with $n=25$, where $C$ is the $25\times25$ Toeplitz correlation matrix in (\ref{eq:CorrelationMatrixToeplitz}) with $\rho=0.9$. Figure~\ref{fig:IGproperties} reports finite-sample results for the three parametrizations, $\hat{\varrho}$, $\hat{\phi}$, and $\hat{\gamma}$, when the marginal distributions are Inverse Gaussian with scale and location parameters equal to one. These marginals are skewed and leptokurtic, with skewness equal to three and excess kurtosis equal to 15. The figure plots excess kurtosis against skewness for the 300 elements of each transformed correlation vector, for sample sizes $T=40$, $T=100$, and $T=250$. The shade of each point indicates the corresponding variance. The results are based on 100,000 simulations.

\begin{figure}[!htbp]
\centering{}\includegraphics[width=0.95\textwidth]{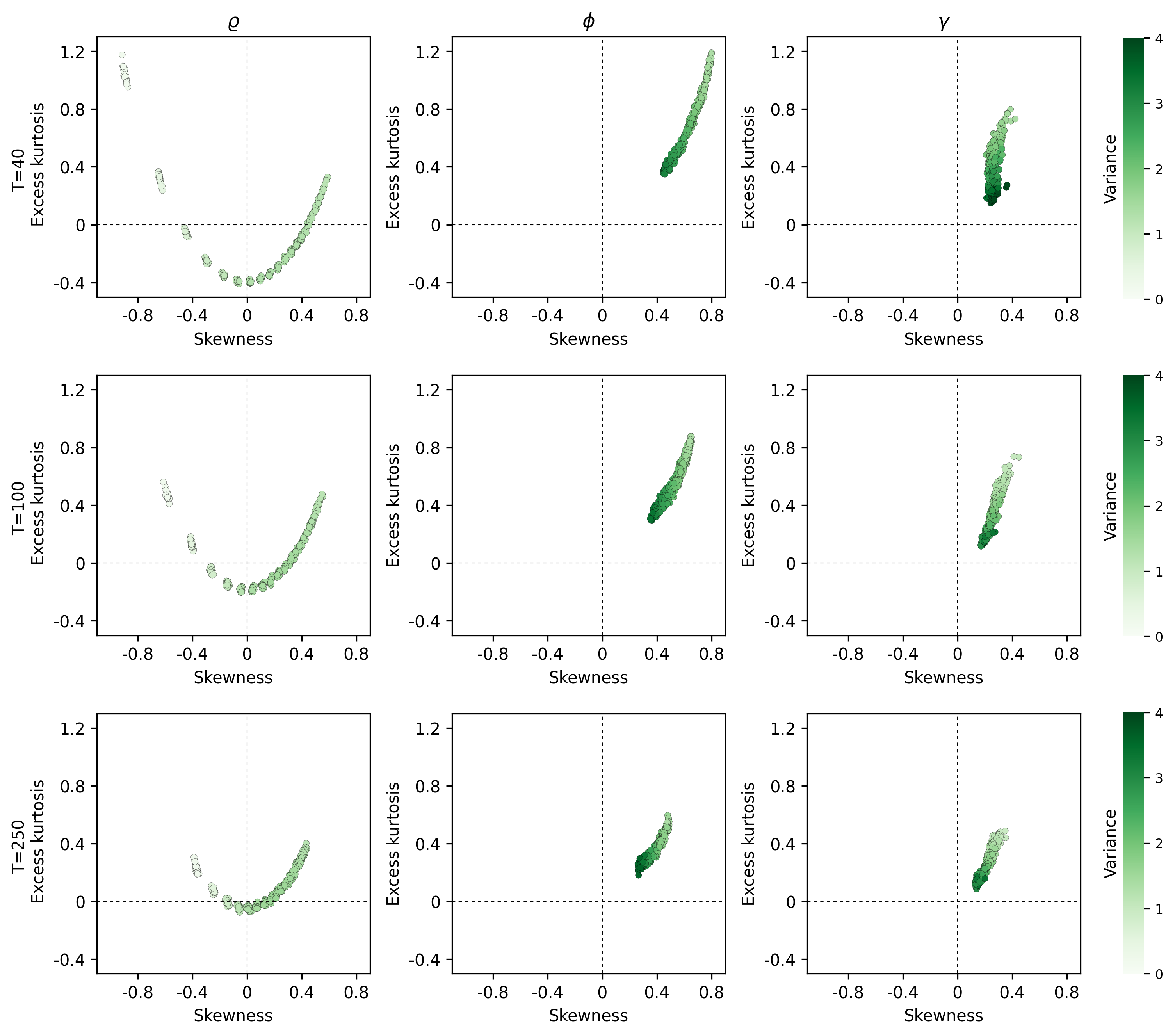}\caption{\small Properties when $\hat{C}$ is computed from $X_t\sim iid \operatorname{IG}(1,1)$ with $n=25$. Variance, skewness, and excess kurtosis for the $d=300$ elements of $\hat{\varrho}$, $\hat{\phi}$, and $\hat{\gamma}$ are shown for sample sizes $T=40$, $T=100$, and $T=250$.\label{fig:IGproperties}}
\end{figure}

In this design, the departures from normality of $\hat{\phi}$ and $\hat{\gamma}$ are evident for all three sample sizes. The GFT coordinates nevertheless perform somewhat better: on average, $\hat{\gamma}$ exhibits less skewness and excess kurtosis than $\hat{\phi}$.

\begin{figure}[!htbp]
\centering{}\includegraphics[width=0.95\textwidth]{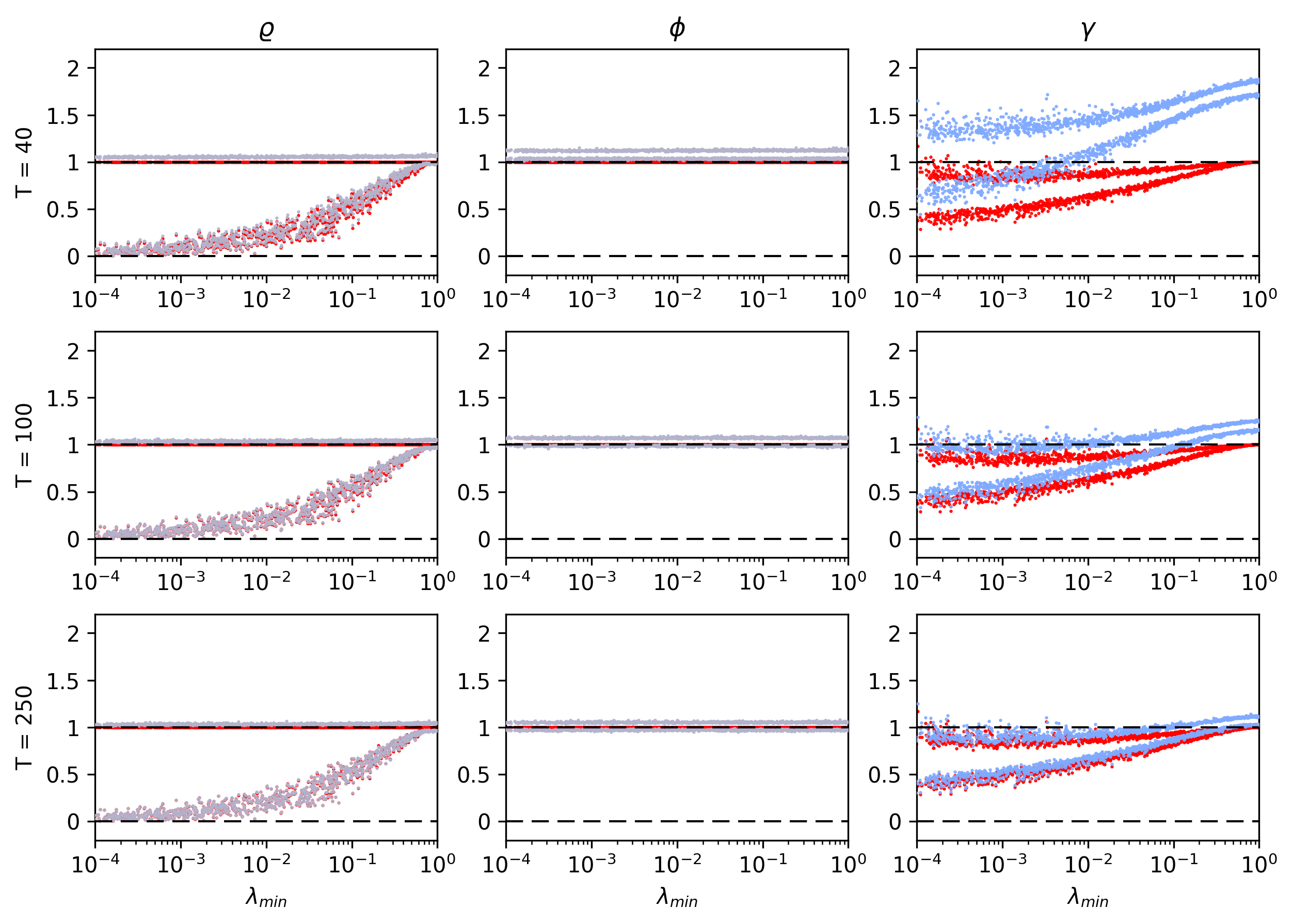}\caption{\small Smallest and largest finite-sample variances of $\sqrt{T}(\hat{\varrho}-\varrho)$, $\sqrt{T}(\hat{\phi}-\phi)$, and $\sqrt{T}(\hat{\gamma}-\gamma)$ for $n=25$ ($d=300$), plotted against $\lambda_{\min}(C)$ for 1,000 random correlation matrices. Columns correspond to $\varrho$, $\phi$, and $\gamma$. Blue dots are finite-sample variances, and red dots are the corresponding asymptotic variances.\label{fig:RandomDesignVariances}}
\end{figure}

\begin{figure}[!htbp]
\centering{}\includegraphics[width=0.95\textwidth]{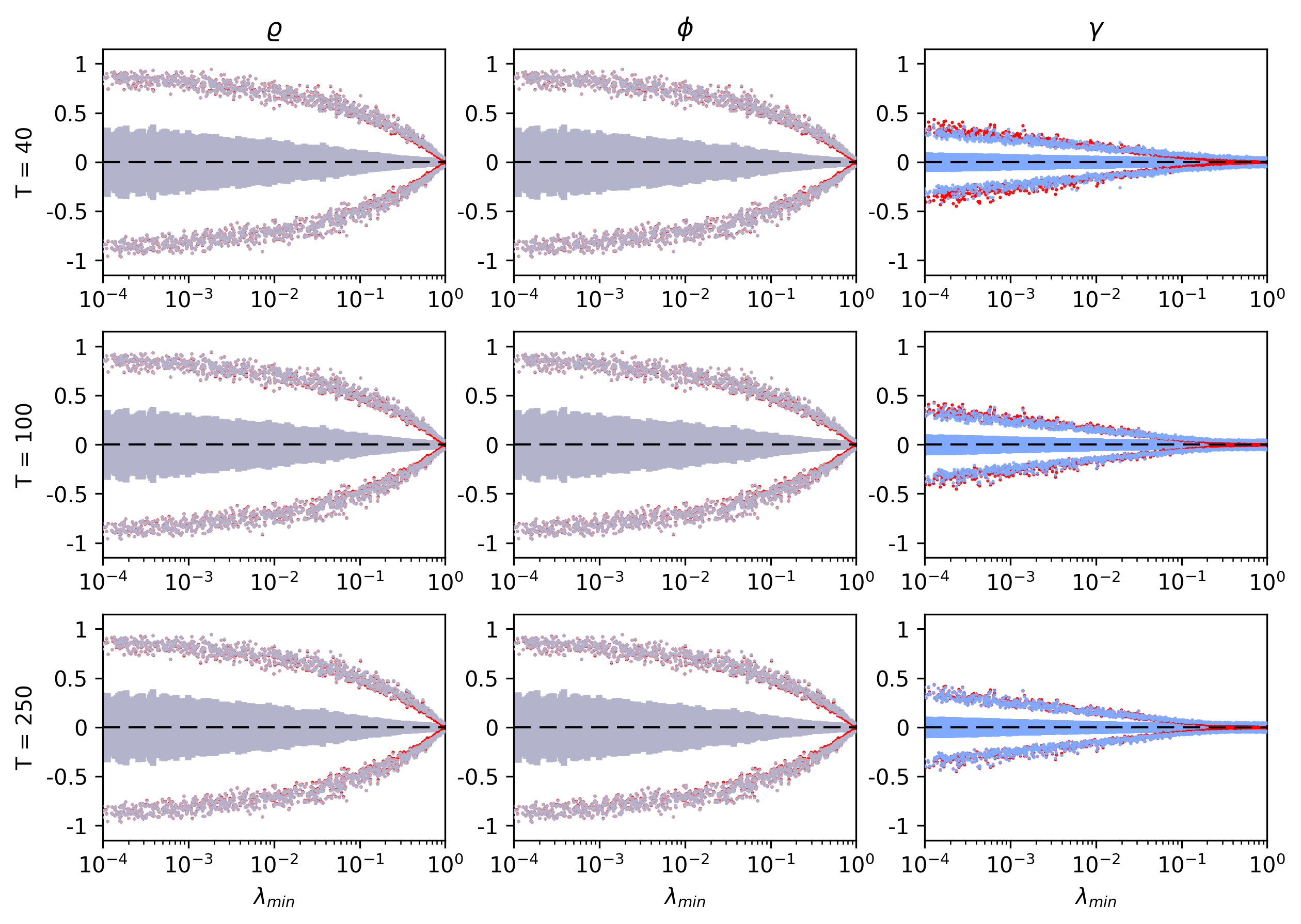}\caption{\small Range of finite-sample correlations between elements of $\hat{\varrho}$, $\hat{\phi}$, and $\hat{\gamma}$ for $n=25$ ($d=300$), plotted against $\lambda_{\min}(C)$ for 1,000 random correlation matrices. Columns correspond to $\varrho$, $\phi$, and $\gamma$, and rows correspond to sample sizes $T=40$, $100$, and $250$. Blue dots show the smallest and largest finite-sample correlations, red dots show the corresponding asymptotic values, and shaded regions contain 80\% of the 44,850 pairwise correlations in each design.\label{fig:R-CorrelationRange}}
\end{figure}
\begin{figure}[!htbp]
\centering{}\includegraphics[width=0.95\textwidth]{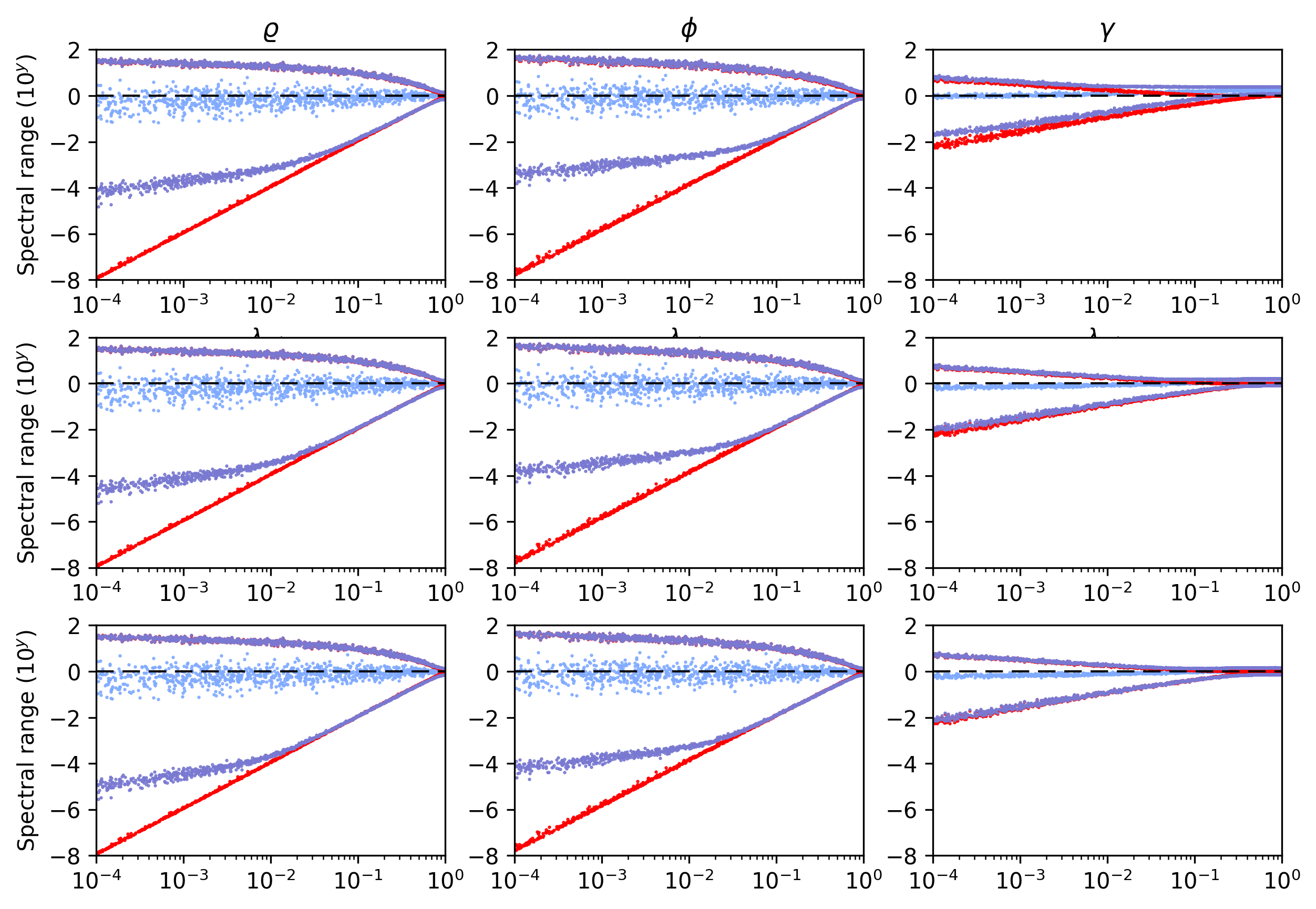}\caption{\small Variance stabilization for linear combinations of $\hat{\varrho}$, $\hat{\phi}$, and $\hat{\gamma}$, based on 1,000 random correlation matrices with $n=25$ and sample sizes $T=40$, $100$, and $250$. Dark blue dots show the smallest and largest eigenvalues of $V_{\varrho,T}(C)$, $V_{\phi,T}(C)$, and $V_{\gamma,T}(C)$, reported in $\log_{10}$ units; red dots show the corresponding eigenvalues of the asymptotic covariance matrices. Light blue dots show the variances of the standardized average elements, $\sqrt{d}\overline{\hat{\varrho}}=d^{-1/2}\sum_{i=1}^d\hat{\varrho}_i$, $\sqrt{d}\overline{\hat{\phi}}$, and $\sqrt{d}\overline{\hat{\gamma}}$.\label{fig:Veigenvalues}}
\end{figure}
Figures~\ref{fig:RandomDesignVariances}, \ref{fig:R-CorrelationRange}, and \ref{fig:Veigenvalues} extend the three rows of Figure~\ref{fig:VarCorrSpectrum} to sample sizes $T=40$, $100$, and $250$. They show, respectively, marginal variances, pairwise correlations between elements, and covariance-matrix spectra. The patterns are stable across sample sizes: the GFT coordinates exhibit substantially weaker dependence and a more stable covariance spectrum than the sample correlations and the element-wise Fisher transformed correlations.

\section{Additional Empirical Results}

Figure~\ref{Fig:Dataset} in the main text shows the empirical correlation structures used in the resampling exercises: for each data set, the left panel shows the sample correlation matrix, the middle panel the corresponding log-transformed correlation matrix, and the right panel plots the element-wise Fisher transformed correlations, $\hat{\phi}_i$, against the corresponding GFT coordinates, $\hat{\gamma}_i$. These plots illustrate how the matrix logarithm reshapes the empirical correlation structure and how closely the GFT coordinates are related to the conventional Fisher transformed correlations.

The three data sets exhibit different forms of dependence. The macroeconomic variables display an approximate block structure, reflecting groups of variables with similar cyclical or trend behavior. The industry portfolios are more uniformly positively correlated, consistent with the presence of a strong market-wide component in daily equity returns. The high-frequency return data also exhibit substantial dependence, but with a correlation structure that is less uniform across assets. In all three cases, the scatter plots show a close relationship between $\hat{\phi}_i$ and $\hat{\gamma}_i$, but the relationship is not exactly linear. This confirms that the GFT coordinates retain a strong element-wise connection to the corresponding correlations while incorporating information from the full correlation matrix through the matrix logarithm.

\subsection{Results for 24 Macroeconomic Variables}

We extract a subset of 24 variables from the FRED-MD database of \citet{McCrackenNg:2016}. This design provides an empirical correlation structure that differs substantially from the equity-return designs considered in the main text. In particular, the macroeconomic variables exhibit a more pronounced block structure, reflecting groups of variables with similar cyclical behavior and common low-frequency components. Table~\ref{tab:FRED-rho} reports the sample correlation matrix together with its log-transformed counterpart.

We use this empirical distribution to assess whether the finite-sample properties of the GFT persist in a macroeconomic setting with non-Gaussian dependence and heterogeneous marginal behavior. Figure~\ref{fig:FRED-VarSkewKurt} reports the variance, skewness, and excess kurtosis of the elements of $\sqrt{T}(\hat{\varrho}-\varrho)$, $\sqrt{T}(\hat{\phi}-\phi)$, and $\sqrt{T}(\hat{\gamma}-\gamma)$ based on empirical resampling. The results are broadly consistent with the evidence in the main text: the marginal distributions of the transformed coordinates are not exactly Gaussian, but the GFT coordinates continue to display more stable behavior than the sample correlations.

\begin{table}[H]
\begin{centering}
\includegraphics[width=0.95\textwidth]{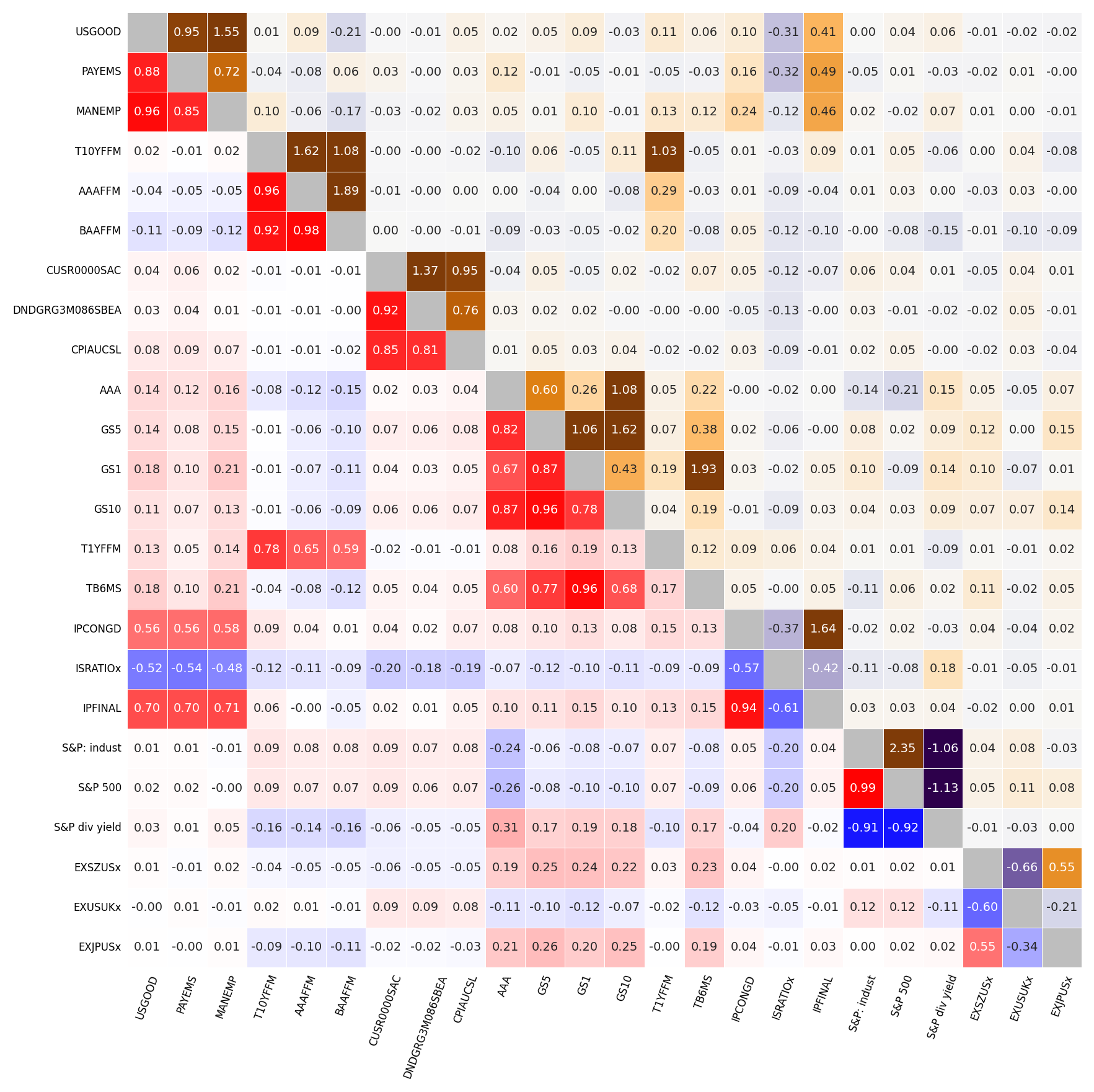}
\par\end{centering}
\caption{{\small Sample correlation matrix and log-transformed sample correlation matrix estimated for 24 macro variables}\label{tab:FRED-rho}}
\end{table}

\begin{figure}[!htbp]
\centering{}\includegraphics[width=0.95\textwidth]{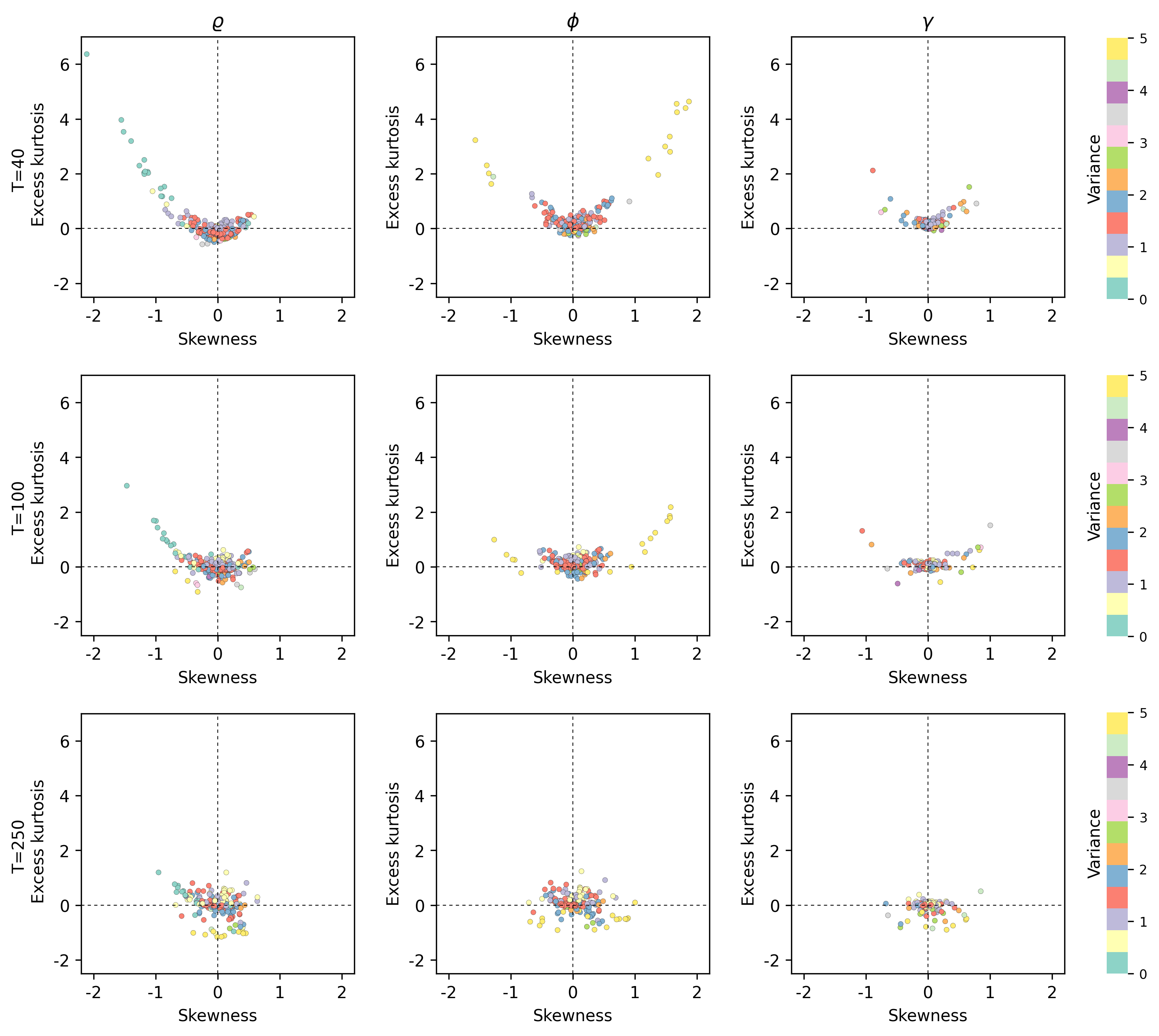}\caption{\small FRED-MD empirical resampling results for 24 macroeconomic variables. The panels report finite-sample variance, skewness, and excess kurtosis for the $d=276$ elements of $\sqrt{T}(\hat{\varrho}-\varrho)$, $\sqrt{T}(\hat{\phi}-\phi)$, and $\sqrt{T}(\hat{\gamma}-\gamma)$. The color of each point indicates the corresponding variance. Results are based on 100,000 sample correlation matrices computed from samples drawn with replacement from the empirical distribution of the 24-dimensional macroeconomic data.\label{fig:FRED-VarSkewKurt}}
\end{figure}

\subsection{Additional Results for Fama-French Industry Portfolios}

This subsection provides additional results for the 30 Fama-French industry portfolios analyzed in Section~\ref{sec:FamaFrench}. The empirical sample contains 503 daily return vectors from January 2, 2018 to December 31, 2019, and the dimension is $n=30$, so $d=435$. We consider two simulation designs. In the first design, observations are drawn from a Gaussian distribution with correlation matrix equal to the empirical correlation matrix. This isolates the effect of the empirical correlation structure. In the second design, observations are drawn with replacement from the empirical distribution of return vectors. This preserves both the empirical correlation structure and the non-Gaussian features of the observed returns.

Figures~\ref{fig:ff-VarSkewKurtGaussian} and \ref{fig:ff-VarSkewKurtEmpirical} report variance, skewness, and excess kurtosis for the elements of $\sqrt{T}(\hat{\varrho}-\varrho)$, $\sqrt{T}(\hat{\phi}-\phi)$, and $\sqrt{T}(\hat{\gamma}-\gamma)$. The Gaussian design confirms that the marginal distributions of $\hat{\phi}$ and $\hat{\gamma}$ are close to normal when the data are Gaussian, even for moderate sample sizes. The empirical resampling design shows slower convergence to normality, reflecting the non-Gaussian features of daily industry returns. Nevertheless, the GFT coordinates continue to exhibit more stable finite-sample behavior than the sample correlations.

Figure~\ref{fig:ff-finite-sample-distr} summarizes the same comparison in terms of the distributions of variances, pairwise correlations, and skewness across elements. The contrast between $\hat{\varrho}$ and $\hat{\gamma}$ is especially clear for the dependence measures: the correlations between elements of $\hat{\gamma}$ are much more concentrated around zero, whereas the corresponding correlations between elements of $\hat{\varrho}$ are substantially more dispersed. This reinforces the main finding that the GFT coordinates provide a more nearly orthogonal representation of correlation estimation errors.

\begin{figure}[!htbp]
\centering{}\includegraphics[width=0.95\textwidth]{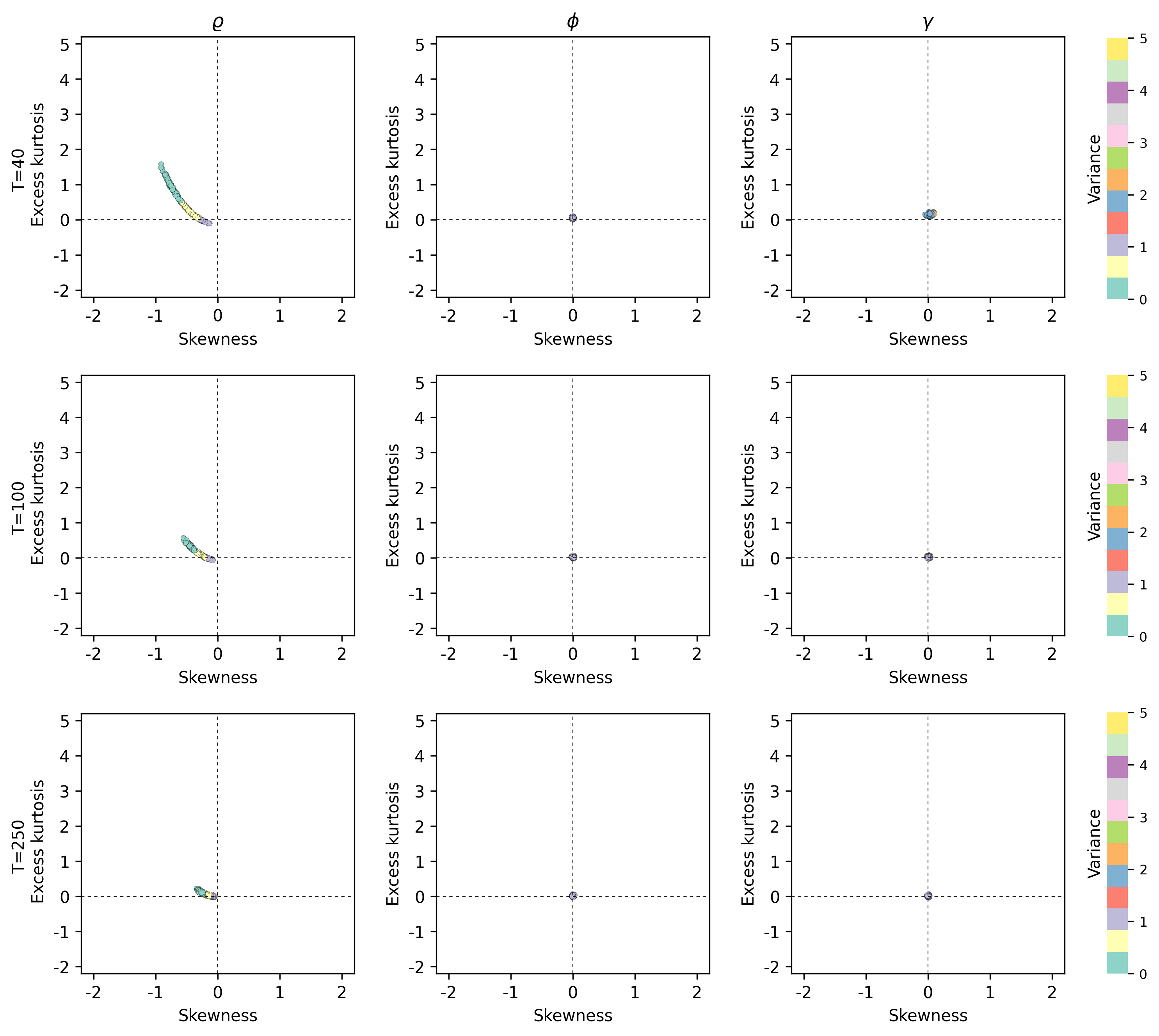}\caption{\small Gaussian simulation results based on the empirical correlation matrix of the 30 Fama-French industry portfolios. The panels report finite-sample variance, skewness, and excess kurtosis for the $d=435$ elements of $\sqrt{T}(\hat{\varrho}-\varrho)$, $\sqrt{T}(\hat{\phi}-\phi)$, and $\sqrt{T}(\hat{\gamma}-\gamma)$. The color of each point indicates the corresponding variance. Results are based on 100,000 sample correlation matrices computed from Gaussian samples with correlation matrix equal to the sample correlation matrix estimated from daily industry returns over the period January 2, 2018, to December 31, 2019.\label{fig:ff-VarSkewKurtGaussian}}
\end{figure}

\begin{figure}[!htbp]
\centering{}\includegraphics[width=0.95\textwidth]{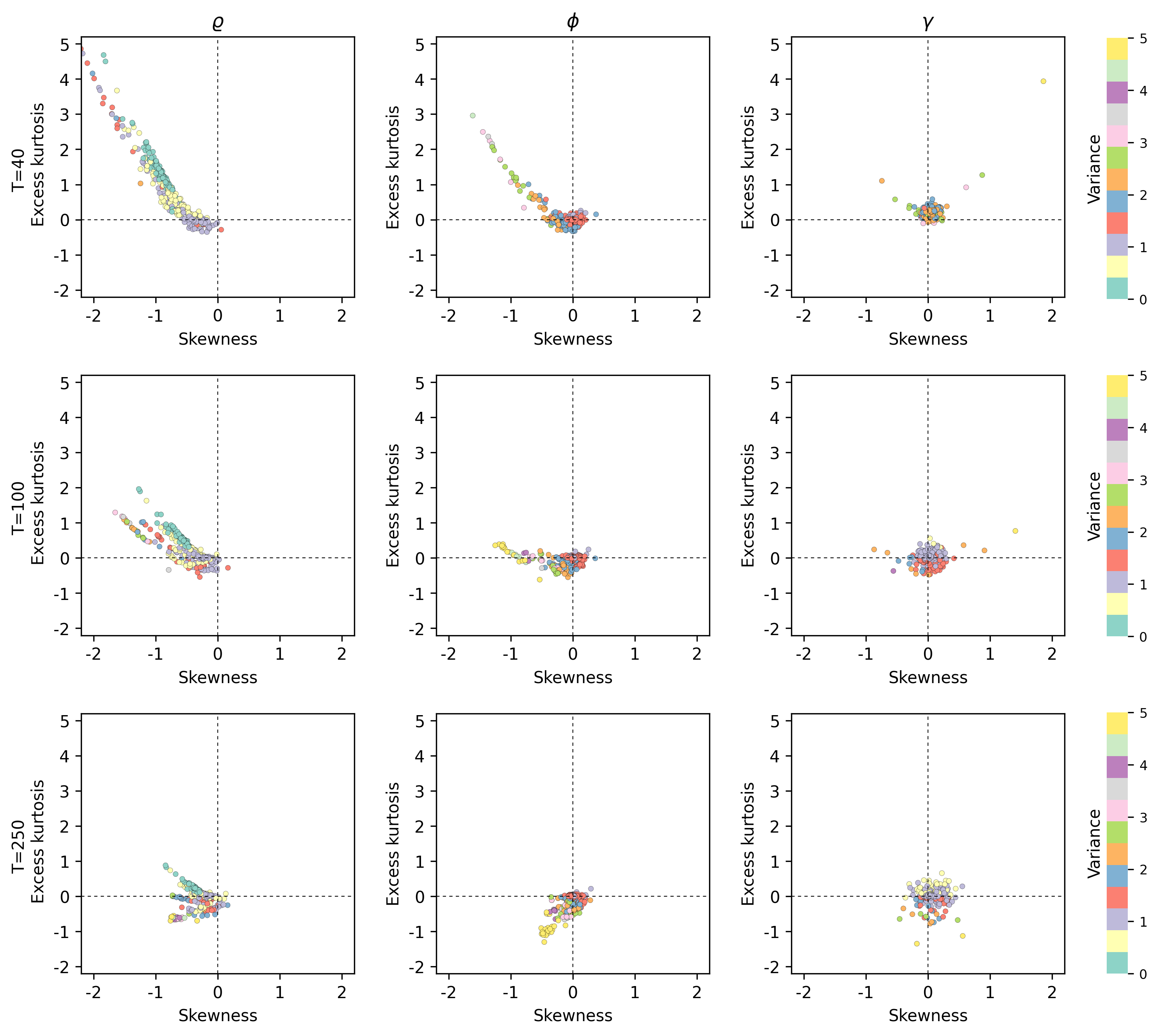}\caption{\small Empirical resampling results based on daily returns for the 30 Fama-French industry portfolios. The panels report finite-sample variance, skewness, and excess kurtosis for the $d=435$ elements of $\sqrt{T}(\hat{\varrho}-\varrho)$, $\sqrt{T}(\hat{\phi}-\phi)$, and $\sqrt{T}(\hat{\gamma}-\gamma)$. The color of each point indicates the corresponding variance. Results are based on 100,000 sample correlation matrices computed from samples drawn with replacement from the 503 daily return vectors over the period January 2, 2018, to December 31, 2019.\label{fig:ff-VarSkewKurtEmpirical}}
\end{figure}

\begin{figure}[!htbp]
\centering{}\includegraphics[width=0.95\textwidth]{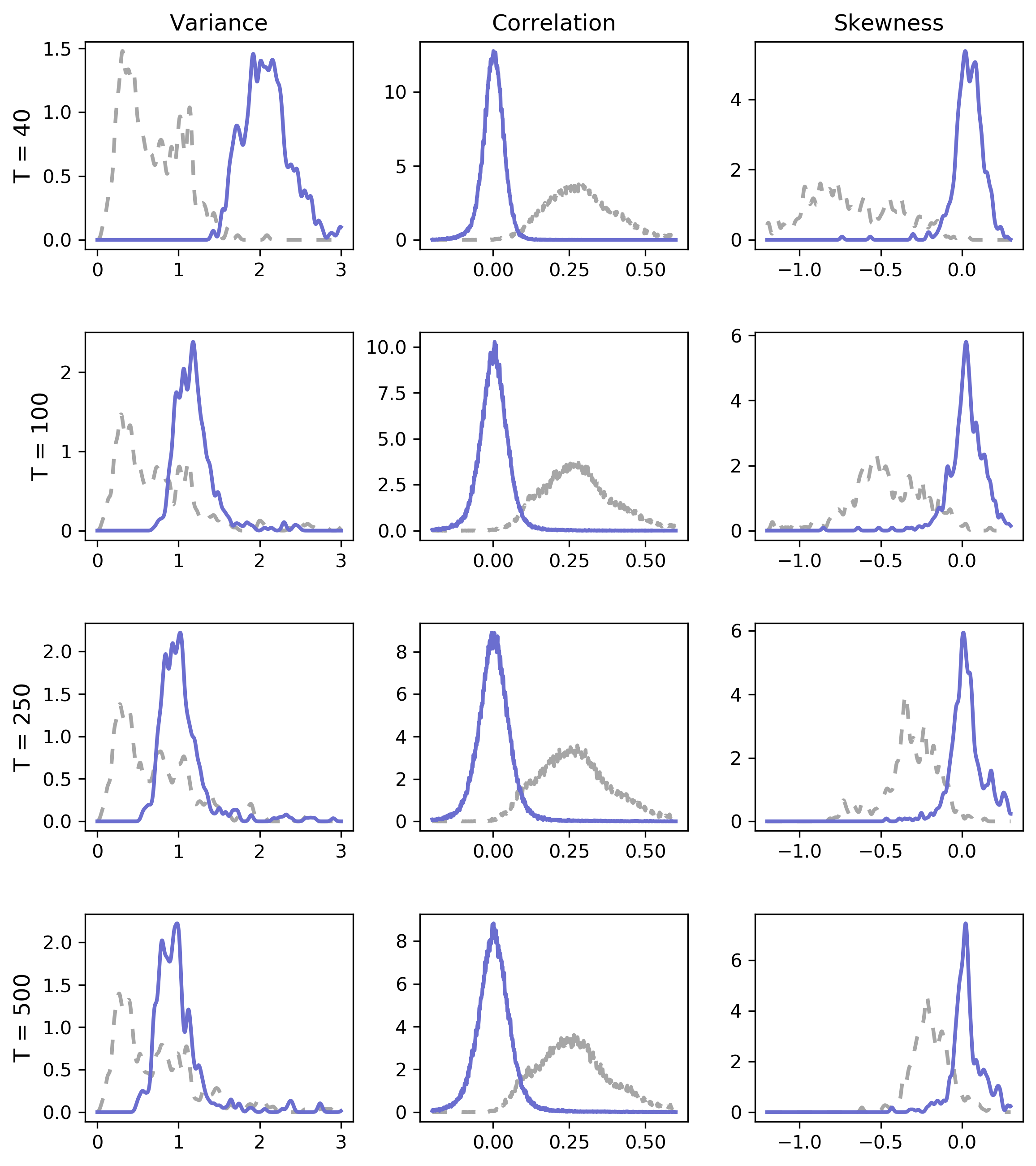}\caption{\small Empirical resampling results for daily returns on the 30 Fama-French industry portfolios. The left panels show density plots of the finite-sample variances of the 435 elements of $\sqrt{T}(\hat{\varrho}-\varrho)$ (grey dashed lines) and $\sqrt{T}(\hat{\gamma}-\gamma)$ (blue solid lines), for sample sizes $T=40$, $100$, $250$, and $500$. The middle panels show the corresponding distributions of pairwise correlations between elements, and the right panels show the distributions of skewness across elements. Results are based on 10,000 sample correlation matrices computed from samples drawn with replacement from the 503 daily return vectors over the period January 2, 2018, to December 31, 2019.\label{fig:ff-finite-sample-distr}}
\end{figure}

\subsection{Results for 5-Minute High-Frequency Returns}

This subsection provides additional results for the high-frequency return data used in Section~\ref{sec:RealizedCorr}. The data consist of five-minute intraday returns for 21 assets, so that each realized correlation matrix has $d=210$ distinct off-diagonal elements. Table~\ref{tab:TAQ-gamma} reports the sample correlation matrix and the corresponding log-transformed correlation matrix for this data set.

The high-frequency design differs from the daily industry-return and macroeconomic designs in two important respects. First, the sample correlation matrices are computed from intraday observations within each trading day, so the sample size for each realized correlation matrix is relatively small. Second, the latent correlation matrix is likely to vary over time and may also vary within the trading day. These features make the empirical setting more challenging than the controlled simulation designs, but also more relevant for applications involving realized covariance and correlation matrices.

Figure~\ref{fig:TAQ-VarSkewKurt} reports finite-sample variance, skewness, and excess kurtosis for the elements of $\sqrt{T}(\hat{\varrho}-\varrho)$, $\sqrt{T}(\hat{\phi}-\phi)$, and $\sqrt{T}(\hat{\gamma}-\gamma)$ based on resampling from the empirical distribution of five-minute returns. The results are consistent with the evidence in the main text: the transformed coordinates have better marginal behavior than the raw realized correlations, and the GFT coordinates remain comparatively stable.

Figure~\ref{fig:TAQ-gamma-series} illustrates the time-series behavior of a representative GFT coordinate, $\hat{\gamma}_{i,t}$, together with its filtered counterpart, $\gamma_{i,t}^f$. This figure highlights the interpretation used in the main empirical analysis: the observed realized GFT coordinate is treated as a noisy measurement of a smoother latent correlation component.

\begin{table}[H]
\begin{centering}
\includegraphics[width=0.95\textwidth]{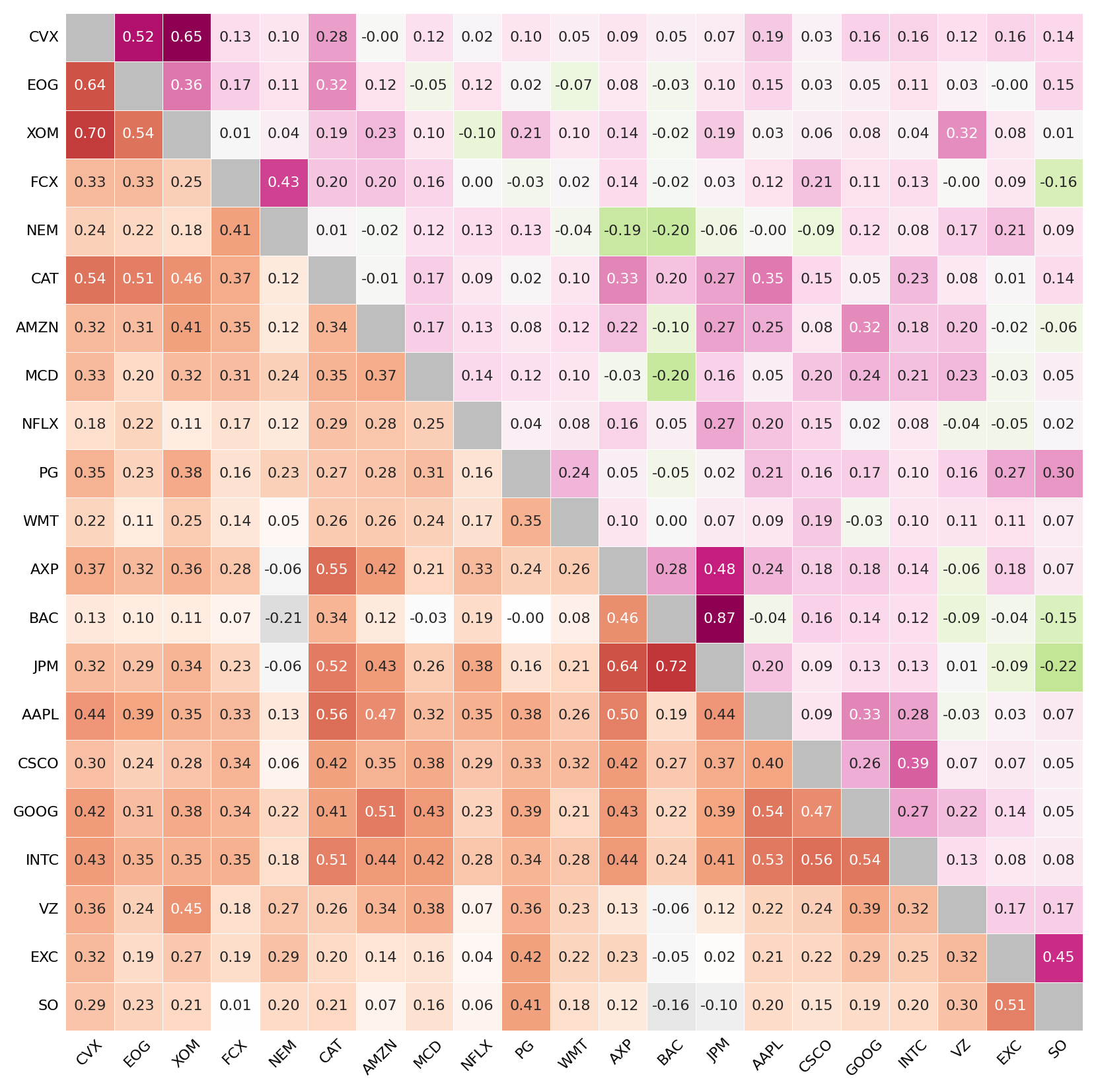}
\par\end{centering}
\caption{\small Sample correlation matrix and log-transformed sample correlation matrix for the 21-asset five-minute high-frequency return data.\label{tab:TAQ-gamma}}
\end{table}

\begin{figure}[!htbp]
\centering{}\includegraphics[width=0.95\textwidth]{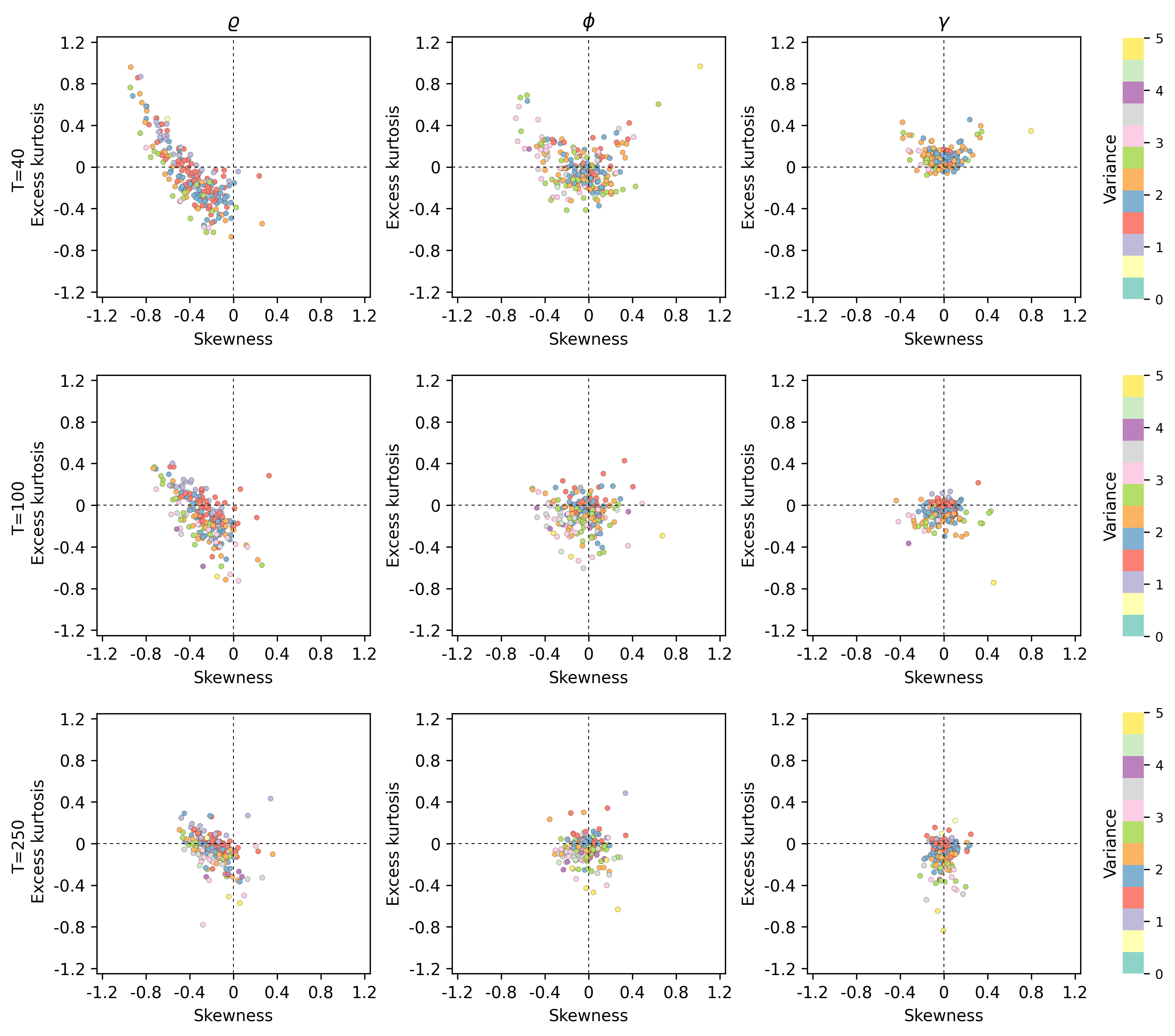}\caption{\small Empirical resampling results based on five-minute high-frequency returns for 21 assets. The panels report finite-sample variance, skewness, and excess kurtosis for the $d=210$ elements of $\sqrt{T}(\hat{\varrho}-\varrho)$, $\sqrt{T}(\hat{\phi}-\phi)$, and $\sqrt{T}(\hat{\gamma}-\gamma)$. The color of each point indicates the corresponding variance. Results are based on 100,000 sample correlation matrices computed from samples drawn with replacement from the empirical distribution of five-minute return vectors.\label{fig:TAQ-VarSkewKurt}}
\end{figure}

\begin{center}
\begin{figure}[!htbp]
\begin{centering}
\includegraphics[width=0.95\textwidth]{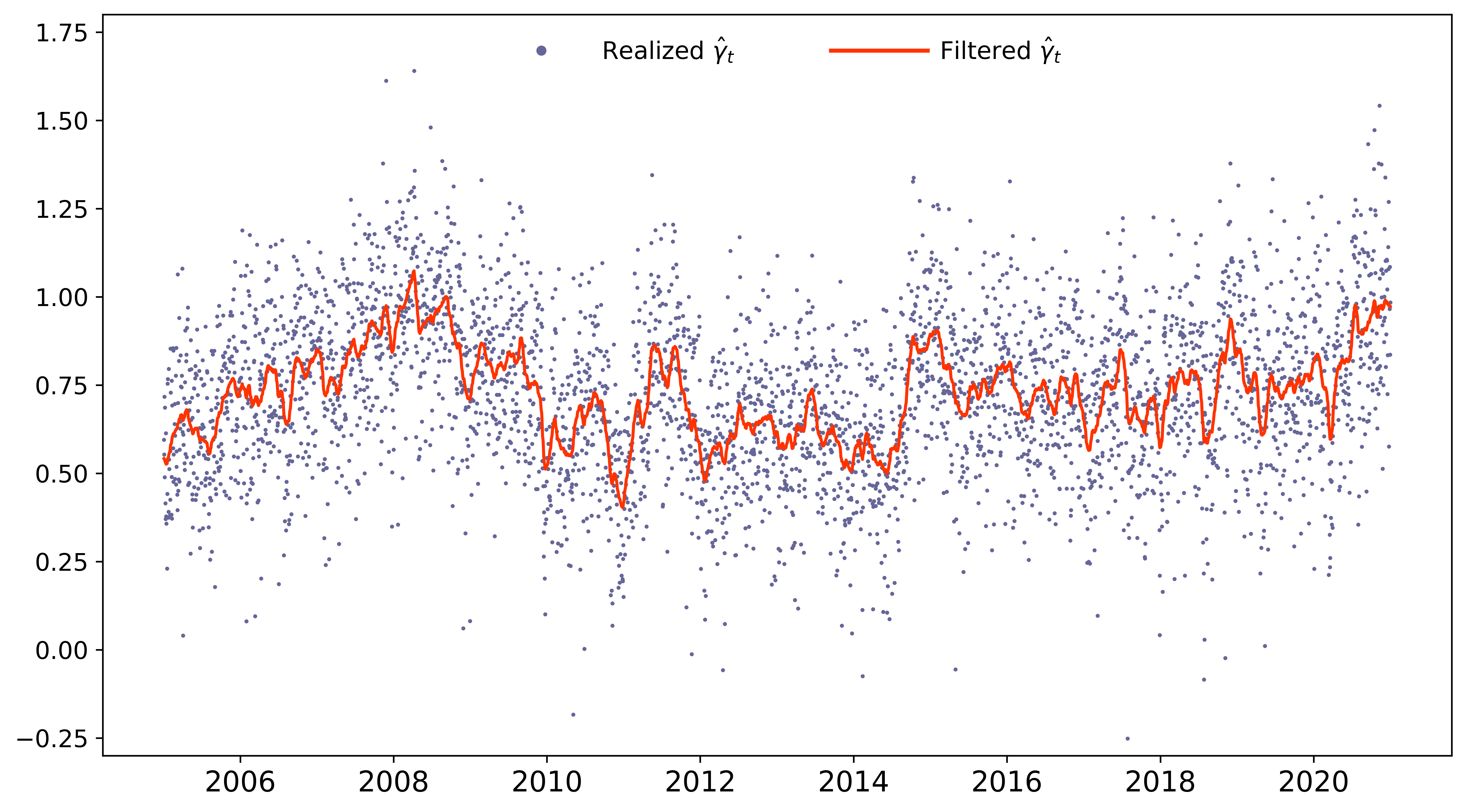}
\par\end{centering}
\caption{\small Representative time series of a realized GFT coordinate, $\hat{\gamma}_{i,t}$, and its filtered counterpart, $\gamma_{i,t}^{f}$.\label{fig:TAQ-gamma-series}}
\end{figure}
\par\end{center}

\section{Additional Evidence on Near-Singularity\label{sec:app_Simulations}}

The results above show that the smallest eigenvalue, $\lambda_{\min}(C)$, is an important determinant of the finite-sample behavior of the GFT coordinates. We provide two additional pieces of evidence. First, we report more detailed results for the random correlation design in Figure~\ref{fig:RandomDesignVariancesGammaAppendix}, focusing on the finite-sample variances of $\sqrt{T}(\hat{\gamma}-\gamma)$ across dimensions and sample sizes. Second, we complement the random-design evidence with Toeplitz correlation matrices of the form in (\ref{eq:CorrelationMatrixToeplitz}), where near-singularity is generated by varying the single parameter $\rho$. The Toeplitz design provides a more structured way to study the role of near-singularity and helps separate the effect of $\lambda_{\min}(C)$ from the broader geometry of the correlation matrix.

\begin{figure}[!htbp]
\centering{}\includegraphics[width=0.95\textwidth]{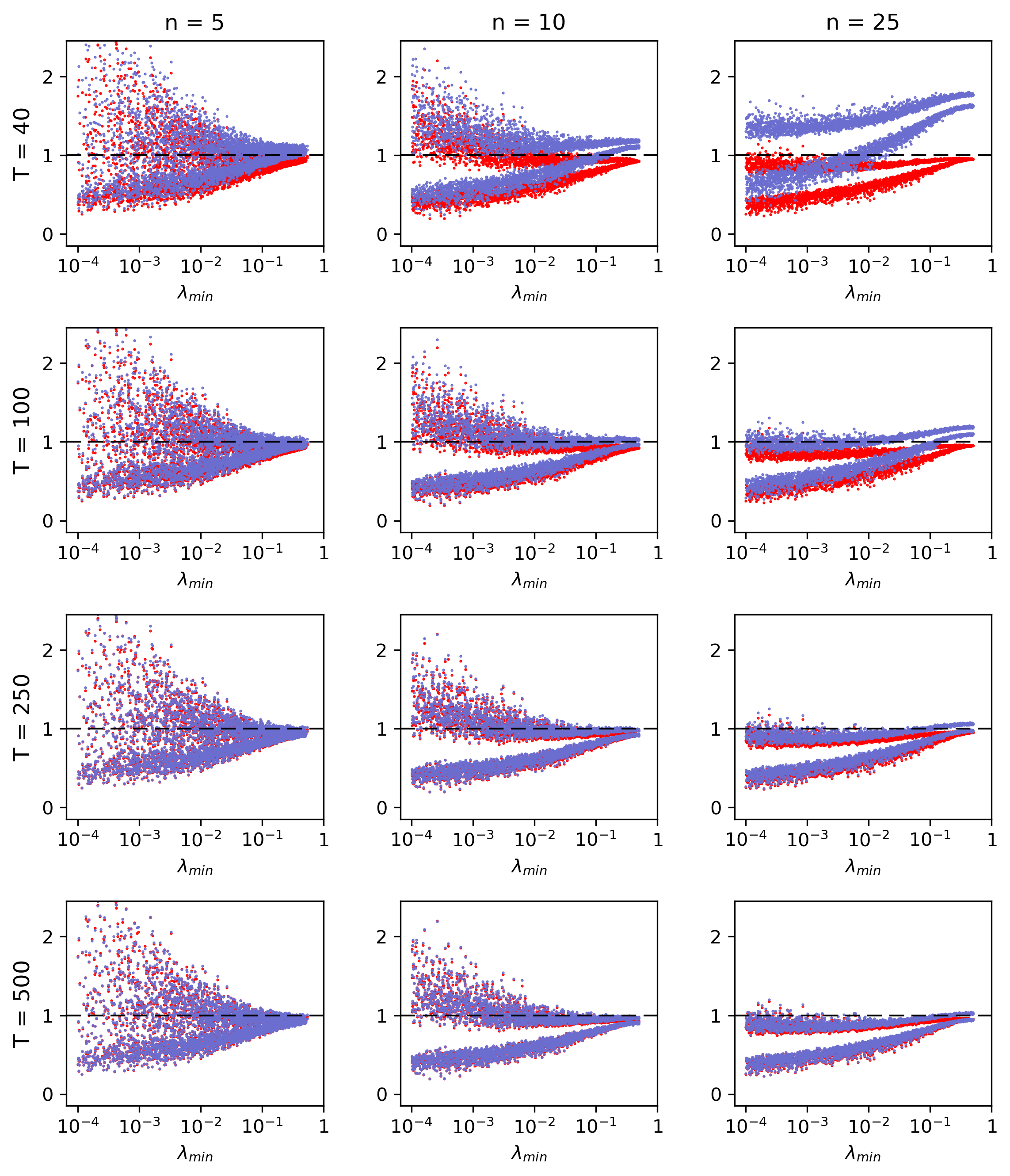}\caption{\small Smallest and largest finite-sample variances of $\sqrt{T}(\hat{\gamma}-\gamma)$ for the random correlation design, plotted against $\lambda_{\min}(C)$. Results are shown for dimensions $n=5$, $10$, and $25$ and sample sizes $T=40$, $100$, $250$, and $500$. Blue dots are finite-sample variances estimated from 10,000 independent estimates of $\hat{C}$ for each of 1,000 random correlation matrices; red dots are the corresponding asymptotic variances from $V_{\gamma}(C)$.\label{fig:RandomDesignVariancesGammaAppendix}}
\end{figure}

\subsection{Toeplitz Correlation Matrices}

Figure~\ref{fig:min_max_var_corr_toeplitz} compares the finite-sample variances and pairwise correlations of $\sqrt{T}(\hat{\gamma}-\gamma)$ with their asymptotic counterparts for Toeplitz correlation matrices. In this design, variation in $\lambda_{\min}(C)$ is generated by varying the single parameter $\rho$. The Toeplitz matrices are more restrictive than the random correlation matrices used above, and this helps explain some differences between the two designs. In particular, the dispersion between the smallest and largest finite-sample variances and correlations is less pronounced for Toeplitz matrices. Thus, although $\lambda_{\min}(C)$ is an important determinant of finite-sample behavior, the structure of the correlation matrix also matters.

\begin{figure}[!htbp]
\begin{centering}
\includegraphics[width=0.95\textwidth]{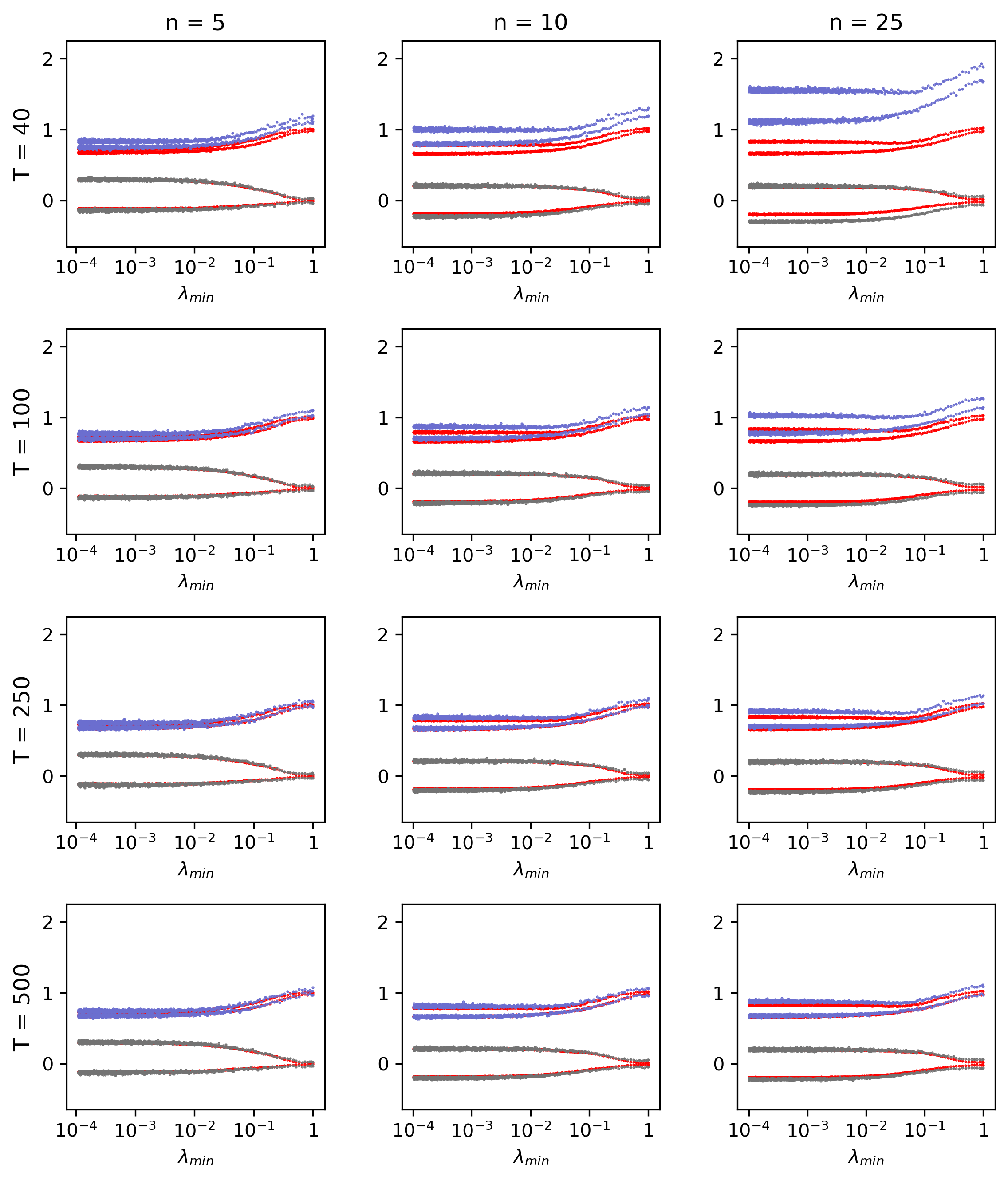}
\par\end{centering}
\caption{\small Finite-sample and asymptotic variances and pairwise correlations for the elements of $\sqrt{T}(\hat{\gamma}-\gamma)$ under Toeplitz correlation matrices. Columns correspond to $n=5$, $10$, and $25$, and rows correspond to sample sizes $T=40$, $100$, $250$, and $500$. In each panel, blue dotted curves show the smallest and largest finite-sample variances, while grey dotted curves show the smallest and largest finite-sample correlations, plotted against $\lambda_{\min}(C)$. Red dots show the corresponding asymptotic minima and maxima for the same Toeplitz correlation matrices. Since the asymptotic quantities do not depend on $T$, the red dots are identical across rows within each column.\label{fig:min_max_var_corr_toeplitz}}
\end{figure}

\subsection{Correlations Between Elements}

We next provide additional evidence on the dependence between elements of the coordinate vectors. Since $\phi$ is obtained from $\varrho$ by an element-wise transformation, the delta method implies that the asymptotic correlation matrices of $\hat{\varrho}$ and $\hat{\phi}$ are identical, $R_{\varrho}(C)=R_{\phi}(C)$. In the simulations, the corresponding finite-sample correlation matrices, $R_{\varrho,T}(C)$ and $R_{\phi,T}(C)$, are also virtually indistinguishable. We therefore report results for $R_{\varrho,T}(C)$ in Figure~\ref{fig:min_max_corr_corr_random}; the results for the element-wise Fisher transformed correlations are essentially the same.
\begin{figure}[!htbp]
\centering{}\includegraphics[width=0.95\textwidth]{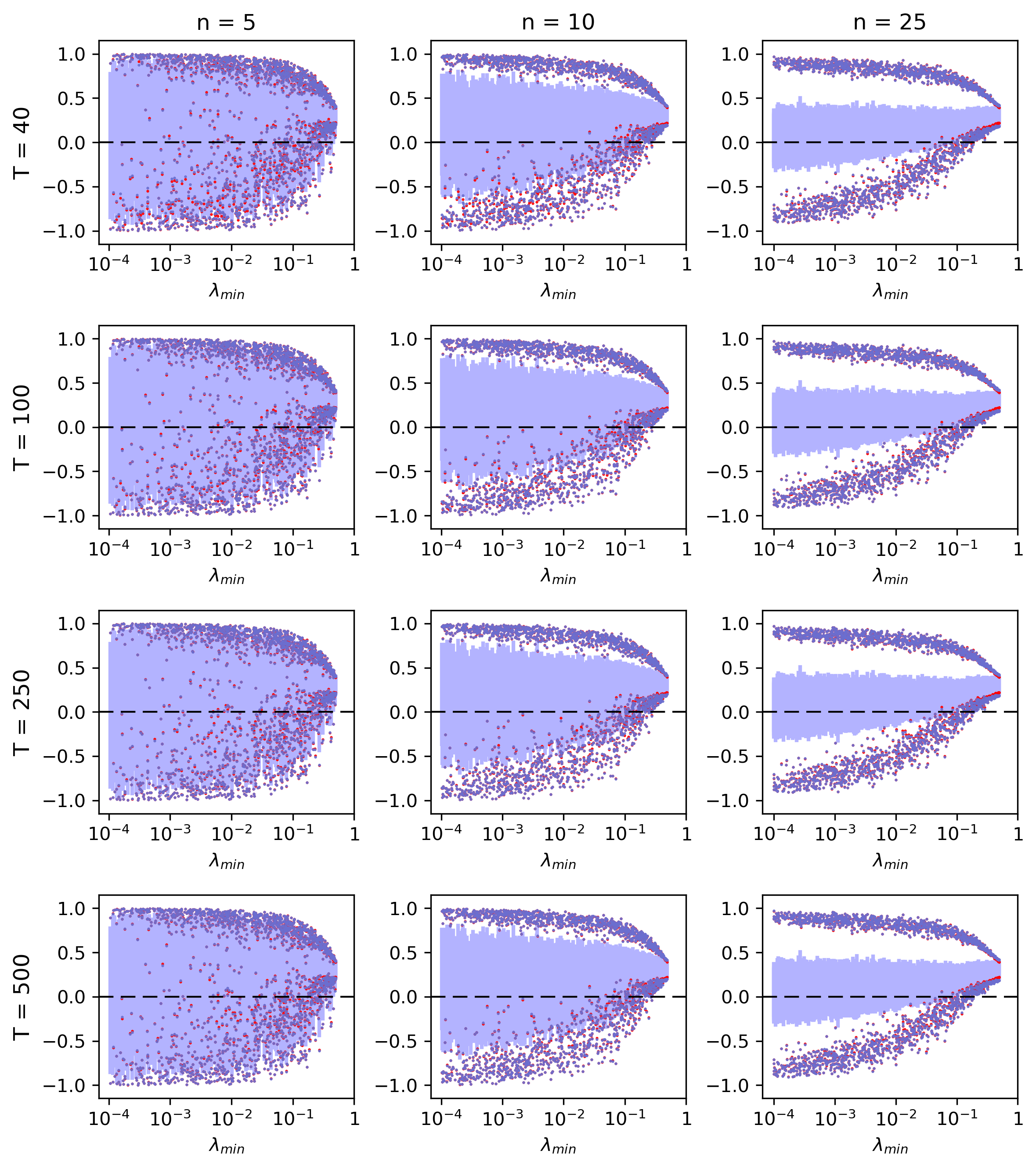}\caption{\small Finite-sample correlations between elements of $\hat{\varrho}$ under the random correlation design, plotted against $\lambda_{\min}(C)$ for 1,000 distinct correlation matrices. Blue dots show the smallest and largest off-diagonal elements of $R_{\varrho,T}(C)$, and red dots show the corresponding asymptotic values from $R_{\varrho}(C)$. The shaded region covers the 10\%-to-90\% interquantile range of the off-diagonal elements of $R_{\varrho,T}(C)$. Results are shown for dimensions $n=5$, $10$, and $25$ and sample sizes $T=40$, $100$, $250$, and $500$. For each correlation matrix, $R_{\varrho,T}(C)$ is estimated from 10,000 samples of $T$ independent Gaussian vectors, $X_t\sim iid N(0,C_{(j)})$.\label{fig:min_max_corr_corr_random}}
\end{figure}

For $n=5$, $10$, and $25$, the coordinate vectors have dimensions $d=10$, $45$, and $300$, respectively. Thus, the corresponding correlation matrices contain 45, 990, and 44,850 distinct off-diagonal elements.

Figure~\ref{fig:min_max_corr_corr_random} reports results for 1,000 randomly generated correlation matrices. For each matrix, we compute the asymptotic correlation matrix $R_{\varrho}(C)$ from the expression for $V_{\varrho}(C)$. The finite-sample counterpart, $R_{\varrho,T}(C)$, is estimated from $N=10,000$ independent estimates of $\hat{C}$, each based on $T$ independent draws from $N(0,C)$. The blue dots show the minimum and maximum off-diagonal elements of $R_{\varrho,T}(C)$, while the shaded region gives the 10\%-to-90\% interquantile range. The corresponding asymptotic minima and maxima from $R_{\varrho}(C)$ are shown with red dots. The close alignment of the blue and red points shows that the dependence between elements of $\hat{\varrho}$ is not primarily a finite-sample artifact.

Figure~\ref{fig:min_max_corr_corr_random} shows that $\lambda_{\min}(C)$ is a key determinant of the dependence structure. As the smallest eigenvalue approaches zero, the pairwise correlations between the elements of $\hat{\varrho}$ become increasingly dispersed. Even for $\lambda_{\min}(C)=0.1$, many of these correlations are large in absolute value, and the dependence becomes stronger as $\lambda_{\min}(C)\rightarrow0$. This pattern is particularly pronounced for $n=5$, where the shaded regions indicate that a substantial fraction of the correlations in $R_{\varrho,T}(C)$ are close to one in absolute value. Increasing the sample size has little effect on this dependence, suggesting that it is driven mainly by the geometry of the correlation matrix rather than by finite-sample noise. As noted above, the same conclusion applies to the corresponding element-wise Fisher transformed correlations.
\begin{figure}[!htbp]
\centering{}\includegraphics[width=0.95\textwidth]{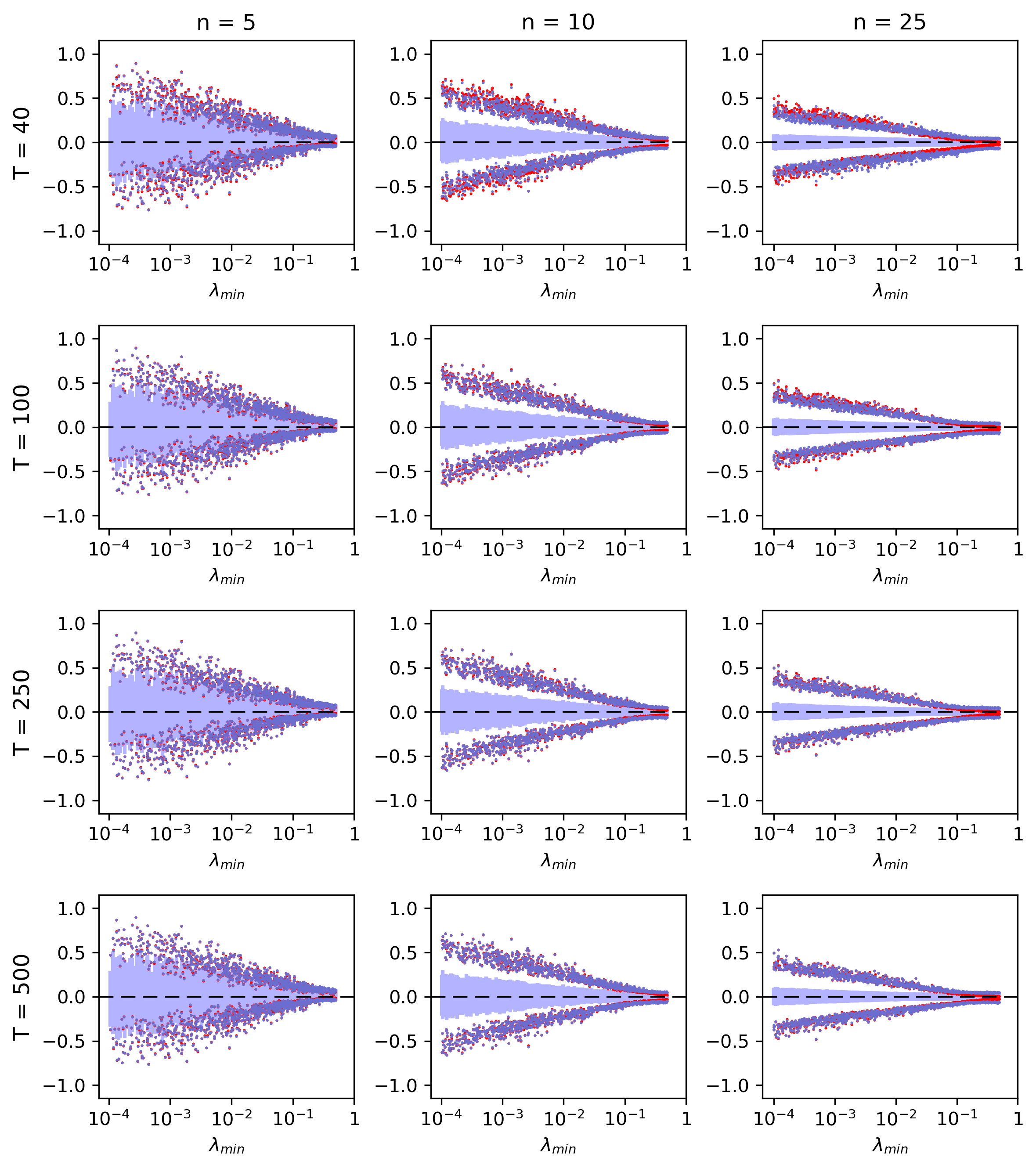}\caption{\small Range of finite-sample and asymptotic correlations between the elements of $\hat{\gamma}$ under the random correlation design, plotted against $\lambda_{\min}(C)$ for $\lambda_{\min}(C)\in[10^{-4},1]$. Results are shown for dimensions $n=5$, $10$, and $25$ and sample sizes $T=40$, $100$, $250$, and $500$. For each dimension, 2,500 random correlation matrices are generated. Blue dots show the smallest and largest off-diagonal elements of $R_{\gamma,T}(C)$, estimated from 10,000 samples of $T$ independent Gaussian vectors, $X_t\sim iid N(0,C_{(j)})$. Red dots show the corresponding asymptotic values from $R_{\gamma}(C)$, and shaded regions contain 80\% of the finite-sample correlations.\label{fig:min_max_corr_gam_random}}
\end{figure}

The analogous results for $\hat{\gamma}$, shown in Figure~\ref{fig:min_max_corr_gam_random}, are markedly different. The pairwise correlations between elements of $\hat{\gamma}$ are much smaller than those for $\hat{\varrho}$ and $\hat{\phi}$ and are centered close to zero. Moreover, the correlations become increasingly concentrated around zero as the dimension $n$ increases. The most extreme correlations become larger in absolute value as $\lambda_{\min}(C)$ approaches zero, but the shaded regions show that most correlations remain close to zero, especially in higher dimensions. As before, increasing the sample size has little effect on the dependence structure.

\end{document}